\documentclass{article}
\usepackage{amsmath, amssymb, amsthm, eucal}
\usepackage{graphicx}
\usepackage[shortlabels]{enumitem}
\usepackage{color}

\usepackage[textheight=630pt,
  textwidth=468pt,
  centering]{geometry}
  
\usepackage{algorithmic}

 \usepackage{pdfsync}

\newtheorem{theorem}{Theorem}

\newtheorem{lemma}{Lemma}

\theoremstyle{definition}
\newtheorem{definition}{Definition}

\newtheorem{remark}{Remark}

\newcommand{\zag}[1]{\left(#1\right)}
\newcommand{\uglate}[1]{\left[#1\right]}
\newcommand{\aps}[1]{\left|#1\right|}
\newcommand{\interval}[1]{\left<#1\right>}
\newcommand{\skup}[1]{\left\{#1\right\}}
\newcommand{\norma}[1]{\left\|#1\right\|}

\newcommand{\R}{\mathbb{R}}
\newcommand{\Rd}{\mathbb{R}^{n_1 \times n_2 \times \cdots \times n_d}}
\newcommand{\C}{\mathbb{C}}
\newcommand{\Cd}{\mathbb{C}^{n_1 \times n_2 \times \cdots \times n_d}}

\newcommand{\E}{\mathbb{E}}
\newcommand{\vjer}[1]{\mathbb{P}\left(#1\right)}
\newcommand{\ocek}[1]{\mathbb{E}#1}

\newcommand{\bm}[1]{\boldsymbol{#1}}
\newcommand{\mbf}[1]{\mathbf{#1}}

\DeclareMathOperator{\rank}{rank}

\DeclareMathOperator{\Vol}{Vol}
\DeclareMathOperator{\f}{f}
\DeclareMathOperator{\vecc}{vec}

\DeclareMathOperator{\level}{level}
\DeclareMathOperator{\tr}{tr}

\DeclareMathOperator{\tensorization}{\mathcal{T}_{vec}}

\begin{document}

\title{Low Rank Tensor Recovery via Iterative Hard Thresholding}

\author{Holger Rauhut\thanks{RWTH Aachen University, Lehrstuhl C f{\"u}r Mathematik (Analysis), Pontdriesch 10, 52062 Aachen,  Germany, rauhut@mathc.rwth-aachen.de}, Reinhold Schneider\thanks{Technische Universit\"{a}t Berlin, Stra\ss e des 17. Juni 136, 10623 Berlin, Germany, schneidr@math.tu-berlin.de} and \v{Z}eljka Stojanac\thanks{RWTH Aachen University, Lehrstuhl C f{\"u}r Mathematik (Analysis), Pontdriesch 10, 52062 Aachen Germany, stojanac@mathc.rwth-aachen.de and University of Bonn, Hausdorff Center for Mathematics, Endenicher Allee 62, 53115 Bonn, Germany}}

%
%
%
\date{February 16, 2016}

%

\maketitle

\begin{abstract}
We study extensions of compressive sensing and low rank matrix recovery (matrix completion) 
to the recovery of low rank tensors of higher order from a small number of linear measurements. While the theoretical understanding
of low rank matrix recovery is already well-developed, only few contributions on the low rank tensor recovery problem are available so far.
In this paper, we introduce versions of the iterative hard thresholding algorithm for several tensor decompositions, namely the higher order singular value
decomposition (HOSVD), the tensor train format (TT), and the general hierarchical Tucker decomposition (HT). We provide a partial convergence result
for these algorithms which is based on a variant of the restricted isometry property of the measurement operator adapted to the tensor decomposition at hand that
induces a corresponding notion of tensor rank. We show that subgaussian measurement ensembles satisfy the tensor 
restricted isometry property with high probability under a certain almost optimal
bound on the number of measurements which depends on the corresponding tensor format. These bounds are extended to partial Fourier maps combined with
random sign flips of the tensor entries. 
Finally, we illustrate the performance of iterative hard thresholding methods for tensor recovery via numerical experiments where we consider
recovery from Gaussian random measurements, tensor completion (recovery of missing entries), and Fourier measurements for third order tensors. 
\end{abstract}



\section{Introduction and Motivation}

Low rank recovery builds on ideas from the theory of compressive sensing which predicts that sparse vectors 
can be recovered from incomplete measurements via efficient algorithms including $\ell_1$-minimization.
The goal of low rank matrix recovery is to reconstruct an unknown matrix $\mathbf{X} \in \R^{n_1 \times n_2}$ 
from linear measurements $\mathbf{y} = \mathcal{A}(\mathbf{X})$, where $\mathcal{A}: \R^{n_1 \times n_2} \rightarrow \R^m$ with $m \ll n_1n_2$.
Since this is impossible without additional assumptions, one requires that $\mathbf{X}$ has rank at most $r \ll \min \{n_1,n_2\}$, or can at least be approximated
well by a rank-$r$ matrix. This setup appears in a number of applications including signal processing \cite{ahro15,krra14}, quantum state tomography \cite{beeiflgrli09,gr11,kurate14,KabanavaKRT15} and recommender system design \cite{care09,cata10}.

Unfortunately, the natural approach of finding the solution of the optimization problem
\begin{equation}
\min_{\mathbf{Z} \in \R^{n_1 \times n_2}} \rank\zag{\mathbf{Z}} \;\text{ s.t. }\; \mathcal{A}\zag{\mathbf{Z}}=\mathbf{y},
\label{rankNP}
\end{equation}
is NP hard in general. Nevertheless, it has been shown that solving the convex optimization problem
\begin{equation}\label{matrix:nuc:norm}
\min_{\mathbf{Z} \in \R^{n_1 \times n_2}} \norma{\mathbf{Z}}_* \;\text{ s.t. }\; \mathcal{A}\zag{\mathbf{Z}}=\mathbf{y},
\end{equation}
where $\norma{\mbf{Z}}_*=\tr\Big(\zag{\mbf{Z}^*\mbf{Z}}^{1/2}\Big)$ denotes the nuclear norm of a matrix $\mbf{Z}$,
reconstructs $\mathbf{X}$ exactly under suitable conditions on $\mathcal{A}$ \cite{care09,fapare10,fora13,krra14}. Provably optimal measurement maps can be constructed using randomness.
For a  (sub-)Gaussian random measurement map, $m \geq C r \max\{n_1,n_2\}$ measurements are sufficient to ensure stable and robust recovery
via nuclear norm minimization \cite{fapare10,capl11,krra14} and other algorithms such as iterative hard thresholding \cite{tawe13}.
We refer to \cite{KabanavaKRT15} for extensions to ensembles with four finite moments. 

In this note, we go one step further and consider the recovery of low rank tensors $\mathbf{X} \in \R^{n_1 \times n_2 \times \cdots \times n_d}$ of order $d \geq 3$
from a small number of linear measurements $\mathbf{y}=\mathcal{A}\zag{\mathbf{X}}$, where $\mathcal{A}:\R^{n_1 \times n_2 \times \cdots \times n_d} \rightarrow \R^m$, $m \ll n_1n_2\cdots n_d$.  
Tensors of low rank appear in a variety of applications such as video processing ($d=3$) \cite{limuwoye09}, time-dependent 3D imaging ($d=4$), ray tracing where the material dependent
bidirectional reflection function is an order four tensor that has to be determined from measurements \cite{limuwoye09}, 
numerical solution of the electronic Schr{\"o}dinger equation ($d = 3N$, where $N$ is the number
of particles) \cite{lubich-moldyn-2008, Beck99themulticonfiguration, :/content/aip/journal/jcp/131/2/10.1063/1.3173823}, machine learning \cite{aubiporo13} and more.

In contrast to the matrix case, several different notions of tensor rank have been introduced. Similar to the matrix rank being 
related to the singular value decomposition, these notions of rank come with different tensor decompositions.
For instance, the CP-rank of a tensor $\mathbf{X} \in \R^{n_1 \times n_2 \times \cdots \times n_d}$ is defined 
as the smallest number of rank one tensors that sum up to $\mathbf{X}$, where a rank one tensor is of the form 
$\mathbf{u}_1 \otimes \mathbf{u}_2 \otimes \cdots \otimes \mathbf{u}_d$. 
Fixing the notion of rank, the recovery problem can be formulated as computing the minimizer of
\begin{equation}\label{recovery}
\min_{\mathbf{Z} \in \R^{n_1 \times n_2 \times \cdots \times n_d}} \rank\zag{\mathbf{Z}} \;\text{ s.t. }\; \mathbf{y}=\mathcal{A}\zag{\mathbf{Z}}.
\end{equation}
Expectedly, this problem is NP hard (for any reasonable notion of rank), see \cite{ha90-2, hili13}.
An analog of the nuclear norm for tensors can be introduced, and having the power of nuclear norm minimization \eqref{matrix:nuc:norm} for the matrix case in mind,  
one may consider the minimization problem 
$$ \min_{\mathbf{Z} \in \R^{n_1 \times n_2 \times \cdots \times n_d}} \norma{\mathbf{Z}}_* \;\text{ s.t. }\; \mathbf{y}=\mathcal{A}\zag{\mathbf{Z}}.$$
Unfortunately, the computation of $\norma{\cdot}_*$ and, thereby this problem, is NP hard for tensors of order $d \geq 3$ \cite{hili13}, so that one has to develop
alternatives in order to obtain a tractable recovery method.

Previous approaches  include \cite{gareya11, limuwoye09, fapare10, SquareDeal, yuzh14, DBLP:journals/corr/RauhutS15, DBLP:journals/corr/BarakM15}. Unfortunately, none of the proposed methods are completely satisfactory so far. Several contributions \cite{gareya11,limuwoye09,SquareDeal} suggest to minimize the sum of nuclear norms of several tensor matricizations. However,
it has been shown in  \cite{fapare10} that this approach necessarily requires a highly non-optimal number of measurements in order to ensure recovery, see also \cite{SquareDeal}. 
Theoretical results for nuclear tensor norm minimization have been derived in \cite{yuzh14}, but as just mentioned, this approach does not lead to a tractable algorithm.
The theoretical results in \cite{DBLP:journals/corr/ShahRT15} are only applicable to a special kind of separable measurement system, and require a non-optimal number of measurements.
Other contributions \cite{gareya11, limuwoye09, DBLP:journals/corr/RauhutS15} only provide numerical experiments for some algorithms which are often promising but lack a theoretical
analysis. This group of algorithms includes also approaches based on Riemannian optimization on low rank tensor manifolds \cite{abmase09,va13,kestva14}. 
The approaches in \cite{DBLP:journals/corr/RauhutS15, DBLP:journals/corr/BarakM15} are based on tools from real algebraic geometry and provide sum-of-squares relaxations of the tensor nuclear norm. More precisley, \cite{DBLP:journals/corr/RauhutS15} uses theta bodies \cite{blpath13,gouveia2009new}, but provides only numerical recovery results, 
whereas the method in  \cite{DBLP:journals/corr/BarakM15} is highly computationally demanding since it requires solving optimization problems at the sixth level of Lassere's hierarchy.

As proxy for \eqref{recovery}, we introduce and analyze tensor variants of the iterative hard thresholding algorithm, well-known from compressive sensing \cite{blda09,fora13} and low rank matrix recovery \cite{NIPS2010_3904}.
We work with tensorial generalizations of the singular value decomposition, namely the higher order singular value decomposition (HOSVD), 
the tensor train (TT) decomposition, and the hierarchical Tucker (HT) decomposition. These lead to notions of tensor rank, and the  corresponding projections onto low rank tensors --- as required
by iterative hard thresholding schemes --- can be computed efficiently via successive SVDs of certain matricizations of the tensor. Unfortunately, these projections
do not compute best low rank approximations in general which causes significant problems in our analysis. Nevertheless, we are at least able to provide a partial convergence result
(under a certain condition on the iterates) and our numerical experiments indicate that this approach works very well in practice.

The HOSVD decomposition is a special case of the Tucker decomposition which was introduced for the first time in 1963 in \cite{Tuck1963a} and was refined in subsequent articles \cite{tucker64extension, Tuck1966c}. Since then, it has been used e.g.\ in data mining for handwritten digit classification \cite{Savas:2007:HDC:1219177.1219545}, in signal processing to extend Wiener filters \cite{journals/sigpro/MutiB05}, in computer vision \cite{Vasilescu02multilinearanalysis, Wang:2005:OTA:1073204.1073224}, and in chemometrics \cite{Henrion1998, HenrR2000}.

 Recently developed hierarchical tensors, introduced by Hackbusch and coworkers (HT tensors) \cite{HackbuschKuhn, grasedyck2010hierarchical} and the group of Tyrtyshnikov (tensor trains, TT) \cite{os11, oseledets2} have extended the Tucker format into a multi-level framework that no longer suffers from high order scaling of the degrees of freedom 
 with respect to the tensor order $d$, as long as the ranks are moderate. Historically, the hierarchical tensor framework has evolved in the quantum physics community hidden within renormalization group ideas \cite{PhysRevLett.69.2863}, and became clearly visible in the framework of matrix product and tensor network states \cite{schollwock-2011}. An independent source of these developments can be found in quantum dynamics at the multi-layer multi-configurational time dependent HartreeMCTDH  method \cite{Beck99themulticonfiguration, :/content/aip/journal/jcp/131/2/10.1063/1.3173823, lubich-moldyn-2008}.

The tensor IHT (TIHT) algorithm consists of the following steps.
Given measurements $\mathbf{y} = \mathcal{A}(\mathbf{X})$, one starts with some initial tensor $\mbf{X}^0$ (usually $\mbf{X}^0 = \mbf{0}$) and iteratively computes, for $j=0,1,\hdots$,
\begin{align}
\mbf{Y}^{j} &= 
 \mbf{X}^j+\mu_j \mathcal{A}^*\zag{\mbf{y}-\mathcal{A}\zag{\mbf{X}^j}}, 
 \label{IHT1}\\ 
\mbf{X}^{j+1} & = \mathcal{H}_{\mbf{r}}(\mbf{Y}^{j}). \label{IHT2}
\end{align}
Here, $\mu_j$ is a suitable stepsize parameter
and 
$\mathcal{H}_{\mbf{r}}(\mathbf{Z})$ computes a rank-$\mathbf{r}$ approximation of a tensor $\mbf{Z}$ within the given tensor format via successive SVDs (see below). 
Unlike in the low rank matrix scenario, it is in general NP hard to compute the best rank-$\mbf{r}$ approximation of a given tensor $\mbf{Z}$, see \cite{ha90-2, hili13}. 
Nevertheless, $\mathcal{H}_{\mbf{r}}$ computes a quasi-best approximation in the sense that
\begin{equation}\label{Hr:approx} 
\|\mathbf{Z} - \mathcal{H}_{\mbf{r}}(\mathbf{Z})\|_F \leq C_d \|\mathbf{Z} - \mathbf{Z}_{\text{BEST}}\|_F,
\end{equation}
where $C_d \leq C \sqrt{d}$ 
and $\mathbf{Z}_{\text{BEST}}$ denotes the best approximation of $\mathbf{Z}$ of rank $\mbf{r}$ within the given  tensor format.

Similarly to the versions of IHT for compressive sensing and low rank matrix recovery, our analysis of TIHT builds on the assumption
that the linear operator $\mathcal{A}$ satisfies a variant of the restricted isometry property adapted to the tensor decomposition at hand (HOSVD, TT, or HT).
Our analysis requires additionally that at each iteration it holds
\begin{equation}\label{cond:iter}
\| \mathbf{Y}^{j} - \mathbf{X}^{j+1}\|_F \leq (1+\varepsilon) \|\mathbf{Y}^{j} - \mathbf{X}_{\mbf{r}} \|_F,
\end{equation}
where $\varepsilon$ is a small number close to $0$ and $\mathbf{X}_{\mbf{r}}$ is the best rank-$\mbf{r}$ approximation to $\mbf{X}$, the tensor to be recovered, 
see Theorem~\ref{convergenceclassTIHT} for details.
(In fact, $\mathbf{X}_{\mbf{r}} = \mathbf{X}$ if $\mbf{X}$ is exactly of rank at most $\mbf{r}$.) Unfortunately, \eqref{Hr:approx} only guarantees that
\[
\| \mathbf{Y}^{j} - \mathbf{X}^{j+1}\|_F \leq C_d \| \mathbf{Y}^{j} - \mathbf{Y}_{\text{BEST}}\|_F \leq C_d  \|\mathbf{Y}^{j} - \mathbf{X}_{\mbf{r}} \|_F.
\]
Since $C_d$ cannot be chosen as $1+\varepsilon$, condition \eqref{cond:iter} cannot be guaranteed a priori. The hope is, however, that \eqref{Hr:approx} is only a worst case estimate, and that
usually a much better low rank approximation to $\mathbf{Y}^j$ is computed satisfying \eqref{cond:iter}. At least, our numerical experiments indicate that this is the case.
Getting rid of condition \eqref{Hr:approx} seems to be a very difficult, if not impossible, task --- considering also that there are no other completely rigorous results for tensor recovery with efficient
algorithms available that work for a near optimal number of measurements. (At least our TRIP bounds below give some hints on what the optimal number of measurements should be.)

The second main contribution of this article consists in an analysis of the TRIP related to the tensor formats HOSVD, TT, and HT for random measurement maps.
We show that subgaussian linear maps $\mathcal{A}:\Rd \rightarrow \R^m$ satisfy the TRIP at rank $\mbf{r}$ and level $\delta_{\mbf{r}}$ with probability exceeding $1-\varepsilon$ provided that
\begin{align*} 
&m \geq C_1 \delta_{\mbf{r}}^{-2} \max\skup{\zag{r^d+dnr}\log\zag{d}, \log\zag{\varepsilon^{-1}}}, \quad \text{for HOSVD,} \\
&m \geq C_2 \delta_{\mbf{r}}^{-2} \max\skup{\zag{(d-1)r^3+dnr} \log\zag{d r}, \log\zag{\varepsilon^{-1}}}, \quad \text{for TT and HT,}
\end{align*}
where $C_1,C_2 >0$ are universal constants and $n=\max\skup{n_i: i \in \uglate{d}}$, $r=\max\skup{r_t: t \in T_I}$ with $T_I$ be the corresponding tree.  
Up to the logarithmic factors, these bounds match the number of degrees of freedom of a rank-$\mbf{r}$ tensor in the particular format, and therefore are almost optimal.

In addition, we show a similar result for  linear maps  $\mathcal{A}:\C^{n_1\times n_2 \times \cdots \times n_d} \rightarrow \C^m$ that are constructed by composing random sign flips of the tensor entries
with a $d$-dimensional Fourier transform followed by random subsampling, see Theorem~\ref{FourierTRIP} for details.

  The remainder of the paper is organized as follows. In Section \ref{Section:Intrototensors}, we introduce the HOSVD, TT, and HT tensor decompositions and the corresponding notions of rank used throughout the paper. Two versions of the tensor iterative hard thresholding algorithm (CTIHT and NTIHT) are presented in Section \ref{Section:TIHT} and 
  a partial convergence proof is provided.
Section~\ref{Sec:TRIP} proves the bounds on the TRIP for subgaussian measurement maps, while Section~\ref{Sec:Fourier} extends them to randomized Fourier maps.
 In Section \ref{Sec:Numerics}, we present some numerical results on recovery of third order tensors.

\subsection{Notation}\label{sec:notation}
We will mostly work with real-valued $d$-th order tensors $\mathbf{X} \in \R^{n_1 \times n_2 \times \cdots \times n_d}$, but the notation introduced below holds
analogously also for complex-valued tensors 
$\mathbf{X} \in \C^{n_1 \times n_2 \times \cdots \times n_d}$ which will appear in Section~\ref{Sec:Fourier}. 
Matrices and tensors are denoted with capital bold letters, linear mappings with capital calligraphic letters, sets of matrices or tensors with bold capital calligraphic letters, and vectors with small bold letters. The expression $\uglate{n}$ refers to the set $\skup{1,2,\ldots,n}$.

 With $\mathbf{X}_{i_k=p}$, for $ p \in \uglate{n_k}$, we denote the $\zag{d-1}$-th order tensor (called subtensor) of size $n_1 \times n_2 \times \cdots \times n_{k-1} \times n_{k+1} \times \cdots \times n_d$ that is obtained by fixing the $k$-th component of a tensor $\mathbf{X}$ to $p$ i.e., $\mathbf{X}_{i_k=p} \zag{i_1, \ldots, i_{k-1}, i_{k+1}, \ldots, i_d}=\mathbf{X}\zag{i_1, \ldots, i_{k-1}, p, i_{k+1}, \ldots, i_d}$, for all $i_l \in \uglate{n_l}$ and for all $l \in \uglate{d} \backslash \skup{k}$. A matrix obtained by taking the first $r_k$ columns of the matrix $\mathbf{U}$ is denoted by $\mathbf{U}\zag{:,\uglate{r_k}}$. Similarly, for a tensor $\mbf{S} \in \Rd$ the subtensor $\mathbf{S}\zag{\uglate{r_1},\uglate{r_2}, \ldots, \uglate{r_d}} \in \R^{r_1 \times r_2 \times \cdots \times r_d}$ is defined elementwise as $\mathbf{S}\zag{\uglate{r_1},\uglate{r_2},\ldots,\uglate{r_d}}\zag{i_1,i_2,\ldots,i_d}=\mathbf{S}\zag{i_1,i_2,\ldots,i_d}$, for all $i_k \in \uglate{r_k}$ and for all $k \in \uglate{d}$.

The vectorized version of a tensor $\mbf{X} \in \Rd$ is denoted by $\vecc\zag{\mbf{X}} \in \R^{n_1n_2\cdots n_d}$ (where the order
of indices is not important as long as we remain consistent).
The operator $\tensorization$ transforms back a vector $\mbf{x} \in \R^{n_1 n_2 \ldots n_d}$ into a $d$-th order tensor in $\Rd$, i.e., 
\begin{equation} \label{eq:tensorization}
\tensorization\zag{\vecc\zag{\mbf{X}}}= \mbf{X}, \quad \text{ for } \mbf{X} \in \Rd.
\end{equation}
The inner product of two tensors $\mathbf{X}, \mathbf{Y} \in \R^{n_1 \times n_2 \times \cdots \times n_d}$ is defined as
\begin{equation}\label{innerprod}
\interval{\mathbf{X}, \mathbf{Y}}=\sum_{i_1=1}^{n_1}\sum_{i_2=1}^{n_2} \ldots \sum_{i_d=1}^{n_d} \mathbf{X}\zag{i_1, i_2, \ldots, i_d}\mathbf{Y}\zag{i_1,i_2, \ldots, i_d}.
\end{equation}
The (Frobenius) norm of a tensor $\mathbf{X} \in \R^{n_1 \times n_2 \times \cdots \times n_d}$, induced by this inner product, is given as
\begin{equation}\label{FrobNorm}
\norma{\mathbf{X}}_F=\interval{\mbf{X},\mbf{X}}^{1/2}=\sqrt{\sum_{i_1=1}^{n_1} \sum_{i_2=1}^{n_2} \ldots \sum_{i_d=1}^{n_d} \mathbf{X}^2\zag{i_1, i_2, \ldots, i_d}}.
\end{equation}

Matricization (also called flattening) is a linear transformation that reorganizes a tensor into a matrix. 
More precisely, for a $d$-th order tensor $\mbf{X} \in \R^{n_1 \times n_2 \times \cdots \times n_d}$ and an ordered subset $\bm{\mathcal{S}} \subseteq \uglate{d}$, the $\bm{\mathcal{S}}$-matricization $\mbf{X}^{\bm{\mathcal{S}}} \in \R^{\prod_{k \in \bm{\mathcal{S}}} n_k \times \prod_{\ell \in \bm{\mathcal{S}}^c}n_{\ell}}$ is defined as
$$\mbf{X}^{\bm{\mathcal{S}}} \zag{(i_k)_{k \in \mathcal{S}};(i_\ell)_{\ell \in \mathcal{S}^c}}=\mbf{X}\zag{i_1,i_2,\ldots,i_d}, $$
i.e., the indexes in the set $\bm{\mathcal{S}}$ define the rows of a matrix and the indexes in the set $\bm{\mathcal{S}}^c=\uglate{d}\backslash \bm{\mathcal{S}}$ define the columns. 
For a singleton set $\bm{\mathcal{S}}=\{k\}$, where $k \in \uglate{d}$, matrix $\mbf{X}^{\bm{\mathcal{S}}}$ is called the  mode-$k$ matricization or the $k$-th unfolding.

For $\mathbf{X} \in \R^{n_1 \times n_2 \times \cdots \times n_d}$, $\mathbf{A} \in \R^{J \times n_k}$, 
the $k$-mode multiplication $\mathbf{X} \times_k \mathbf{A} \in \R^{n_1 \times \cdots \times n_{k-1} \times J \times n_{k+1} \times \cdots \times n_d}$ is defined element-wise as
\begin{equation}\label{def:kmode}
 \zag{\mathbf{X} \times_k \mathbf{A}}\zag{i_1,\ldots,i_{k-1},j,i_{k+1},\ldots,i_d} = \sum_{i_k=1}^{n_k} \mathbf{X}\zag{i_1,\ldots,i_d}\mathbf{A}\zag{j, i_k}, \quad k \in [d].
\end{equation}
For a tensor $\mbf{X} \in \Rd$ and matrices $\mbf{A} \in \R^{J \times n_j}$, $\mbf{B} \in \R^{K \times n_k}$, $\mbf{C} \in \R^{L \times K}$ it holds
\begin{align}
& \mbf{X} \times_j \mbf{A} \times_k \mbf{B} = \mbf{X} \times_k \mbf{B} \times_j \mbf{A}, \text{ whenever } j\neq k \\
& \mbf{X} \times_k \mbf{B} \times_k \mbf{C} = \mbf{X} \times_k \mbf{CB}. 
\end{align}
Notice that the SVD decomposition of a matrix $\mathbf{X} \in \R^{n_1 \times n_2}$ can be written using the above notation as $\mathbf{X} =\mathbf{U} \mathbf{\Sigma} \mathbf{V}^T = \mathbf{\Sigma} \times_1 \mathbf{U} \times_2 \mathbf{V}.$

We can write the measurement operator $\mathcal{A}:\Rd \rightarrow \R^m$ in the form
\[
\zag{\mathcal{A}\zag{\mathbf{X}}}_i=\interval{\mbf{A}_i,\mbf{X}}, \quad i \in \uglate{m},
\]
for a set of so-called {\it sensing tensors} $\mbf{A}_i \in \Rd$, for $i \in \uglate{m}$.
The matrix representation $\mbf{A} \in \R^{m \times n_1n_2\cdots n_d}$ of $\mathcal{A}$ 
is defined as
\begin{equation*}
\mbf{A}\zag{i,:}=\vecc\zag{\mbf{A}_i}^T, \text{ for } i \in \uglate{m},
\end{equation*}
where $\mbf{A}\zag{i,:}$ denotes the $i$-th row of  $\mbf{A}$ and $\vecc\zag{\mbf{A}_i}$ denotes the vectorized version of the sensing tensor $\mbf{A}_i$.  
Notice that $\mathcal{A}\zag{\mbf{X}}=\mbf{A}\vecc\zag{\mbf{X}}.$

\subsection{Acknowledgement}

H.~Rauhut and \v{Z}.~Stojanac acknowledge support by the Hausdorff Center for Mathematics at the University of Bonn, by
the Hausdorff Institute for Mathematics Bonn during the trimester program {\it Mathematics of Signal Processing} and by the European Research Council through
the grant StG 258926.

\section{Tensor decompositions and tensor rank} \label{Section:Intrototensors}

Before studying the tensor recovery problem we first need to introduce the tensor decompositions and the associated notion of rank 
that we are working with.
We confine ourselves to finite dimensional linear spaces $\R^{n_i}$ from which the tensor product space $$\mathcal{H}_d=\bigotimes_{i=1}^d \R^{n_i}$$ is built. 
Then any $\mbf{X} \in \mathcal{H}_d$ can be represented as
$$ \mbf{X}=\sum_{\mu_1=1}^{n_1} \ldots \sum_{\mu_d=1}^{n_d} \mbf{X}\zag{\mu_1,\ldots,\mu_d} \mbf{e}_{\mu_1}^1 \otimes \cdots \otimes  \mbf{e}_{\mu_d}^d,$$
where $\{\mbf{e}_1^i,\ldots,\mbf{e}_{n_i}^i\}$ is the canonical basis of the space $\R^{n_i}$.
 Using this basis, with slight abuse of notation, we can identify $\mbf{X} \in \mathcal{H}_d$ with its representation by a $d$-variate function, often called hyper matrix,
 $$ \bm{\mu}=\zag{\mu_1,\ldots,\mu_d} \mapsto \mbf{X}\zag{\mu_1,\ldots,\mu_d} \in \R, \quad \mu_i \in \uglate{n_i}, \, i \in \uglate{d},$$
 depending on a discrete multi-index $\bm{\mu} = (\mu_1,\hdots,\mu_{d})$.
 We equip the linear space $\mathcal{H}_d$ with the inner product defined in \eqref{innerprod} and the 
 Frobenius norm defined in \eqref{FrobNorm}. 
 \medskip
 
 The idea of the classical {\it Tucker format} is to search, given a tensor $\mbf{X}$ and a rank-tuple $\mbf{r}=\zag{r_j}_{j=1}^d$,
 for optimal (or at least near optimal) subspaces $U_i \subset \R^{n_i}$ of dimension $r_i$, $i \in [d]$,
such that
 $$ \min_{\mbf{Y} \in U_1 \otimes \cdots \otimes U_d} \norma{\mbf{X}-\mbf{Y}}_F$$ 
 is minimized.
 Equivalently, for each coordinate direction $i \in \uglate{d}$, we are looking for corresponding basis $\skup{\mbf{u}_{k_i}^i}_{k_i \in [r_i]}$ of $U_i$, which can be written in the form
\begin{equation} \label{leaftransfer}
\mbf{u}_{k_i}^i:=\sum_{\mu_i=1}^{n_i} \mbf{U}^i \zag{\mu_i,k_i} \mbf{e}_{\mu_i}^i, \quad k_i \in \uglate{r_i}, \, r_i<n_i,
\end{equation}
  where $\mbf{U}^i\zag{\mu_i,k_i} \in \R$. With a slight abuse of notation we often identify the basis vectors $\mbf{u}_{k_i}^i$ with their
  representation $\big(\mathbf{U}^i  ( \mu_i, k_i )  \big)_{\mu_i \in [n_i]}$.
Given the basis $\skup{\mbf{u}_{k_i}^i}_{k_i}$, a tensor $\mbf{X}  \in \bigotimes_{i=1}^d U_i \subset \mathcal{H}_d=\bigotimes_{i=1}^d \R^{n_i}
$ can be represented by 
\begin{equation}\label{eq:Tucker}
\mbf{X}=\sum_{k_1=1}^{r_1} \cdots \sum_{k_d=1}^{r_d} \mbf{C}\zag{k_1,\ldots,k_d} \mbf{u}_{k_1}^1 \otimes \cdots \otimes \mbf{u}_{k_d}^d .
\end{equation} 
 In case $\skup{\mbf{u}_{k_i}^i}_{k_i}, i \in \uglate{d}$, form orthonormal bases, the {\it core tensor} $\mbf{C} \in \bigotimes_{i=1}^d \R^{r_i}$ is given entry-wise by 
 $$\mbf{C}\zag{k_1,\ldots,k_d}=\interval{\mbf{X}, \mbf{u}_{k_1}^1 \otimes \cdots \otimes \mbf{u}_{k_d}^d}. $$
 In this case, we call the Tucker decomposition \eqref{eq:Tucker} a higher order singular value decomposition (or HOSVD decomposition). The HOSVD decomposition 
 can be constructed such that it satisfies the following properties
 \begin{itemize}
 \item the bases $\{\mbf{u}_{k_i}^i \in \R^{n_i}: k_i \in \uglate{r_i}\}$ are orthogonal and normalized, for all $i \in \uglate{d}$;
 \item the core tensor $\mbf{C} \in \mathcal{H}_d$ is all orthogonal, i.e., $\interval{\mbf{C}_{k_i=p},\mbf{C}_{k_i=q}}=0$, for all $i \in \uglate{d}$ and whenever $p\neq q$;
 \item the subtensors of the core tensor $\mbf{C}$ are ordered according to their Frobenius norm, i.e., $\norma{\mbf{C}_{k_i=1}}_F \geq \norma{\mbf{C}_{k_i=2}}_F \geq \cdots \geq \norma{\mbf{C}_{k_i=n_i}}_F \geq 0$.
  \end{itemize}
In contrast to the matrix case, the rank of a tensor is a tuple $\mbf{r} = (r_1,\hdots,r_d)$ of $d$ numbers. For $\mbf{X}  \in \bigotimes_{i=1}^d U_i$ 
it is obtained via the matrix ranks of the unfoldings, i.e.,  $$r_k=\rank\zag{\mbf{X}^{\{k\}}}, \quad k \in \uglate{d}.$$
We say that  tensor $\mbf{X} \in \Rd$ has rank at most $\mbf{r}$ if $$\rank\zag{\mbf{X}^{\{k\}}} \leq r_k \quad \text{for } k \in \uglate{d}. $$
For high order $d$, the HOSVD has the disadvantage that the core tensor contains $r_1\cdots r_d \sim r^d$ entries (assuming $r_i \approx r$ for all $i$), which
may potentially all be nonzero. Therefore, the number of degrees of freedom that are required to describe a rank-$\mbf{r}$ tensor scales exponentially in $d$.
Assuming $n_i \sim n$ for all $i$, the overall complexity for storing the required data (including the basis vectors) scales like $\mathcal{O}\zag{ndr+r^d}$
(which nevertheless for small $r_i / n_i \sim \alpha$ implies a high compression rate of $\alpha^d$).
Without further sparsity of the core tensor the Tucker format is appropriate for low order tensors, i.e., $d=3$.

The HOSVD $\mathbf{X}=\mathbf{S} \times_1 \mathbf{U}_1 \times_2 \mbf{U}_2 \times \cdots \times_d \mathbf{U}_d $ of a tensor $\mbf{X}$ can be computed via the singular value decompositions (SVDs) of all 
unfoldings. The columns of the matrices $\mbf{U}_k$ contain all left-singular vectors of the $k$th unfolding $\mbf{X}^{\{k\}}$, $k \in [d]$,
and the core tensor is given by $\mbf{S}=\mbf{X} \times_1 \mbf{U}_1^T \times_2 \mbf{U}_2^T \times \cdots \times_d \mbf{U}_d^T$. 
An important step in the iterative hard thresholding algorithm consists in computing a rank-$\mbf{r}$ approximation to a tensor $\mbf{X}$ (the latter not necessarily being of low rank).
Given the HOSVD decomposition of $\mathbf{X}$ as just described, such an approximation is given by
$$
\mathcal{H}_{\mathbf{r}}\zag{\mathbf{X}}=\overline{\mathbf{S}} \times_1 \overline{\mathbf{U}}_1 \times_2 \overline{\mbf{U}}_2 \times \cdots \times_d \overline{\mathbf{U}}_d,
$$ 
where $\overline{\mathbf{S}}=\mathbf{S}\zag{\uglate{r_1},\uglate{r_2},\ldots,\uglate{r_d}} \in \R^{r_1 \times r_2 \times \cdots \times r_d}$ and $\overline{\mathbf{U}}_k=\mathbf{U}_k\zag{:,\uglate{r_k}} \in \R^{n_k \times r_k}$ for all $k \in \uglate{d}$. Thereby, the matrix $\overline{\mbf{U}}_k$ contains $\mbf{r}_k$ left-singular vectors corresponding to the $r_k$ largest singular values of $\mbf{X}^{\{k\}}$, for each $k \in [d]$.
In general, $\mathcal{H}_{\mathbf{r}}\zag{\mathbf{X}}$ is unfortunately not the best rank-$\mbf{r}$ approximation to $\mathbf{X}$, formally defined as the minimizer $\mbf{X}_{\text{BEST}}$ of
\[
\min_{\mbf{Z}} \|\mathbf{X} - \mathbf{Z}\|_F \quad \mbox{subject to } \rank(\mbf{Z}^{\{k\}}) \leq r_k \mbox{ for all } k \in [d].
\]
Nevertheless, one can show that $\mathcal{H}_{\mathbf{r}}\zag{\mathbf{X}}$ is a quasi-best rank-$\mbf{r}$ approximation in the sense that
\begin{equation}\label{Hr:HOSVD}
\| \mathbf{X} - \mathcal{H}_{\mathbf{r}}\zag{\mathbf{X}}\|_F \leq \sqrt{d} \|\mathbf{X} - \mathbf{X}_{\text{BEST}}\|_F.
\end{equation}
Since the best approximation $\mathbf{X}_{\text{BEST}}$ is in general NP hard to compute \cite{ha90-2, hili13}, $\mathcal{H}_{\mathbf{r}}$ may serve as a tractable and reasonably good alternative.
 \medskip
 
The \emph{hierarchical Tucker format} (HT)  in the form introduced by Hackbusch and K\"uhn in \cite{HackbuschKuhn},  
extends the idea of subspace approximation to a hierarchical or multi-level framework, which potentially does not suffer exponential scaling of the number of degrees of freedom
in $d$ anymore.  
In order to introduce this format, we first consider 
$ V_1 \otimes V_2 = \mathbb{R}^{n_1} \otimes \mathbb{R}^{n_2}$ or preferably the subspaces 
$U_1 \otimes U_2 \subset  V_1 \otimes V_2$, similarly to the HOSVD above. 
For the approximation of $\mathbf{X} \in \mathcal{H}_d$ we use a subspace $U_{\{ 1,2 \} } \subset  U_1 \otimes U_2 $ with dimension $r_{\{1,2\}}  \leq r_1  r_2$. 
We can define $U_{\{ 1,2\}} $  
via a basis,
$U_{\{ 1,2\}}  =  \mbox{span } 
\{  \mathbf{u}^{\{ 1,2 \} }_{k_{ \{ 1,2 \}}}  :  k_{ \{ 1,2\} }  = 1, \ldots , r_{\{ 1,2\} }  \} $,
with basis vectors  given by 
$$
  \mathbf{u}^{\{ 1,2 \} }_{k_{ \{ 1,2 \}}}  = \sum_{{k}_1 =1 }^{r_1} \sum_{{k}_2 =1}^{r_2}
  \mathbf{B}_{\{ 1,2 \} } ( k_{\{ 1,2 \} } , k_1 , k_2 )  \;  {\bf u}^1_{k_1}   \otimes  {\bf u}^2_{k_2 } , \quad k_{\{ 1,2 \} }  =1,\hdots,r_{\{ 1,2\} },
$$
for some numbers $ \mathbf{B}_{\{ 1,2 \} } ( k_{\{ 1,2 \} } , k_1 , k_2  )$.
One may continue hierarchically in several ways, e.g.~by building a subspace  
$U_{\{ 1,2,3\} } \subset  U_{\{ 1,2\} } \otimes U_3  \subset U_1 \otimes U_2 \otimes U_3 \subset 
V_1 \otimes V_2 \otimes V_3 $, or 
  $U_{\{ 1,2,3,4\} } \subset  U_{\{ 1,2\} } \otimes U_{\{3,4\}}  $, where $U_{\{3,4\}} $ is defined analogously to $U_{\{1,2\}} $
   and so on.  

For a systematic treatment, this approach can be cast into the framework of a partition tree, 
with leaves $ \{1\}, \ldots \{d\} $, 
simply abbreviated by $1,\ldots,d$, 
and vertices  $ \alpha  \subset D:= \{ 1, \ldots , d\}$. 
The partition tree $T_I$ contains partitions of a vertex $\alpha$ (not being a leaf) into vertices $\alpha_1,\alpha_2$, i.e.,
$ \alpha = \alpha_1 \cup \alpha_2$, $ \alpha_1 \cap \alpha_2  = \emptyset$. 
We call $\alpha_1, \alpha_2 $ the {\em sons} of the {\em father}  $\alpha$.
In this notation, we can assume without loss of generality that $i <j$ for all $i \in \alpha_1$, $j \in \alpha_2$. 
The vertex $D$ is called the {\em root} of the tree. The set of leaves of a tree $T_I$ is denoted by $\mathcal{L}\zag{T_I}$ and the set of interior (non-leaf) vertices by
$\mathcal{I}\zag{T_I}$.
In the example above we have $ \alpha  : = \{ 1,2,3 \} = \alpha_1 \cup \alpha_2$  
with $\alpha_1 :=  \{ 1,2 \} $ {and} $ \alpha_2 : = \{ 3 \}$. The partition tree corresponding to the HT representation in Figure~\ref{fig:hierarchical} is given as 
$T_I=\skup{\{1\},\{2\},\{1,2\},\{3\},\{1,2,3\},\{4\},\{5\},\{4,5\},\{1,2,3,4,5\}}$ with $\mathcal{L}\zag{T_I} = \{ \{1\},\{2\},\{3\},\{4\},\{5\}\}$ and $\mathcal{I}\zag{T_I} = \{\{1,2\},\{1,2,3\},\{4,5\},\{1,2,3,4,5\}\}$. 
In general, we do not need to restrict the number of sons
of a vertex, but for simplicity we confine ourselves in the following to binary trees, i.e., to two sons per father.

Consider a non-leaf vertex $\alpha$, $\alpha \not= \{ i \}$,
with two sons $\alpha_1,\alpha_2$.
Then the corresponding subspace $U_{\alpha } \subset U_{\alpha_1} \otimes U_{\alpha_2}$ with $\dim U_\alpha = r_\alpha$ is defined by a  basis 
 \begin{equation} \label{eq:transfer} 
  \mathbf{u}^{\alpha}_{\ell}   = \sum_{i=1 }^{r_{\alpha_1} } \sum_{j=1}^{r_{\alpha_2}} 
      \mathbf{B}_{\alpha} (\ell , i,j   ) \, \mathbf{u}^{\alpha_1 }_i  \otimes  \mathbf{u}^{\alpha_2}_j, \quad \ell \in [r_{\alpha}],
\end{equation}
which is often represented by a matrix $\mathbf{U}^{\alpha}   \in  \mathbb{R}^{n_{\alpha} \times r_{\alpha}}$
 with columns $\mbf{U}^{\alpha} \zag{:,\ell}= \vecc\zag{\mbf{u}_{\ell}^{\alpha}}$ and $n_{\alpha}=\prod_{\ell \in \alpha} n_{\ell}$. 
Without loss of generality, all basis vectors $\mathbf{u}^{\alpha }_{\ell}$, $\ell=1,\ldots,r_{\alpha}$,
can be chosen to be orthonormal as long as $\alpha$ is not the root ($\alpha \neq D$).
The tensors $(\ell , i,j) \mapsto \mathbf{B}_{\alpha} (\ell, i,j) $ are called {\em transfer} or {\em component tensors}. 
For a leaf $\{i\} \simeq i$, the matrix 
$ (\mathbf{U}^i(\mu_i,k_i))_{\mu_i,k_i} \in \R^{n_i \times r_i}$ representing the basis of $U_i$ as
in \eqref{leaftransfer} is called {\em $i$-frame}.
The component tensor $ \mathbf{B}_{D} = \mathbf{B}_{ \{1,  \ldots, d\} } $  at the root   is called the {\em root tensor}.

The rank tuple $\mathbf{r} = (r_\alpha)_{\alpha \in T_I}$ of a tensor $\mathbf{X}$ associated to a partition tree $T_I$ is defined via the (matrix) ranks of the matricizations $\mathbf{X}^\alpha$, i.e.,
\[
r_\alpha = \rank \zag{\mbf{X}^{\alpha}} \quad \mbox{ for } \alpha \in T_I.
\]
In other words, a tensor $\mathbf{X}$ of rank $\mathbf{r}$ obeys several low rank matrix constraints simultaneously 
(defined via the set $\{ \mathbf{X}^\alpha : \alpha \in T_I\}$ of matricizations).
When choosing the right subspaces $U_\alpha$ related to $\mathbf{X}$, these ranks correspond precisely to the numbers $r_\alpha$ (and $r_{\alpha_1}, r_{\alpha_2}$) in \eqref{eq:transfer}.

It can be shown \cite{Hackbuschbook} that a tensor of rank $\mathbf{r}$ is determined completely by the transfer tensors 
$\mathbf{B}_{t} $, $ t\in \mathcal{I}\zag{T_I} $ and the $\alpha$-frames $\mbf{U}_{\alpha}$, $\alpha \in \mathcal{L}\zag{T_I}$. 
This correspondence is realized by a multilinear function $\tau$, i.e.,
\begin{equation*} 
\big( \{ \mathbf{B}_{t} : t \in \mathcal{I}\zag{T_I} \}, \{\mathbf{U}_{\alpha}: \alpha \in \mathcal{L}\zag{T_I} \}  \big) \mapsto
\mathbf{X} = \tau \big( \{ \mathbf{B}_{t} : t \in \mathcal{I}\zag{T_I} \}, \{\mathbf{U}_{\alpha}: \alpha \in \mathcal{L}\zag{T_I} \}  \big). 
\end{equation*}
For instance, the tensor presented in Figure~\ref{fig:hierarchical} is completely parametrized by 
$$
\zag{\skup{\mbf{B}_{\{1,2\}}, \mbf{B}_{\{1,2,3\}}, \mbf{B}_{\{4,5\}}, \mbf{B}_{\{1,2,3,4,5\}}}, \skup{\mbf{U}_1,\mbf{U}_2,\mbf{U}_3,\mbf{U}_4,\mbf{U}_5}}.
$$
The map $\tau$ is defined  by applying
(\ref{eq:transfer}) recursively. Since $\mathbf{B}_{t} $ depends bi-linearly on  $\mathbf{B}_{t_1} $ and 
$\mathbf{B}_{t_2}$, the composite function $ \tau $ is indeed multi-linear in its arguments $ \mathbf{B}_{t}$ and $\mbf{U}_{\alpha}$. 

One can store a tensor $\mathbf{X} \in \Rd$ by storing only 
the transfer tensors $\mathbf{B}_{t} $, $ t\in \mathcal{I}\zag{T_I} $ and the $\alpha$-frames $\mbf{U}_{\alpha}$, $\alpha \in \mathcal{L}\zag{T_I}$, which implies
significant compression in the low rank case. More precisely, setting $ n :=  \max \{n_i:i=1 , \ldots , d\} $, $r := \max \{ r_{t} : t \in T_I \}$
the number of data required for such a representation of a rank-$\mbf{r}$ tensor is 
$\mathcal{O} (ndr +  d r^3)$, in particular it does not scale exponentially with respect to the order $d$.

Like for the HOSVD, computing the best rank-$\mbf{r}$ approximation $\mbf{X}_{\text{BEST}}$ 
to a given tensor $\mbf{X}$, i.e., the minimizer of $\mbf{Z} \mapsto \| \mbf{X} - \mbf{Z}\|_F$ subject to
$\rank(\mbf{Z}^{\alpha}) \leq r_{\alpha}$ for all $\alpha \in T_I$ is NP hard. A quasi-best approximation $\mathcal{H}_\mbf{r}(\mbf{Z})$ 
can be computed efficiently via successive SVDs. In \cite{grasedyck2010hierarchical} 
two strategies are introduced: hierarchical root-to-leaves truncation or hierarchical leaves-to-root truncation, the latter being the computationally more efficient one.
Both strategies satisfy
\begin{equation}\label{Hr:HT}
\| \mathbf{X} -\mathcal{H}_\mbf{r}(\mathbf{X})\|_F \leq C_d \| \mathbf{X} - \mathbf{X}_{\text{BEST}}\|_F,
\end{equation}
where $C_d = \sqrt{2d-2}$ for root-to-leaves truncation and $C_d = (2+\sqrt{2})\sqrt{d}$ for leaves to root truncation.
We refer to \cite{grasedyck2010hierarchical} for more details and another method that achieves $C_d = \sqrt{2d-3}$.

\begin{figure}
\caption{Hierarchical Tensor representation of an order $5$ tensor}
\centering
\includegraphics[scale=0.9, trim=90 262 342 322, clip]{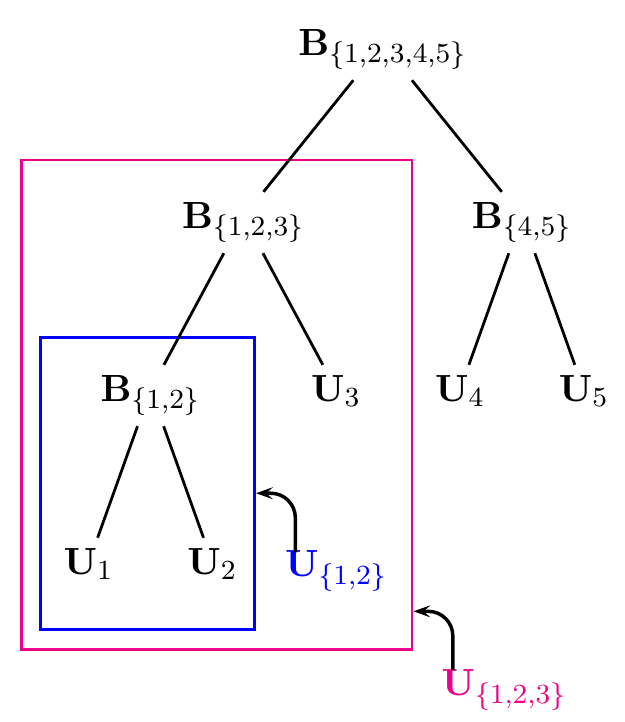}
\label{fig:hierarchical}
\end{figure}

\medskip

An important special case of a hierarchical tensor decomposition 
is the {\it tensor train decomposition} (TT) \cite{os11,oseledets2}, also known as
 {\em  matrix product states} in the physics literature.  
It is defined via the unbalanced tree  
$$T_I = \{  \{ 1, 2, 3,\ldots,d\} , \{1\},\{ 2, 3, \ldots, d\} , \{2\}, \{3,\ldots, d\}, \{3\}, \ldots ,\{d-1,d\},\{d-1\},\{ d\}\}. $$
The corresponding subspaces satisfy $U_{ \{ 1, \ldots , p +1 \}} \subset U_{\{ 1, \ldots , p \} } \otimes V_{\{p+1\}}$.
The  $\alpha$-frame $\mathbf{U}_{\alpha}$ for a leaf $ \alpha \in \{\{1\}, \{2\},\ldots, \{d-1\}\}$ is usually defined as identity matrix of appropriate size.
Therefore, the tensor $\mathbf{X} \in \mathcal{H}_d$ is completely parametrized by the transfer tensors $\mathbf{B}_{t}, {t \in \mathcal{I}(T_I)}$, and the 
$d$-frame $\mathbf{U}_{\{d\}}$. 
Applying the recursive construction,  the tensor $\mathbf{X}$  can be written as
\begin{eqnarray}
 ( \mu_1, \ldots , \mu_d )  & \mapsto & \mathbf{X} ( \mu_1, \ldots , \mu_d )  
 =  \sum_{k_1=1}^{r_1} \ldots\sum_{k_{d-1}=1}^{r_{d-1}} 
  {\mathbf{B}_1 ( \mu_1 , k_1) \mathbf{B}_2 (k_1,\mu_2 ,k_2) \cdots  
 \mathbf{B}_d (k_{d-1},\mu_d )} \ ,  \label{eq:utt} 
\end{eqnarray}
where $\mbf{B}_d:=\mbf{U}_{\{d\}}$ and we used the abbreviation $\alpha \simeq \{\alpha, \alpha+1,\ldots, d\}$, for all $\alpha \in \uglate{d-1}$.
\begin{figure}
\caption{TT representation of an order $5$ tensor with abbreviation $i \simeq \{i,\ldots,d\}$ for the interior nodes}
\centering
\includegraphics[scale=0.8, trim=85 342 398 267, clip]{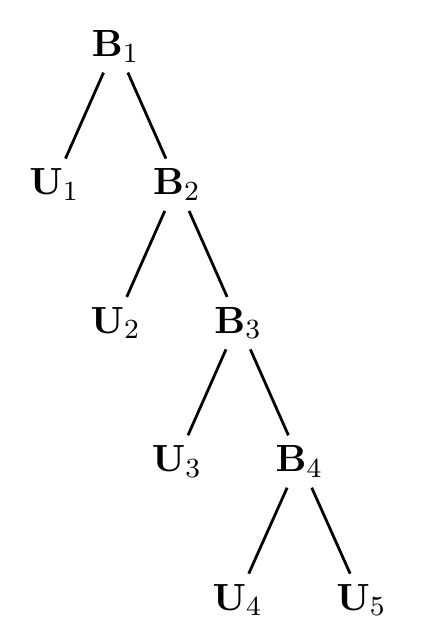}
\label{fig:tt}
\end{figure}
Introducing  the matrices 
$  \mathbf{G}_i (\mu_i) \in \mathbb{R}^{r_{i-1} \times r_{i} }  $,
$$
\big( \mathbf{G}_i (\mu_i) \big)_{k_{i-1},k_i } = \mathbf{B}_i  ( k_{i-1} , \mu_i , k_i  ) , \quad i = 2, \hdots, d-1 , 
$$
and, with the convention that
$r_0 =r_d =1$,   
$
\big( \mathbf{G}_1 (\mu_1) \big)_{k_1 } = \mathbf{B}_1 ( \mu_1 , k_1 )$   
and
\begin{equation}\label{def:Gd}
\big( \mathbf{G}_d (\mu_d) \big)_{k_{d-1} } = \mathbf{B}_d (k_{d-1}, \mu_d ),
\end{equation}
formula (\ref{eq:utt}) can be rewritten entry-wise  by
matrix--matrix products
\begin{equation}
\label{eq:mps} 
  \mathbf{X} (\mu_1, \ldots, \mu_d )  = {\mathbf{G}_{1} (\mu_{1})  
\cdots \mathbf{G}_{i} (\mu_{i}) \cdots \mathbf{G}_{d} (\mu_{d})} = \tau ( \mathbf{B}_1 , \ldots ,  \mathbf{B}_{d} ). 
\end{equation}
This representation is by no means unique. The tree is ordered according to the father-son relation into a hierarchy of levels, where $\mathbf{B}_1$ is the  root tensor. 
As in the HOSVD and HT scenario, the rank tuple $\mbf{r} = (r_1,\hdots,r_{d-1})$ is given by the ranks of certain matricizations, i.e., 
$$r_i=\rank\zag{\mbf{X}^{\{i,i+1,\ldots,d\}}} \quad \mbox{ for } i \in \uglate{d-1}.$$
The computation of $\mbf{X}_{\text{BEST}} = \operatorname{argmin}_\mbf{Z} \| \mbf{X} - \mbf{Z}\|_F$ subject to $\rank(\mbf{Z}^{\{i,i+1,\ldots,d\}}) \leq r_i$ is
again NP hard (for $d\geq 3$). A quasi-best approximation in the TT-format of a given tensor $\mbf{X}$ can be efficiently determined 
by computing the tensors $\overline{\mbf{B}}_1 \in \R^{n_1 \times r_1}$, $\overline{\mbf{B}}_k \in \R^{r_{k-1}\times n_2 \times r_2},k  \in [d-1]$, $\overline{\mbf{B}}_d \in \R^{r_{d-1} \times n_d}$ in 
a  representation of the form \eqref{eq:utt} of $\mathcal{H}_{\mbf{r}}(\mathbf{X})$. 
First compute the best rank-$r_1$ approximation $\overline{\mbf{X}}^{\{1\}}$ of the matrix $\mbf{X}^{\{1\}} \in \R^{n_1 \times n_2 \cdots n_d}$ via the singular value decomposition
so that $\overline{\mbf{X}}^{\{1\}} = \overline{\mbf{B}}_1 \mbf{\Sigma}_1\mbf{V}_1^T = \overline{\mbf{B}}_1 \mbf{M}_1^{\{1\}}$
with $\overline{\mbf{B}}_1 \in \R^{n_1 \times r_1}$ and $\mbf{M}^{\{1\}}_1 \in \R^{r_1 \times n_2 \cdots n_d}$. Reshaping $\mbf{M}^{\{1\}}_1$ yields a tensor $\mbf{M}_1 \in \R^{r_1 \times n_2 \times \cdots \times n_d}$.
Next we compute the best rank-$r_2$ approximation $\overline{\mbf{M}}^{\{1,2\}}_1$ of the matricization $\mbf{M}^{\{1,2\}}_1 \in \R^{r_1 n_1 \times n_2 \cdots n_d}$ via an SVD as
$\overline{\mbf{M}}^{\{1,2\}}_1 = \overline{\mbf{B}}_2^{\{1,2\}} \mbf{\Sigma}_{2} \mbf{V}_2^T =  \overline{\mbf{B}}_2^{\{1,2\}} \mbf{M}_2^{\{1\}}$ with
$\overline{\mbf{B}}^{\{1,2\}}_2 \in \R^{r_1 n_2 \times r_2}$ and $\mbf{M}_2^{\{1\}} \in \R^{r_2 \times n_3 \cdots n_d}$ so that $\overline{\mbf{B}}_2 \in \R^{r_1 \times n_2 \times r_2}$. One iteratively continues in this way
via computing approximations of (matrix) rank $r_k$ via SVDs,
\begin{align*}
&\overline{\mbf{M}}_k^{\{1,2\}}=\overline{\mbf{B}}_{k+1}^{\{1,2\}}\mbf{\Sigma}_{k+1}\mbf{V}_{k+1}^T=\overline{\mbf{B}}_{k+1}^{\{1,2\}}\mbf{M}_{k+1}^{\{1\}}, \quad \text{ for } k = 2,\hdots, d-2, \qquad 
 \overline{\mbf{B}}_d=\mbf{M}_{d-1}.
\end{align*}
Forming the matrices $\overline{\mbf{G}}_k(\mu_k) \in \R^{r_{k-1} \times r_{k}}$, $\mu_k \in [n_k]$, $k \in [d]$, from the tensors $\overline{\mbf{B}}_k$ as in \eqref{def:Gd},
the computed rank-$\mbf{r}$ approximation is given as 
$$
\mathcal{H}_{\mbf{r}}\zag{\mbf{X}}\zag{i_1,i_2,\ldots,i_d}=\overline{\mbf{G}}_1\zag{i_1}\overline{\mbf{G}}_2\zag{i_2}\cdots \overline{\mbf{G}}_d\zag{i_d}.
$$
As shown in \cite{os11}, it satisfies the inequality
\begin{equation}\label{Hr:TT}
\|\mathbf{X} - \mathcal{H}_{\mbf{r}}(\mathbf{X}) \|_F \leq  \sqrt{d-1} \|\mathbf{X} - \mathbf{X}_{\text{BEST}}\|_F.
\end{equation}

\section{Analysis of iterative hard thresholding}
\label{Section:TIHT}

We now pass to our iterative hard thresholding algorithms.
For each tensor format (HOSVD, TT, HT), we let $\cal{H}_{\mbf{r}}$ be a corresponding low rank projection operator as described in the previous section.
Given measurements $\mbf{y} = \mathcal{A}(\mathbf{X})$ of a low rank tensor $\mathbf{X}$, 
or $\mbf{y} = \mathcal{A}(\mbf{X}) + \mbf{e}$ if the measurements are noisy,
the iterative thresholding algorithm starts with an initial guess $\mathbf{X}^0$ (often $\mathbf{X} = \mathbf{0}$) 
and performs the iterations
\begin{align}
\mbf{Y}^{j} &=  \mbf{X}^j+\mu_j \mathcal{A}^*\zag{\mbf{y}-\mathcal{A}\zag{\mbf{X}^j}} , \label{IHTB1}\\ 
\mbf{X}^{j+1} & = \mathcal{H}_{\mbf{r}}(\mbf{Y}^{j}). \label{IHTB2}
\end{align}
We analyze two variants of the algorithm which only differ by the choice of the step lengths $\mu_j$.
\begin{itemize}
\item {\bf Classical TIHT} (CTIHT) uses simply $\mu_j = 1$, see \cite{blda09} for the sparse recovery variant.
\item {\bf Normalized TIHT} (NTIHT) uses (see \cite{blda10} for the sparse vector and \cite{tawe13} for the matrix variant)
\begin{equation}\label{muj:NTIHT}
\mu_j=\frac{\norma{\mathcal{M}^j\zag{\mathcal{A}^*\zag{\mbf{y}-\mathcal{A}\zag{\mbf{X}^j}}}}_F^2}{\norma{\mathcal{A}\zag{\mathcal{M}^j\zag{\mathcal{A}^*\zag{\mbf{y}-\mathcal{A}\zag{\mbf{X}^j}}}}}_2^2}.
\end{equation}
\end{itemize}
Here, the operator $\mathcal{M}^j:\Rd \rightarrow \Rd$ depends on the choice of the tensor format 
and is computed via projections onto spaced spanned by left singular vectors of
several matricizations of $\mathbf{X}^j$. This choice of $\mu_j$ is motivated by the 
fact that in the sparse vector recovery scenario, the corresponding choice of the step length
maximally decreases the residual if the support set does not change in this iteration \cite{blda10}.

Let us describe the operator $\mathcal{M}^j:\Rd \rightarrow \Rd$ appearing in \eqref{muj:NTIHT}. 
For the sake of illustration we first specify it for the special case $d=2$, i.e., the matrix case.
Let $P_{\mbf{U}_1}^j$ and $P_{\mbf{U}_2}^j$ be the projectors onto the top $r$ left and right singular vector spaces of $\mbf{X}^j$, respectively. Then 
$\mathcal{M}^j(\mathbf{Z}) = \mbf{P}_{\mbf{U}_1}^j \mathbf{Z} \mbf{P}_{\mbf{U}_2}^j$ for a matrix $\mathbf{Z}$ so that \eqref{muj:NTIHT} yields
$$
\mu_j=\frac{\norma{\mbf{P}_{\mbf{U}_1}^j \mathcal{A}^*\zag{\mbf{y}-\mathcal{A}\zag{\mbf{X}^j}} \mbf{P}_{\mbf{U}_2}^j}_F^2}{\norma{\mathcal{A}\zag{\mbf{P}_{\mbf{U}_1}^j \mathcal{A}^*\zag{\mbf{y}-\mathcal{A}\zag{\mbf{X}^j}} \mbf{P}_{\mbf{U}_2}^j}}_2^2}.
$$
For the general tensor case, let $\mbf{U}_{i,j}$ be the left singular vectors of the matricizations ${\mbf{X}^j}^{\{i\}}$, ${\mbf{X}^j}^{\{1,\ldots,i\}}$, ${\mbf{X}^j}^{T_I(i)}$ in case of HOSVD, TT, HT decomposition with the corresponding ordered tree $T_I$, respectively. The corresponding projection operators are given as $\mbf{P}_{\mbf{U}_i}^j:=\hat{\mbf{U}}_{i,j}\hat{\mbf{U}}_{i,j}^*$, 
where $\hat{\mbf{U}}_{i,j}=\mbf{U}_{i,j}\zag{:,\uglate{r_i}}$, with $r_i=r_{{T_I}(i)}$ in the HT case. 
Then in the case of the HOSVD decomposition we define $$\mathcal{M}^j\zag{\mbf{Z}}=\mbf{Z} \times_1 \mbf{P}_{\mbf{U}_1}^j \times_2 \mbf{P}_{\mbf{U}_2}^j \times \cdots \times_d \mbf{P}_{\mbf{U}_d}^j.$$
In order to define the operator $\mathcal{M}^j$ for the TT decomposition we use the $k$-mode product defined in \eqref{def:kmode}.
The TT decomposition of  a $d$-th order tensor $\mbf{Z}$ can be written as
\begin{align*}
\mbf{Z}\zag{i_1,i_2,\ldots,i_d}&=\mbf{Z}_1(i_1)\mbf{Z}_2(i_2)\cdots \mbf{Z}_{d}(i_d) \\&= \mbf{Z}_d \times_1 \zag{\mbf{Z}_{d-1} \times_1 \zag{ \cdots \zag{\mbf{Z}_2 \times_1 \mbf{Z}_1}^{\{1,2\}} \cdots}^{\{1,2\}}}^{\{1,2\}}\zag{(i_1,i_2,\ldots,i_{d-1}),i_d}.
\end{align*}
Then the operator $\mathcal{M}^j: \Rd \rightarrow \Rd$ is defined as
$$\mathcal{M}^j\zag{\mbf{Z}}:=\tensorization\zag{\mbf{Z}_d \times_1 \mbf{P}_{\mbf{U}_{d-1}}^j\zag{\mbf{Z}_{d-1} \times_1 \mbf{P}_{\mbf{U}_{d-2}}^j \zag{ \cdots \mbf{P}_{\mbf{U}_2}^j\zag{\mbf{Z}_2 \times_1 \mbf{P}_{\mbf{U}_1}^j\mbf{Z}_1}^{\{1,2\}} \cdots}^{\{1,2\}}}^{\{1,2\}}}^{\{1,2\}},$$
where $\tensorization\zag{\mbf{x}} \in \Rd$ represents the tensorized version of a vector $\mbf{x}$, defined in \eqref{eq:tensorization}.

Using the general $k$-mode product, one can define the operator $\mathcal{M}^j$ for the 
general HT-decomposition by applying the above procedure in an analogous way.   

In the normalized version of the tensor iterative hard thresholding algorithm (NTIHT algorithm), one computes the projection operators $P_{\mbf{U}_i}^j$ in each iteration $j$. To accomplish this, the tensor decomposition has to be computed one extra time in each iteration which makes one iteration of algorithm substantially slower in comparison to the CTIHT algorithm. However, we are able to provide better convergence results  
for NTIHT than for the CTIHT algorithm.

\medskip

The available analysis of the IHT algorithm for recovery of sparse vectors \cite{blda09} and low rank matrices \cite{NIPS2010_3904} is based on the restricted isometry property (RIP). 
Therefore, we start by introducing an analog for tensors, which we call the tensor restricted isometry property (TRIP). 
Since different tensor decomposition induce different notions of tensor rank, they also induce different notions of the TRIP.

\begin{definition}[TRIP]
Let $\mathcal{A} \colon \R^{n_1 \times n_2 \times \cdots \times n_d} \rightarrow \R^m$ be a measurement map.  Then for a fixed tensor decomposition and a corresponding rank tuple $\mbf{r}$, the tensor restricted isometry constant $\delta_{\mathbf{r}}$ of $\mathcal{A}$ is the smallest quantity such that
\begin{equation}
\zag{1-\delta_{\mathbf{r}}}\norma{\mathbf{X}}_F^2 \leq \norma{\mathcal{A}\zag{\mathbf{X}}}_{2}^{2} \leq \zag{1+\delta_{\mathbf{r}}} \norma{\mathbf{X}}_F^2
\label{RIPeq}
\end{equation}
holds for all tensors $\mbf{X} \in \Rd$ of rank at most $\mathbf{r}$.
\end{definition}
We say that $\mathcal{A}$ satisfies the TRIP at rank $\mathbf{r}$ if $\delta_{\mathbf{r}}$ is bounded by a sufficiently small constant between $0$ and $1$. When referring to a particular tensor decomposition
we use the notions HOSVD-TRIP, TT-TRIP, and HT-TRIP.
Under the TRIP of the measurement operator $\mathcal{A}$, we prove partial convergence results for the two versions of the TIHT algorithm.
Depending on some number $a \in (0,1)$, the operator norm and the restricted isometry constants of $\mathcal{A}$, and on the version of TIHT, we define
\begin{align}
\delta(a) & = \left\{ \begin{array}{ll} \frac{a}{4} & \mbox{ for CTIHT},\\ \frac{a}{a+8} & \mbox{ for NTIHT}, \end{array} \right. \\
\varepsilon(a) & = \left\{ \begin{array}{ll} \frac{a^2}{17\zag{1+\sqrt{1+\delta_{3\mbf{r}}}\norma{\mathcal{A}}_{2\rightarrow 2}}^2} & \mbox{ for CTIHT}, \\
 \frac{a^2\zag{1-\delta_{3\mbf{r}}}^2}{17\zag{1-\delta_{3\mbf{r}}+\sqrt{1+\delta_{3\mbf{r}}}\norma{\mathcal{A}}_{2\rightarrow 2}}^2} &\mbox{ for NTIHT},
 \end{array} \right. \label{def:eps} \\
 b(a) & =  \left\{ \begin{array}{ll} 2\sqrt{1+\delta_{3\mbf{r}}}+ \sqrt{4\varepsilon(a)+2\varepsilon(a)^2}\norma{\mathcal{A}}_{2\rightarrow 2} & \mbox{ for CTIHT}, \\
  2\frac{\sqrt{1+\delta_{3\mbf{r}}}}{1-\delta_{3\mbf{r}}}+ \sqrt{4\varepsilon(a)+2\varepsilon(a)^2}\frac{1}{1-\delta_{3\mbf{r}}}\norma{\mathcal{A}}_{2\rightarrow 2} &
 \mbox{ for NTIHT}. \end{array} \right. \label{def:ba}
\end{align}

\begin{theorem}\label{convergenceclassTIHT}
For $a \in (0,1)$, let $\mathcal{A}:\Rd \rightarrow \R^m$ satisfy the TRIP (for a fixed tensor format) with 
\begin{equation}\label{assump:delta3r}
\delta_{3\mathbf{r}} < \delta(a)
\end{equation}
and let $\mathbf{X} \in \Rd$ be a tensor of rank at most $\mbf{r}$. 
Given measurements $\mbf{y}=\mathcal{A}\zag{\mbf{X}}$, the sequence $(\mathbf{X}^j)$ produced by CTIHT or NTIHT
converges to $\mathbf{X}$ if 
\begin{equation}\label{assump:1+vareps}
\norma{\mathbf{Y}^j-\mathbf{X}^{j+1}}_F \leq \zag{1+\varepsilon(a)}\norma{\mathbf{Y}^j-\mathbf{X}}_F \quad \mbox{ for all } j = 1,2,\hdots\;. 
\end{equation}
If the measurements are noisy, $\mbf{y}=\mathcal{A}\zag{\mbf{X}}+\mbf{e}$ for some $\mbf{e} \in \R^m$,
and if \eqref{assump:1+vareps} holds, then
\begin{equation}\label{iterating}
\norma{\mbf{X}^{j+1}-\mbf{X}}_F \leq a^j \norma{\mbf{X}^0-\mbf{X}}_F + \frac{b(a)}{1-a}\norma{\mbf{e}}_2 \quad \mbox{ for all } j=1,2,\hdots\;.
\end{equation}
Consequently, if $\mbf{e}\neq \mbf{0}$ then
after at most $j^*:=\lceil{\log_{1/a}\zag{\norma{\mbf{X}^0-\mbf{X}}_F/\norma{\mbf{e}}_2}}\rceil$ iterations, $\mathbf{X}^{j^*+1}$ 
estimates $\mbf{X}$ with accuracy 
\begin{equation}\label{approxi}
\norma{\mbf{X}^{j^*+1}-\mbf{X}}_F \leq \frac{1-a+b(a)}{1-a} \norma{\mbf{e}}_2.
\end{equation}
\end{theorem}

\begin{remark}
\begin{itemize}
\item[(a)] The unpleasant part of the theorem is that condition \eqref{assump:1+vareps} cannot be checked. It is implied by the stronger condition
$$
\norma{\mbf{Y}^j-\mbf{X}^{j+1}}_F \leq \zag{1+\varepsilon(a)} \norma{\mbf{Y}^j-\mbf{Y}_{\text{BEST}}^j}_F, 
$$
where $\mbf{Y}_{\text{BEST}}^j$ is the best rank-$\mbf{r}$ approximation of $\mbf{Y}^j$, since the best approximation
$\mbf{Y}_{\text{BEST}}^j$ is by definition a better approximation of rank $\mbf{r}$ to $\mbf{Y}^j$ than $\mbf{X}$. 
Due to \eqref{Hr:HOSVD}, \eqref{Hr:HT} and \eqref{Hr:TT}, we can only guarantee that this condition holds with $(1+\varepsilon(a))$ replaced by $C(d) \asymp \sqrt{d}$, but the proof
of the theorem only works for $(1+\varepsilon(a))$. In fact, $\varepsilon(a)$ is close to $0$ as $\|\mathcal{A}\|_{2 \to 2}$ scales like
$\sqrt{n_1\cdot n_2 \cdots n_d/m}$ for reasonable measurement maps with $\delta_{3\mbf{r}} < 1$, see below. 
However, the approximation guarantees for $\mathcal{H}_\mathbf{r}$
are only worst case estimates and one may expect that usually much better approximations are computed that satisfy \eqref{assump:1+vareps}, which only requires
a comparison of the computed approximation error of the Frobenius distance of $\mathbf{Y}^j$ to $\mathbf{X}$ rather than to $\mathbf{Y}^j_{\text{BEST}}$. In fact, during the initial iterations one is usually still far from the original tensor $\mbf{X}$ so that \eqref{assump:1+vareps} will hold.
In any case, the algorithms work in practice
so that the theorem may explain why this is the case.
\item[(b)] The corresponding theorem \cite{tawe13} for the matrix recovery case applies also to approximately low rank matrices -- not only to exactly low rank matrices -- 
and provides approximation guarantees also for this case. This is in principle also contained in our theorem by splitting 
$\mathbf{X} = \mathbf{X}_{\text{BEST}} + \mathbf{X}_{c}$ into the best rank-$\mbf{r}$ approximation and a remainder term $\mathbf{X}_c$, and writing
\[
\mathbf{y} = \mathcal{A}(\mathbf{X}) + \mathbf{e} = \mathcal{A}(\mathbf{X}_{\text{BEST}}) + \mathcal{A}(\mathbf{X}_c) + \mathbf{e} =  \mathcal{A}(\mathbf{X}_{\text{BEST}}) + \widetilde{\mathbf{e}},
\]
where $\widetilde{\mathbf{e}} = \mathcal{A}(\mathbf{X}_c) + \mathbf{e}$. Then the theorem may be applied to $\widetilde{\mathbf{e}}$ instead of $\mathbf{e}$ and 
\eqref{approxi} gives the error estimate 
\[
\norma{\mbf{X}^{j^*+1}-\mbf{X}_{\text{BEST}}}_F \leq \frac{1-a+b(a)}{1-a} \| \mathcal{A}(\mathbf{X}_c) + \mathbf{e} \|_2.
\]
In the matrix case, the right hand side can be further estimated by a sum of three terms (exploiting the restricted isometry property), 
one of them being the nuclear norm of $\mathbf{X}_c$, i.e., the error of best rank-$r$ approximation in the nuclear norm.
In the tensor case, a similar estimate is problematic, in particular, the analogue of the nuclear norm approximation error is unclear.
\item[(c)] In \cite{rascst15} local convergence of a class of algorithms including iterative hard thresholding has been shown, i.e., once an iterate $\mbf{X}^j$ is close enough to the original $\mbf{X}$ then convergence is guaranteed. (The theorem in \cite{rascst15} requires $\mathcal{H}_{\mathbf{r}}$ to be a retraction on the manifold of rank-$\mathbf{r}$ tensors which is in fact true \cite{kestva14, steinlechner2015riemannian}.) Unfortunately, the distance to $\mbf{X}$ which ensures local convergence depends on the curvature at $\mbf{X}$ 
of the manifold of rank-$\mbf{r}$ tensors and is therefore unknown a-priori.
Nevertheless, together with Theorem~\ref{convergenceclassTIHT}, we conclude that the initial iterations 
decrease the distance to the original 
$\mbf{X}$ (if the initial distance is large enough), and if the iterates become sufficiently close to $\mbf{X}$,
then we are guaranteed convergence. 
The theoretical question remains about the ``intermediate phase'', i.e., whether the iterates always 
do come close enough to $\mbf{X}$ at some point.
\item[(d)] In \cite{heinsc15}, Hedge, Indyk, and Schmidt find a way to deal with approximate projections onto model sets satisfying
a relation like \eqref{Hr:approx} within iterative hard thresholding algorithms by working with a second approximate 
projection $\widetilde{\mathcal{H}}_{\mathbf{r}}$ satisfying a so-called head approximation
guarantee of the form $\|\widetilde{\mathcal{H}}_{\mathbf{r}}(\mathbf{X}) \|_F \geq c \| \mathbf{X} \|_F$ for some constant $c > 0$.
Unfortunately, we were only able to find such head approximations for the tensor formats at hand
with constants $c$ that scale unfavorably with $r$ and the dimensions $n_1,\hdots,n_d$, so that in the end one arrives only at trivial estimates for the minimal number of required measurements. 
\end{itemize}
\end{remark}

\begin{proof}[Proof of Theorem~\ref{convergenceclassTIHT}]
We proceed similar to the corresponding proofs for the sparse vector \cite{fora13} and matrix recovery case \cite{tawe13}.
The fact that \eqref{assump:1+vareps} only holds with an additional $\varepsilon = \varepsilon(a)$ requires extra care.

It follows from assumption \eqref{assump:1+vareps} that
\begin{align}
\zag{1+\varepsilon}^2
\norma{\mathbf{Y}^j-\mathbf{X}}_F^2 & \geq  \norma{\mathbf{Y}^j-\mathbf{X}^{j+1}}_F^2  =  \norma{\mathbf{Y}^j-\mathbf{X}+\mathbf{X}-\mathbf{X}^{j+1}}_F^2  \nonumber\\
& = \norma{\mathbf{Y}^j-\mathbf{X}}_F^2+ \norma{\mathbf{X}-\mathbf{X}^{j+1}}_F^2 + 2\interval{\mathbf{Y}^j-\mathbf{X}, \mathbf{X}-\mathbf{X}^{j+1}}.  \label{Thm1:proof:assumvareps}
\end{align}
Subtracting $\norma{\mathbf{Y}^j-\mathbf{X}}_F^2$ and using $\mathbf{Y}^j=\mathbf{X}^j-\mu_j\mathcal{A}^*\zag{\mathcal{A}\zag{\mathbf{X}^j}-\mathbf{y}}=\mathbf{X}^j-\mu_j\mathcal{A}^*\mathcal{A}\zag{\mathbf{X}^j-\mathbf{X}}+\mu_j\mathcal{A}^*\mbf{e}$ gives 
\begin{align}\label{ineq9}
\norma{\mathbf{X}^{j+1}-\mathbf{X}}_F^2 \leq & \, 2 \interval{\mathbf{Y}^{j}-\mathbf{X}, \mathbf{X}^{j+1}-\mathbf{X}}+ \zag{2\varepsilon + \varepsilon^2} \norma{\mathbf{Y}^j-\mathbf{X}}_F^2 \nonumber \\
 = &\, 2 \interval{\mathbf{X}^j-\mathbf{X}, \mathbf{X}^{j+1}-\mathbf{X}}-2\mu_j\interval{\mathcal{A}^*\mathcal{A}\zag{\mathbf{X}^j-\mathbf{X}}, \mathbf{X}^{j+1}-\mathbf{X}}\nonumber \\
&+ 2\mu_j\interval{\mathcal{A}^*\mbf{e},\mbf{X}^{j+1}-\mbf{X}}+\zag{2\varepsilon + \varepsilon^2} \norma{\mathbf{Y}^j-\mathbf{X}}_F^2 \nonumber \\
= & \, 2 \interval{\mathbf{X}^j-\mathbf{X}, \mathbf{X}^{j+1}-\mathbf{X}}-2\mu_j\interval{\mathcal{A}\zag{\mathbf{X}^j-\mathbf{X}}, \mathcal{A}\zag{\mathbf{X}^{j+1}-\mathbf{X}}} \nonumber \\
&+ 2\mu_j\interval{\mbf{e},\mathcal{A}\zag{\mbf{X}^{j+1}-\mbf{X}}}+\zag{2\varepsilon + \varepsilon^2} \norma{\mathbf{Y}^j-\mathbf{X}}_F^2 \nonumber \\
 \leq & \, 2 \interval{\mathbf{X}^j-\mathbf{X}, \mathbf{X}^{j+1}-\mathbf{X}}-2\mu_j\interval{\mathcal{A}\zag{\mathbf{X}^j-\mathbf{X}}, \mathcal{A}\zag{\mathbf{X}^{j+1}-\mathbf{X}}} \nonumber \\
&+ 2\mu_j\sqrt{1+\delta_{3\mbf{r}}}\norma{\mbf{X}^{j+1}-\mbf{X}}_F\norma{\mbf{e}}_2 +\zag{2\varepsilon + \varepsilon^2} \norma{\mathbf{Y}^j-\mathbf{X}}_F^2,
\end{align}
where the last inequality is valid since $\rank\zag{\mbf{X}^{j+1}-\mbf{X}}\leq 2\mbf{r}\leq 3\mbf{r}$ so that $$\interval{\mbf{e},\mathcal{A}\zag{\mbf{X}^{j+1}-\mbf{X}}} \leq \norma{\mathcal{A}\zag{\mbf{X}^{j+1}-\mbf{X}}}_2\norma{\mbf{e}}_2 \leq \sqrt{1+\delta_{3\mbf{r}}}\norma{\mbf{X}^{j+1}-\mbf{X}}_F\norma{\mbf{e}}_2.$$

Now let $U^j$ be the subspace of $\Rd$ spanned by the tensors $\mbf{X}$, $\mbf{X}^j$, and $\mbf{X}^{j+1}$ and denote by $\mathcal{Q}^j: \Rd \rightarrow U^j$ 
the orthogonal projection onto $U^j$. Then $\mathcal{Q}^j\zag{\mbf{X}}=\mbf{X}$, $\mathcal{Q}^j\zag{\mbf{X}^j}=\mbf{X}^j$, and  $\mathcal{Q}^j\zag{\mbf{X}^{j+1}}=\mbf{X}^{j+1}$. Clearly, the rank of $\mathcal{Q}^j\zag{\mbf{Y}}$ is at most $3\mbf{r}$ for all $\mbf{Y} \in \Rd$. Further, we define the operator $\mathcal{A}_{\mbf{Q}}^j : \Rd \rightarrow \Rd$ by $\mathcal{A}_{\mbf{Q}}^j\zag{\mbf{Z}}=\mathcal{A}\zag{\mathcal{Q}^j\zag{\mbf{Z}}}$ for $\mbf{Z} \in \Rd$.

With these notions the estimate
\eqref{ineq9} is continued as
\begin{align}\label{Xj+1-Xr}
\norma{\mathbf{X}^{j+1}-\mathbf{X}}_F^2
 \leq &\, 2 \interval{\mathbf{X}^j-\mathbf{X}, \mathbf{X}^{j+1}-\mathbf{X}}-2\mu_j\interval{\mathcal{A}_{\mathbf{Q}}^j\zag{\mathbf{X}^{j}-\mathbf{X}}, \mathcal{A}_{\mathbf{Q}}^j\zag{\mathbf{X}^{j+1}-\mathbf{X}}} \nonumber \\
&+ 2\mu_j\sqrt{1+\delta_{3\mbf{r}}}\norma{\mbf{X}^{j+1}-\mbf{X}}_F\norma{\mbf{e}}_2 + \zag{2\varepsilon + \varepsilon^2}\norma{\mathbf{Y}^j-\mathbf{X}}_F^2 \nonumber \\
= & \, 2 \interval{\mathbf{X}^j-\mathbf{X}, \zag{\mathbf{X}^{j+1}-\mathbf{X}}-\mu_j\mathcal{A}_{\mathbf{Q}}^{j*}\mathcal{A}_{\mathbf{Q}}^j\zag{\mathbf{X}^{j+1}-\mathbf{X}}} \nonumber \\
& + 2\mu_j\sqrt{1+\delta_{3\mbf{r}}}\norma{\mbf{X}^{j+1}-\mbf{X}}_F\norma{\mbf{e}}_2+ \zag{2\varepsilon + \varepsilon^2}\norma{\mathbf{Y}^j-\mathbf{X}}_F^2 \nonumber \\
= & \, 2 \interval{\mathbf{X}^j-\mathbf{X}, \zag{\mathbf{I}-\mu_j\mathcal{A}_{\mathbf{Q}}^{j*}\mathcal{A}_{\mathbf{Q}}^j}\zag{\mathbf{X}^{j+1}-\mathbf{X}}}\nonumber \\
& + 2\mu_j\sqrt{1+\delta_{3\mbf{r}}}\norma{\mbf{X}^{j+1}-\mbf{X}}_F\norma{\mbf{e}}_2 + \zag{2\varepsilon + \varepsilon^2}\norma{\mathbf{Y}^j-\mathbf{X}}_F^2 \nonumber \\
 \leq &\, 2 \norma{\mathbf{I}-\mu_j\mathcal{A}_{\mathbf{Q}}^{j*}\mathcal{A}_{\mathbf{Q}}^j}_{2\rightarrow 2} \norma{\mathbf{X}^j-\mathbf{X}}_F \norma{\mathbf{X}^{j+1}-\mathbf{X}}_F  \nonumber \\
& + 2\mu_j\sqrt{1+\delta_{3\mbf{r}}}\norma{\mbf{X}^{j+1}-\mbf{X}}_F\norma{\mbf{e}}_2 + \zag{2\varepsilon + \varepsilon^2}\norma{\mathbf{Y}^j-\mathbf{X}}_F^2. 
\end{align}
The last term can be bounded by
\begin{align}
\norma{\mathbf{Y}^j-\mathbf{X}}_F  &= \norma{\mathbf{X}^j-\mu_j\mathcal{A}^*\mathcal{A}\zag{\mathbf{X}^j-\mathbf{X}}+\mu_j\mathcal{A}^*\mbf{e}-\mbf{X}}_F \nonumber \\
&= \norma{\zag{\mathbf{X}^j-\mathbf{X}}-\mu_j\mathcal{A}^*\mathcal{A}\zag{\mathbf{X}^j- \mathbf{X}}+\mu_j\mathcal{A}^*\mbf{e}}_F \nonumber \\
& = \norma{\zag{\mathbf{I}-\mu_j\mathcal{A}^*\mathcal{A}}\zag{\mathbf{X}^j-\mathbf{X}}+\mu_j\mathcal{A}^*\mbf{e}}_F = \norma{\zag{\mathbf{I}-\mu_j\mathcal{A}^*\mathcal{A}_{\mathbf{Q}}^j}\zag{\mathbf{X}^j-\mathbf{X}}+\mu_j\mathcal{A}^*\mbf{e}}_F  \nonumber \\
& \leq \norma{\mathbf{I}-\mu_j\mathcal{A}^*\mathcal{A}_{\mathbf{Q}}^j}_{2\rightarrow 2} \norma{\mathbf{X}^j-\mathbf{X}}_F + \mu_j\norma{\mathcal{A}^*\mbf{e}}_F  \nonumber \\
&\leq  \zag{1+\mu_j\norma{\mathcal{A}}_{2\rightarrow 2}\norma{\mathcal{A}_{\mbf{Q}}^j}_{2\rightarrow 2}} \norma{\mathbf{X}^j-\mathbf{X}}_F +\mu_j\norma{\mathcal{A}}_{2\rightarrow 2}\norma{\mbf{e}}_2 \nonumber \\
& \leq \zag{1+\mu_j\sqrt{1+\delta_{3\mbf{r}}}\norma{\mathcal{A}}_{2\rightarrow 2}} \norma{\mathbf{X}^j-\mathbf{X}}_F+ \mu_j\norma{\mathcal{A}}_{2\rightarrow 2}\norma{\mbf{e}}_2 . \label{est:aux1}
\end{align}
Using that $\zag{u+v}^2\leq 2\zag{u^2+v^2}$ for all $u,v \in \R$, we obtain the estimate
\begin{equation}\label{Y-Xr}
\norma{\mathbf{Y}^j-\mathbf{X}}_F^2 \leq 2\zag{1+\mu_j\sqrt{1+\delta_{3\mbf{r}}}\norma{\mathcal{A}}_{2 \rightarrow 2}}^2 \norma{\mathbf{X}^j-\mathbf{X}}_F^2+ 2\mu_j^2\norma{\mathcal{A}}_{2 \rightarrow 2}^2\norma{\mbf{e}}_2^2.
\end{equation}
Combining inequalities \eqref{Xj+1-Xr} and \eqref{Y-Xr} yields
\begin{align}
& \norma{\mathbf{X}^{j+1}-\mathbf{X}}_F^2  \leq  2 \norma{\mathbf{I}-\mu_j\mathcal{A}_{\mathbf{Q}}^{j*}\mathcal{A}_{\mathbf{Q}}^j}_{2\rightarrow 2} \norma{\mathbf{X}^j-\mathbf{X}}_F \norma{\mathbf{X}^{j+1}-\mathbf{X}}_F \nonumber\\
& + 2\mu_j\sqrt{1+\delta_{3\mbf{r}}}\norma{\mbf{X}^{j+1}-\mbf{X}}_F\norma{\mbf{e}}_2 
 + 2\zag{2\varepsilon+\varepsilon^2} \zag{1+\mu_j\sqrt{1+\delta_{3\mbf{r}}}\norma{\mathcal{A}}_{2\rightarrow 2}}^2 \norma{\mathbf{X}^j-\mathbf{X}}_F^2 \nonumber \\ &+ 2\zag{2\varepsilon+\varepsilon^2}\mu_j^2\norma{\mathcal{A}}_{2\rightarrow 2}^2\norma{\mbf{e}}_2^2.
\end{align}
This implies that there exist $\alpha, \beta, \gamma \in \uglate{0,1}$ such that $\alpha+\beta+\gamma \leq 1$ and
\begin{align}
\zag{1-\alpha-\beta-\gamma} \norma{\mathbf{X}^{j+1}-\mathbf{X}}_F^2 & \leq 2\norma{\mathbf{I}-\mu_j\mathcal{A}_{\mathbf{Q}}^{j*}\mathcal{A}_{\mathbf{Q}}^j}_{2 \rightarrow 2} \norma{\mathbf{X}^j-\mathbf{X}}_F \norma{\mathbf{X}^{j+1}-\mathbf{X}}_F   \label{alfajed} \\
\alpha \norma{\mathbf{X}^{j+1}-\mathbf{X}}_F^2 & \leq 2\mu_j\sqrt{1+\delta_{3\mbf{r}}}\norma{\mathbf{X}^{j+1}-\mathbf{X}}_F\norma{\mbf{e}}_2 \label{alpha}\\
\beta \norma{\mathbf{X}^{j+1}-\mathbf{X}}_F^2 &\leq 2\zag{2\varepsilon+\varepsilon^2} \zag{1+\mu_j\sqrt{1+\delta_{3\mbf{r}}}\norma{\mathcal{A}}_{2 \rightarrow 2}}^2 \norma{\mathbf{X}^j-\mathbf{X}}_F^2, \label{beta} \\
\gamma \norma{\mathbf{X}^{j+1}-\mathbf{X}}_F^2 & \leq 2\zag{2\varepsilon+\varepsilon^2}\mu_j^2\norma{\mathcal{A}}_{2 \rightarrow 2}^2\norma{\mbf{e}}_2^2.\label{gamma}
\end{align}
Canceling one power of $\norma{\mathbf{X}^{j+1}-\mathbf{X}}_F$ in inequalities \eqref{alfajed} and \eqref{alpha},
taking the square root of the inequalities \eqref{beta} and \eqref{gamma}, and summation of all resulting inequalities 
yields
\begin{align}
&\norma{\mathbf{X}^{j+1}-\mathbf{X}}_F  \nonumber\\& \quad \leq \f\zag{\beta,\gamma} \zag{2 \norma{\mathbf{I}-\mu_j\mathcal{A}_{\mathbf{Q}}^{j*}\mathcal{A}_{\mathbf{Q}}^j}_{2 \rightarrow 2}+\sqrt{4\varepsilon+2\varepsilon^2}\zag{1+\mu_j\sqrt{1+\delta_{3\mbf{r}}}\norma{\mathcal{A}}_{2 \rightarrow 2}}} \norma{\mathbf{X}^j-\mathbf{X}}_F \nonumber \\
& \quad + \f\zag{\beta,\gamma} \zag{ 2\mu_j\sqrt{1+\delta_{3\mbf{r}}}+ \sqrt{4\varepsilon+2\varepsilon^2} \mu_j\norma{\mathcal{A}}_{2 \rightarrow 2}} \norma{\mbf{e}}_2 \label{ineq:Xj+1-Xr}
\end{align}
with $\f\zag{\beta,\gamma}= (1-\beta+\sqrt{\beta}-\gamma+\sqrt{\gamma})^{-1}$.
Notice that $f$ is positive and strictly less than $1$ on $\uglate{0,1} \times \uglate{0,1}$ and will therefore be omitted in the following.

Let us now specialize to CTIHT where $\mu_j = 1$. Since $\mathcal{A}_{\mathbf{Q}}^j$ is the restriction of 
$\mathcal{A}$ to the space $U^j$ which contains
only tensors of rank at most $3 \mathbf{r}$, 
we have (with $\mathbf{I}$ denoting the identity operator on $U^j$)
\begin{align*}
\| \mathbf{I}-\mu_j\mathcal{A}_{\mathbf{Q}}^{j*}\mathcal{A}_{\mathbf{Q}}^j\|_{2 \rightarrow 2} & = 
\| \mathbf{I}-\ \mathcal{A}_{\mathbf{Q}}^{j*}\mathcal{A}_{\mathbf{Q}}^j\|_{2 \rightarrow 2} = \sup_{\mathbf{X} \in U_j : \|\mathbf{X}\|_F = 1} | \|\mbf{X}\|_F^2 - \|\mathcal{A}(\mbf{X}) \|_2^2 | \\
& \leq  \sup_{\mbf{X} : \rank(\mathbf{X}) \leq 3 \mbf{r}, \|\mbf{X}\|_F = 1} | \|\mbf{X}\|_F^2 - \|\mathcal{A}(\mbf{X}) \|_2^2 | = \delta_{3 \mathbf{r}}.
\end{align*}
Plugging $\mu_j=1$ and above estimate into \eqref{ineq:Xj+1-Xr} yields
\begin{align}
\norma{\mathbf{X}^{j+1}-\mathbf{X}}_F  \leq & \,  \zag{2 \norma{\mathbf{I}-\mathcal{A}_{\mathbf{Q}}^{j*}\mathcal{A}_{\mathbf{Q}}^j}_{2 \rightarrow 2}+\sqrt{4\varepsilon+2\varepsilon^2}\zag{1+\sqrt{1+\delta_{3\mbf{r}}}\norma{\mathcal{A}}_{2 \rightarrow 2}}} \norma{\mathbf{X}^j-\mathbf{X}}_F \nonumber \\
& + \zag{ 2\sqrt{1+\delta_{3\mbf{r}}}+ \sqrt{4\varepsilon+2\varepsilon^2} \norma{\mathcal{A}}_{2 \rightarrow 2}} \norma{\mbf{e}}_2 \nonumber \\ 
\leq & \, \zag{2 
\delta_{3\mbf{r}}+\sqrt{4\varepsilon+2\varepsilon^2}\zag{1+\sqrt{1+\delta_{3\mbf{r}}}\norma{\mathcal{A}}_{2 \rightarrow 2}}} \norma{\mathbf{X}^j-\mathbf{X}}_F \nonumber \\
& +\zag{ 2\sqrt{1+\delta_{3\mbf{r}}}+ \sqrt{4\varepsilon+2\varepsilon^2} \norma{\mathcal{A}}_{2 \rightarrow 2}} \norma{\mbf{e}}_2. \nonumber
\end{align}
Setting $\kappa := 1 + \sqrt{1+\delta_{3\mathbf{r}}} \|\mathcal{A}\|_{2 \to 2} > 1$, 
the bound $\delta_{3 \mathbf{r}} \leq a/4$ with $a < 1$ and the definition of $\varepsilon = \varepsilon(a)$ in \eqref{def:eps} yield
\[
2 \delta_{3\mbf{r}}+\sqrt{4\varepsilon+2\varepsilon^2}\zag{1+\sqrt{1+\delta_{3\mbf{r}}}\norma{\mathcal{A}}_{2 \rightarrow 2}}
\leq \frac{a}{2} + \sqrt{\frac{4 a^2}{17 \kappa^2} + \frac{2 a^4}{17^2 \kappa^4}} \kappa \leq a \left(\frac{1}{2} + \sqrt{\frac{4}{17} + \frac{2}{17^2}}\right) < a.
\]
Thus, with the definition \eqref{def:ba} of $b=b(a)$ for CTIHT we obtain
\begin{equation*}
\norma{\mbf{X}^{j+1}-\mbf{X}}_F \leq a \norma{\mbf{X}^{j}-\mbf{X}}_F + b \norma{\mbf{e}}_2.
\end{equation*}
Iterating this inequality leads to \eqref{iterating}, which implies a recovery accuracy of $\norma{\mbf{X}^{j+1}-\mbf{X}}_F \leq \frac{1-a+b}{1-a}\norma{\mbf{e}}_2$ 
if $a^{j}\norma{\mbf{X}^0-\mbf{X}}_F \leq \norma{\mbf{e}}_2$. Hence, if $\mbf{e}\neq \mbf{0}$ then after $j^*:=\lceil{\log_{1/a}\zag{\norma{\mbf{X}^0-\mbf{X}}_F/\norma{\mbf{e}}_2}}\rceil$ iterations, \eqref{approxi} holds. 

Let us now consider the variant NTIHT. Since the image of the operator $\mathcal{M}^{j}$ is contained in the set of rank-$\mathbf{r}$ tensors, 
the tensor restricted isometry property yields 
\begin{equation}\label{boundmuj}
\frac{1}{1+\delta_{\mathbf{r}}} \leq \mu_j=\frac{\norma{\mathcal{M}^j\zag{\mathcal{A}^*\zag{\mathbf{y}-\mathcal{A}\zag{\mathbf{X}^j}}}}_F^2}{\norma{\mathcal{A}\zag{\mathcal{M}^j\zag{\mathcal{A}^*\zag{\mathbf{y}-\mathcal{A}\zag{\mathbf{X}^j}}}}}_2^2} \leq \frac{1}{1-\delta_{\mathbf{r}}}.
\end{equation}
Since $\mathcal{Q}^j$ maps onto rank-3$\mbf{r}$ tensors, the TRIP implies that every eigenvalue of $\mathcal{A}_{\mathbf{Q}}^{j*}\mathcal{A}_{\mathbf{Q}}^j$
is contained in the interval $[1-\delta_{3\mathbf{r}}, 1+ \delta_{3\mathbf{r}}]$. Therefore, every eigenvalue of 
$\mathbf{I} - \mu_j \mathcal{A}_{\mathbf{Q}}^{j*}\mathcal{A}_{\mathbf{Q}}^j$ is contained in 
$[1- \frac{1+\delta_{3\mathbf{r}}}{1-\delta_{\mathbf{r}}}, 1- \frac{1-\delta_{3\mathbf{r}}}{1+\delta_{\mathbf{r}}}]$. The magnitude of the lower end point 
is greater than that of the upper end point, giving the operator norm bound
\begin{equation}
\norma{\mathbf{I}-\mu_j\mathcal{A}_{\mathbf{Q}}^{j*}\mathcal{A}_{\mathbf{Q}}^j}_{2 \rightarrow 2} \leq \frac{1+\delta_{3\mathbf{r}}}{1-\delta_{\mathbf{r}}}-1 \leq
\frac{1+\delta_{3\mathbf{r}}}{1-\delta_{3\mathbf{r}}} - 1.
\end{equation}
Hence, plugging the upper bound on $\mu_j$ in \eqref{boundmuj} and the above inequality into  \eqref{ineq:Xj+1-Xr} leads to
\begin{align}
\norma{\mathbf{X}^{j+1}-\mathbf{X}}_F  \leq & \,  \zag{2 \norma{\mathbf{I}-\mu_j\mathcal{A}_{\mathbf{Q}}^{j*}\mathcal{A}_{\mathbf{Q}}^j}_{2 \rightarrow 2}+\sqrt{4\varepsilon+2\varepsilon^2}\zag{1+\mu_j\sqrt{1+\delta_{3\mbf{r}}}\norma{\mathcal{A}}_{2 \rightarrow 2}}} \norma{\mathbf{X}^j-\mathbf{X}}_F \nonumber \\
& + \zag{ 2\mu_j\sqrt{1+\delta_{3\mbf{r}}}+ \sqrt{4\varepsilon+2\varepsilon^2} \mu_j\norma{\mathcal{A}}_{2 \rightarrow 2}} \norma{\mbf{e}}_2 \nonumber \\ 
\leq & \, \zag{2 
\zag{\frac{1+\delta_{3\mbf{r}}}{1-\delta_{3\mbf{r}}}-1}+\sqrt{4\varepsilon+2\varepsilon^2}\zag{1+\frac{\sqrt{1+\delta_{3\mbf{r}}}}{1-\delta_{3\mbf{r}}}\norma{\mathcal{A}}_{2 \rightarrow 2}}} \norma{\mathbf{X}^j-\mathbf{X}}_F \nonumber \\
& +\zag{ 2\frac{\sqrt{1+\delta_{3\mbf{r}}}}{1-\delta_{3\mbf{r}}}+ \frac{\sqrt{4\varepsilon+2\varepsilon^2}}{1-\delta_{3\mbf{r}}} \norma{\mathcal{A}}_{2 \rightarrow 2}} \norma{\mbf{e}}_2. \nonumber
\end{align}
Setting $\nu:= 1 + \frac{\sqrt{1+\delta_{3 \mathbf{r}}}}{1-\delta_{3\mathbf{r}}} \|\mathcal{A}\|_{2 \to 2} \geq 1$, using $\delta_{3 \mathbf{r}} \leq a/(a+8)$
and the definition \eqref{def:eps} of $\varepsilon = \varepsilon(a) = a^2/(17 \nu^2)$, 
gives
\[
2 \zag{\frac{1+\delta_{3\mbf{r}}}{1-\delta_{3\mbf{r}}}-1}+\sqrt{4\varepsilon+2\varepsilon^2}\zag{1+\frac{\sqrt{1+\delta_{3\mbf{r}}}}{1-\delta_{3\mbf{r}}}\norma{\mathcal{A}}_{2 \rightarrow 2}} \leq \frac{a}{2} + \nu \sqrt{\frac{4 a^2}{17\nu^2} + \frac{2 a^2}{17^2 \nu^4}}  < a
\]
so that with the definition of $b$ in \eqref{def:ba} we arrive at
\begin{equation*}
\norma{\mbf{X}^{j+1}-\mbf{X}}_F \leq a \norma{\mbf{X}^{j}-\mbf{X}}_F + b \norma{\mbf{e}}_2.
\end{equation*}
The proof is concluded in the same way as for CTIHT.
\end{proof}

\begin{remark}
For the noiseless scenario where $\norma{\mbf{e}}_2=0$, one may work with a slightly improved definition of $\varepsilon(a)$.
In fact, \eqref{est:aux1} implies then
\begin{equation*}
\norma{\mbf{Y}^j-\mbf{X}}_F \leq \zag{1+\mu_j\sqrt{1+\delta_{3\mbf{r}}}\norma{\mathcal{A}}_{2 \rightarrow 2}} \norma{\mbf{X}^j-\mbf{X}}_F.
\end{equation*}
Following the proof in the same way as above, one finds that the constant $17$ in the definition \eqref{def:eps} of $\varepsilon(a)$ can be improved to $9$.
\end{remark}

\section{Tensor RIP} \label{Section:TRIP}
\label{Sec:TRIP}

Now that we have shown
a (partial) convergence result for the TIHT algorithm based on the TRIP, the question arises which measurement maps
satisfy the TRIP under suitable conditions on the number of measurements in terms of the rank $\mathbf{r}$, the order $d$ and the dimensions $n_1,\hdots,n_d$.
As common in compressive sensing and low rank recovery, we study this question for random measurement maps.
We concentrate first on subgaussian measurement maps and consider maps based on partial random Fourier transform afterwards.

A random variable $X$ is called $L$-subgaussian if there exists a constant $L>0$ such that 
\[
\ocek{\exp\zag{tX}}\leq\exp\zag{L^2t^2/2}
\]
 holds for all $t\in \R$.
We call $\mathcal{A}:\R^{n_1 \times n_2 \times \cdots \times n_d} \rightarrow \R^m$ an $L$-subgaussian measurement ensemble 
if all elements of $\mathcal{A}$, interpreted as a tensor in $\R^{n_1 \times n_2 \times \cdots \times n_d \times m}$, are independent mean-zero, variance one, $L$-subgaussian variables. Gaussian and Bernoulli random measurement ensembles where the entries are
standard normal distributed random variables and Rademacher $\pm 1$ variables (i.e., taking the values $+1$ and $-1$ with equal probability), 
respectively, are special cases of $1$-subgaussian measurement ensembles.

\begin{theorem}\label{glavni}
Fix one of the tensor formats HOSVD, TT, HT (with decomposition tree $T_I$).
For $\delta, \varepsilon \in \zag{0,1}$, a random draw of an $L$-subgaussian measurement ensemble $\mathcal{A} : \R^{n_1 \times n_2 \times \cdots \times n_d} \rightarrow \R^m$ satisfies $\delta_{\mathbf{r}}\leq \delta$ with probability at least $1-\varepsilon$ provided that
\begin{align*}
\text{HOSVD:}& \quad m\geq C_1\delta^{-2}\max \skup{\zag{r^d+dnr} \log \zag{d}, \log\zag{\varepsilon^{-1}}}, \\
\text{TT \& HT:}& \quad m \geq C_2\delta^{-2}\max \skup{\zag{(d-1)r^3+dnr}\log\zag{dr}, \log\zag{\varepsilon^{-1}}}, 
\end{align*}
where  $n=\max\skup{n_i: i \in \uglate{d}}$, $r=\max\skup{r_t: t \in T_I}$.
The constants $C_1,C_2,C_3>0$ only depend on the subgaussian parameter $L$.
\end{theorem}

One may generalize the above theorem to situations where it is no longer required that all entries of the tensor $\mathcal{A}$ are independent, but only
that the sensing tensors $\mbf{A}_i \in \Rd$, $i=1,\hdots,m$, are independent. 
We refer to \cite{DBLP:journals/corr/Dirksen14} for details, in particular to Corollary~5.4 and Example~5.8.
Furthermore, we note that the term $dnr$ in all bounds for $m$ may be refined to $\sum_{i=1}^d n_i r_i$.

The proof of Theorem~\ref{glavni} uses $\varepsilon$-nets and covering numbers, see e.g.\ \cite{ve12} for more background on this topic.

\begin{definition} 
A set $\mathcal{N}_{\varepsilon}^{\boldsymbol{\mathcal{X}}} \subset \boldsymbol{\mathcal{X}}$, where $\boldsymbol{\mathcal{X}}$ is a subset of a normed space, is called an $\varepsilon$-net of $\boldsymbol{\mathcal{X}}$ with respect to the norm $\norma{\cdot}$ if for each $v \in \boldsymbol{\mathcal{X}}$, there exists $v_0 \in \mathcal{N}_{\varepsilon}^{\boldsymbol{\mathcal{X}}}$ with $\norma{v_0-v}\leq \varepsilon$. The minimal cardinality of an $\varepsilon$-net of $\boldsymbol{\mathcal{X}}$ with respect to the norm $\norma{\cdot}$ is denoted by $\mathcal{N}\zag{\boldsymbol{\mathcal{X}},\norma{\cdot},\varepsilon}$ and is called the covering number of $\boldsymbol{\mathcal{X}}$ (at scale $\varepsilon$).
\end{definition}
The following well-known result will be used frequently in the following. 
\begin{lemma}[\cite{ve12}]\label{unitball}
Let $\boldsymbol{\mathcal{X}}$ be a subset of a vector space of real dimension $k$ with norm $\|\cdot\|$, and let $0 < \varepsilon<1$. 
Then there exists an $\varepsilon$-net $\mathcal{N}_{\varepsilon}^{\boldsymbol{\mathcal{X}}}$ 
satisfying $\mathcal{N}_{\varepsilon}^{\boldsymbol{\mathcal{X}}} \subset \boldsymbol{\mathcal{X}}$ and
$\aps{\mathcal{N}_{\varepsilon}^{\boldsymbol{\mathcal{X}}}} \leq \frac{\Vol\zag{\boldsymbol{\mathcal{X}}+\frac{\varepsilon}{2}\boldsymbol{\mathcal{B}}}}{\Vol\zag{\frac{\varepsilon}{2}\boldsymbol{\mathcal{B}}}}$,
where $\frac{\varepsilon}{2}\boldsymbol{\mathcal{B}}$ is an $\varepsilon/2$ ball with respect to the norm $\norma{\cdot}$ and $\boldsymbol{\mathcal{X}}+\frac{\varepsilon}{2}\boldsymbol{\mathcal{B}}=\skup{x+y: x \in \mathcal{N}_{\varepsilon}^{\boldsymbol{\mathcal{X}}}, y \in \frac{\varepsilon}{2}\boldsymbol{\mathcal{B}}}$. Specifically, if $\boldsymbol{\mathcal{X}}$ is a subset of the $\|\cdot\|$-unit ball then $\boldsymbol{\mathcal{X}}+\frac{\varepsilon}{2}\boldsymbol{\mathcal{B}}$ is contained in the $\zag{1+\frac{\varepsilon}{2}}$-ball and thus 
\begin{equation*}
\aps{\mathcal{N}_{\varepsilon}^{\boldsymbol{\mathcal{X}}}} \leq \frac{\zag{1+\varepsilon/2}^k}{\zag{\varepsilon/2}^k}= \zag{1+\frac{2}{\varepsilon}}^k< \zag{3/\varepsilon}^k.
\end{equation*}
\end{lemma}

It is crucial for the proof of Theorem~\ref{glavni} to estimate the covering numbers of the set of unit Frobenius norm rank-$\mbf{r}$ tensors with respect
to the different tensor formats. We start with the HOSVD.
 
\begin{lemma}[Covering numbers related to HOSVD] \label{covering}
The covering numbers of
\begin{equation*}
\boldsymbol{\mathcal{S}}_{\mathbf{r}}=\skup{\mathbf{X} \in \R^{n_1 \times n_2 \times \cdots \times n_d} : \rank_{\operatorname{HOSVD}}\zag{\mathbf{X}}\leq \mathbf{r}, \norma{\mathbf{X}}_F=1}
\end{equation*}
with respect to the Frobenius norm satisfy
\begin{equation}
\mathcal{N}\zag{\boldsymbol{\mathcal{S}}_{\mathbf{r}},\norma{\cdot}_F,\varepsilon}\leq \zag{3\zag{d+1}/\varepsilon}^{r_1r_2\cdots r_d+\sum_{i=1}^{d}n_ir_i}.
\label{eqcovering}
\end{equation} 
\end{lemma}

\begin{proof}
The proof follows a similar strategy as the one of \cite[Lemma~3.1]{capl11}.
The HOSVD decomposition $\mathbf{X}=\mathbf{S} \times_1 \mathbf{U}_1 \times_2 \mathbf{U}_2 \times \cdots \times_d \mathbf{U}_d$ of any $\mathbf{X} \in \boldsymbol{\mathcal{S}}_{\mathbf{r}}$ obeys $\norma{\mathbf{S}}_F=1$. Our argument constructs an $\varepsilon$-net for $\boldsymbol{\mathcal{S}}_{\mathbf{r}}$ by covering the sets of matrices $\mathbf{U}_1,\mathbf{U}_2,\ldots,\mathbf{U}_d$ with orthonormal columns and the set of unit Frobenius norm tensors $\mathbf{S}$. 
For simplicity we assume that $n_1=n_2=\ldots=n_d=n$ and $r_1=r_2=\ldots=r_d=r$ since the general case requires only a straightforward modification.

The set $\boldsymbol{\mathcal{D}}$ of all-orthogonal $d$-th order tensors $\mathbf{X} \in \R^{r \times r \times \cdots \times r}$ with unit Frobenius norm
is contained in $\boldsymbol{\mathcal{F}}=\skup{\mathbf{X} \in \R^{r \times r \times \cdots \times r} : \norma{\mathbf{X}}_F=1}$.
Lemma~\ref{unitball} therefore provides an $\varepsilon/\zag{d+1}$-net
 $\mathcal{N}_{\varepsilon/\zag{d+1}}^{\boldsymbol{\mathcal{F}}}$ with respect to the Frobenius norm of cardinality
$\left|\mathcal{N}_{\varepsilon/\zag{d+1}}^{\boldsymbol{\mathcal{F}}}\right| \leq \zag{3\zag{d+1}/\varepsilon}^{r^d}$.
For covering $\boldsymbol{\mathcal{O}}_{n,r}=\skup{\mathbf{U} \in \R^{n \times r} : \mathbf{U}^*\mathbf{U}=\mathbf{I}}$, 
it is beneficial to use the norm $\norma{\cdot}_{1,2}$ defined as
$$ \norma{\mathbf{V}}_{1,2}=\max_i\norma{\mathbf{V}\zag{:,i}}_{2},$$
where $\mbf{V}\zag{:,i}$ denotes the $i$-th column of $\mbf{V}$.
Since the elements of $\boldsymbol{\mathcal{O}}_{n,r}$ have normed columns, it holds
$\boldsymbol{\mathcal{O}}_{n,r} \subset \boldsymbol{\mathcal{Q}}_{n,r} =\skup{\mathbf{V} \in \R^{n \times r}: \norma{\mathbf{V}}_{1,2}\leq 1}$.  
Lemma~\ref{unitball} 
gives $\mathcal{N}\zag{\boldsymbol{\mathcal{O}}_{n,r},\norma{\cdot}_{1,2},\varepsilon/\zag{d+1}} \leq \zag{3\zag{d+1}/\varepsilon}^{nr}$, i.e.,
there exists an $\varepsilon/\zag{d+1}$-net $\mathcal{N}_{\varepsilon/\zag{d+1}}^{\boldsymbol{\mathcal{O}}_{n,r}}$ of this cardinality.

Then the set
\begin{equation*}
\mathcal{N}_{\varepsilon}^{\boldsymbol{\mathcal{S}}_{\mathbf{r}}}:=\left\{\overline{\mathbf{S}} \times_1 \overline{\mathbf{U}}_1 \times_2 \overline{\mathbf{U}}_2 \times \cdots \times_d \overline{\mathbf{U}}_d : \overline{\mathbf{S}}\in \mathcal{N}_{\varepsilon/\zag{d+1}}^{\boldsymbol{\mathcal{D}}} \text{ and }\overline{\mathbf{U}}_i \in \mathcal{N}_{\varepsilon/\zag{d+1}}^{\boldsymbol{\mathcal{O}}_{n,r}} \text{ for all } i \in \uglate{d} \right\},
\end{equation*} 
obeys 
$$
\left|\mathcal{N}_{\varepsilon}^{\boldsymbol{\mathcal{S}}_{\mathbf{r}}} \right|
\leq \mathcal{N}\zag{\boldsymbol{\mathcal{D}},\norma{\cdot}_F,\varepsilon/\zag{d+1}}\uglate{\mathcal{N}\zag{\boldsymbol{\mathcal{O}}_{n,r}, \norma{\cdot}_{1,2}, \varepsilon/\zag{d+1}}}^d  \leq \zag{3\zag{d+1}/\varepsilon}^{r^d+dnr}.
$$ 
It remains to show that $\mathcal{N}_{\varepsilon}^{\boldsymbol{\mathcal{S}}_{\mathbf{r}}}$ is an $\varepsilon$-net for $\boldsymbol{\mathcal{S}}_{\mathbf{r}}$, i.e., that 
for all $\mathbf{X} \in \boldsymbol{\mathcal{S}}_{\mathbf{r}}$ there exists $\overline{\mathbf{X}}\in \mathcal{N}_{\varepsilon}^{\boldsymbol{\mathcal{S}}_{\mathbf{r}}}$ with $\norma{\mathbf{X}-\overline{\mathbf{X}}}_F\leq \varepsilon$. To this end, we fix $\mathbf{X} \in \boldsymbol{\mathcal{S}}_{\mathbf{r}}$ and decompose $\mathbf{X}$ as $\mathbf{X}=\mathbf{S} \times_1 \mathbf{U}_1 \times_2 \mathbf{U}_2 \times \cdots \times_d \mathbf{U}_d$. Then there exists $\overline{\mathbf{X}}=\overline{\mathbf{S}} \times_1 \overline{\mathbf{U}}_1 \times_2 \overline{\mathbf{U}}_2 \times \cdots \times_d \overline{\mathbf{U}}_d \in \mathcal{N}_{\varepsilon}^{\boldsymbol{\mathcal{S}}_{\mathbf{r}}}$ with $\overline{\mathbf{U}}_i \in \mathcal{N}_{\varepsilon/\zag{d+1}}^{\boldsymbol{\mathcal{O}}_{n,r}}$, for all $i \in \uglate{d}$ and $\overline{\mathbf{S}} \in \mathcal{N}_{\varepsilon/\zag{d+1}}^{\boldsymbol{\mathcal{D}}}$ obeying $\norma{\mathbf{U}_i-\overline{\mathbf{U}}_i}_{1,2} \leq \varepsilon/\zag{d+1}$, for all $i \in \uglate{d}$ and $\norma{\mathbf{S}-\overline{\mathbf{S}}}_F \leq \varepsilon/\zag{d+1}$. This gives
\begin{align}
\norma{\mathbf{X}-\overline{\mathbf{X}}}_F & = \norma{\mathbf{S} \times_1 \mathbf{U}_1 \times \cdots \times_d \mathbf{U}_d-
\overline{\mathbf{S}} \times_1 \overline{\mathbf{U}}_1 \times \cdots \times_d \overline{\mathbf{U}}_d}_F \notag \\
& = \left\| \mathbf{S} \times_1 \mathbf{U}_1 \times_2 \mathbf{U}_2 \times \cdots \times_d \mathbf{U}_d \pm \mathbf{S} \times_1 \mathbf{U}_1 \times \mathbf{U}_2 \times \cdots \times_{d-1} \mathbf{U}_{d-1} \times_d \overline{\mathbf{U}}_d  \right. \notag\\
& \pm \left. \mathbf{S} \times_1 \mathbf{U}_1 \times_2 \mathbf{U}_2 \times \cdots \times_{d-2} \mathbf{U}_{d-2} \times_{d-1} \overline{\mathbf{U}}_{d-1} \times_d \overline{\mathbf{U}}_d  \right. \notag\\ 
& \pm \cdots \pm \left. \mathbf{S} \times_1 \overline{\mathbf{U}}_1  \times \cdots \times_d \overline{\mathbf{U}}_d - \overline{\mathbf{S}} \times_1 \overline{\mathbf{U}}_1 \times \cdots \times_d \overline{\mathbf{U}}_d \right\|_F \notag \\
& \leq  \norma{\mathbf{S} \times_1 \mathbf{U}_1 \times_2 \mathbf{U}_2 \times \cdots \times_d \zag{\mathbf{U}_d-\overline{\mathbf{U}}_d}}_F \notag  \\
& + \norma{\mathbf{S} \times_1 \mathbf{U}_1 \times_2 \mbf{U}_2 \times \cdots \times_{d-1} \zag{ \mathbf{U}_{d-1} - \overline{\mathbf{U}}_{d-1}} \times_d \overline{\mathbf{U}}_d}_F \notag \\
& +  \cdots \, + \norma{\mbf{S} \times_1 \zag{\mbf{U}_1-\overline{\mbf{U}}_1} \times_2 \overline{\mbf{U}}_2 \times \cdots \times_d \overline{\mbf{U}}_d}_F \notag\\
& + \norma{\zag{\mathbf{S}-\overline{\mathbf{S}}} \times_1 \overline{\mathbf{U}}_1 \times_2 \overline{\mathbf{U}}_2 \times \cdots \times_d \overline{\mathbf{U}}_d}_F.
\label{last:term:HOSVD}
\end{align}
For the first $d$ terms note that by unitarity $\sum_{i_j} \mathbf{U}_j\zag{i_j,k_j}\mathbf{U}_j\zag{i_j,l_j} = \delta_{k_jl_j}$ and $\sum_{i_j} \overline{\mathbf{U}}_j\zag{i_j,k_j}\overline{\mathbf{U}}_j\zag{i_j,l_j} = \delta_{k_jl_j}$, for all $j \in \uglate{d}$, and 
$\interval{\mbf{S}_{i_j=k_j}, \mbf{S}_{i_j=l_j}}=0$  for all  $j \in \uglate{d}$ whenever $k_j\neq l_j$.
Therefore, we obtain
\begin{align*}
& \norma{\mathbf{S} \times_1 \mathbf{U}_1 \times_2 \mathbf{U}_2 \times \cdots \times_{j-1} \mathbf{U}_{j-1} \times_j \zag{\mathbf{U}_j-\overline{\mathbf{U}}_j} \times_{j+1} \overline{\mbf{U}}_{j+1} \times \cdots \times_d \overline{\mbf{U}}_d}_F^2  \\ 
& =  \sum_{i_1,\ldots,i_d} \uglate{ \zag{\mathbf{S} \times_1 \mathbf{U}_1 \times_2 \mbf{U}_2 \times \cdots \times_{j-1} \mathbf{U}_{j-1} \times_j \zag{\mathbf{U}_j-\overline{\mathbf{U}}_j} \times_{j+1} \overline{\mbf{U}}_{j+1} \times \cdots \times_d \overline{\mbf{U}}_d} \zag{i_1,\ldots,i_d}}^2 \\
& = \sum_{i_1,\ldots,i_d} \sum_{k_1,\ldots,k_d} \sum_{l_1,\ldots,l_d} \mathbf{S}\zag{k_1,\ldots,k_d} \mathbf{S}\zag{l_1,\ldots,l_d} \mathbf{U}_1\zag{i_1,k_1}  \mathbf{U}_1\zag{i_1,l_1} \mbf{U}_2\zag{i_2,k_2} \mbf{U}_2\zag{i_2,l_2}  \\
& \cdot \ldots \cdot \zag{\mathbf{U}_j-\overline{\mathbf{U}}_j}\zag{i_j,k_j} \zag{\mathbf{U}_j-\overline{\mathbf{U}}_j}\zag{i_j,l_j} \cdot \ldots \cdot \overline{\mbf{U}}_d\zag{i_d,k_d} \overline{\mbf{U}}_d\zag{i_d,l_d} \\
&= \sum_{i_j} \sum_{k_1,\ldots,k_d} \sum_{l_j} \mathbf{S}\zag{k_1,\ldots,k_j,\ldots,k_d}\mathbf{S}\zag{k_1,\ldots,l_j,\ldots,k_d} \zag{\mathbf{U}_j-\overline{\mathbf{U}}_j}\zag{i_j,k_j} \zag{\mathbf{U}_j-\overline{\mathbf{U}}_j}\zag{i_j,l_j}  \\
&=\sum_{i_j}\sum_{k_1,k_2,\ldots,k_d}\mathbf{S}\zag{k_1,k_2,\ldots,k_d}^2\zag{\zag{\mathbf{U}_j-\overline{\mathbf{U}}_j}\zag{i_j,k_j}}^2 \leq \norma{\mathbf{U}_j-\overline{\mathbf{U}}_j}_{1,2}^2 \norma{\mathbf{S}}_F^2 = \norma{\mathbf{U}_j-\overline{\mathbf{U}}_j}_{1,2}^2 \\ 
& \leq \zag{\varepsilon/\zag{d+1}}^2.
\end{align*}
In order to bound the last term in \eqref{last:term:HOSVD}, observe that the unitarity of the matrices $\overline{\mathbf{U}}_i$ gives
\begin{equation*}
\norma{\zag{\mathbf{S}-\overline{\mathbf{S}}} \times_1 \overline{\mathbf{U}}_1 \times \cdots \times_d \overline{\mathbf{U}}_d}_F=\norma{\mathbf{S}-\overline{\mathbf{S}}}_F \leq \varepsilon/\zag{d+1}.
\end{equation*} 
This completes the proof. 
\end{proof}

Next, we bound the covering numbers related to the HT decomposition, which includes the TT decomposition as a special case.

\begin{lemma}[Covering numbers related to HT-decomposition]\label{coveringHT} 
For a given HT-tree $T_I$, 
the covering numbers of the set of unit norm, rank-$\mathbf{r}$ tensors
$$ 
\boldsymbol{\mathcal{S}}_{\mathbf{r}}^{\operatorname{HT}}=\skup{\mbf{X} \in \R^{n_1 \times n_2 \times \cdots \times n_d}: \rank_{\operatorname{HT}}\zag{\mbf{X}} \leq \mbf{r}_{\operatorname{HT}}, \norma{\mbf{X}}_F=1}
$$
satisfy
\begin{align} \label{coveringHTmin}
\mathcal{N}\zag{\boldsymbol{\mathcal{S}}_{\mathbf{r}}^{\operatorname{HT}}, \norma{\cdot}_F, \varepsilon} \leq  \zag{3(2d-1)\sqrt{r}/\varepsilon}^{\sum_{t \in \mathcal{I}\zag{T_{I}}} r_t r_{t_1} r_{t_2} + \sum_{i=1}^d r_i n_i } \quad \text{for  } 0\leq \varepsilon \leq 1, 
\end{align}
where $r=\max\skup{r_t: t \in T_{I}}$, and $t_1$, $t_2$ are the left and the right son of a node $t$, respectively.
\end{lemma} 

The proof requires a non-standard orthogonalization of the HT-decomposition. (The standard orthogonalization leads to worse bounds, in both TT and HT case.)
We say that a tensor $\mbf{B}_t \in \C^{r_t \times r_{t_1} \times r_{t_2}}$ is \textit{right-orthogonal} if $\zag{\mbf{B}_t^{\{2,3\}}}^T\mbf{B}_t^{\{2,3\}}=\mbf{I}_{r_t}$. We call an HT-decomposition \textit{right-orthogonal} if all transfer tensors $\mbf{B}_t$, for $t \in \mathcal{I}(T_I) \backslash \{t_{\text{root}}\}$, i.e.\ except for the root, are \textit{right orthogonal} and all frames $\mbf{U}_i$ have orthogonal columns.
For the sake of simple notation, we write the right-orthogonal HT-decomposition of a tensor $\mbf{X} \in \R^{n_1 \times n_2 \times n_3 \times n_4}$ with the corresponding HT-tree as in Figure~\ref{HTcovering} as
\begin{equation}\label{4:rightorthogonal}
  \mbf{X}=\mbf{B}_{\{1,2,3,4\}} \bigtriangledown \left(\mbf{B}_{\{1,2\}} \bigtriangledown \mbf{U}_1 \bigtriangledown \mbf{U}_2 \right)\bigtriangledown \left(\mbf{B}_{\{3,4\}} \bigtriangledown \mbf{U}_3 \bigtriangledown \mbf{U}_4\right).
  \end{equation}
  \begin{figure*}[h]
\centering
\setlength\fboxsep{0pt}
\setlength\fboxrule{0.5pt}
\includegraphics[scale=1, trim=120 325 360 350, clip]{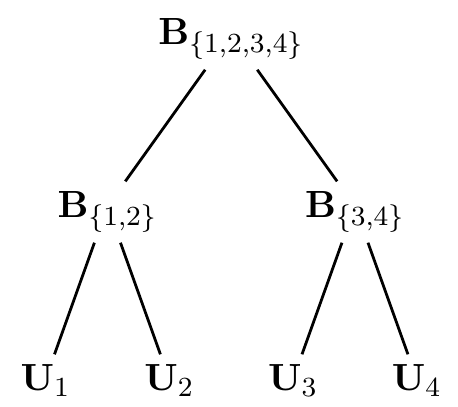}
\caption{Tree for the HT-decomposition with $d=4$}
\label{HTcovering}
\end{figure*}
In fact, the above decomposition can be written as
$$ 
  \mbf{X}=\mbf{B}_{\{1,2,3,4\}} \bigtriangledown \left(\mbf{B}_{\{1,2\}} \times_2 \mbf{U}_1 \times_3 \mbf{U}_2 \right)\bigtriangledown \left(\mbf{B}_{\{3,4\}} \times_2 \mbf{U}_3 \times_3 \mbf{U}_4\right).
$$
since $\mbf{U}_i$ is a matrix for all $i \in \uglate{4}$. However, for simplicity, we are going to use the notation as in \eqref{4:rightorthogonal}.
A right-orthogonal HT-decomposition can be obtained as follows from the standard orthogonal HT-decomposition (see  \cite{grasedyck2010hierarchical}), where in particular, all frames $\mbf{U}_i$ have orthogonal columns. 

We first compute the QR-decomposition of the flattened transfer tensors $\mbf{B}_t^{\{2,3\}} = \mbf{Q}_t^{\{2,3\}} \mbf{R}_t$ 
for all nodes $t$ at the
highest possible level $\ell=p-1$. The level $\ell$ of the tree is defined as the set of all nodes having the distance of exactly $\ell$ to the root. We denote the level $\ell$ of the tree $T_I$ as $T_I^{\ell}=\skup{t \in T_I:\level(t)=\ell}$. (For example, for tree $T_I$ as in Figure \ref{HTcovering}, $T_I^0=\{\{1,2,3,4\}\}$, $T_I^1=\{\{1,2\},\{3,4\}\}$, $T_I^2=\{\{1\},\{2\},\{3\},\{4\}\}$.) The $\mbf{Q}_t$'s are then right-orthogonal by construction.
In order to obtain a representation of the same tensor, we have to replace the tensors $\mbf{B}_{t'}$ with nodes at lower level $p-2$
by $\bar{\mbf{B}}_{t'} = \mbf{B}_{t'} \times_2 \mbf{R}_{t_{\text{left}}} \times_3 \mbf{R}_{t_{\text{right}}}$, where
$t_{\text{left}}$ corresponds to the left son of $t'$ and  $t_{\text{right}}$ to the right son. We continue
by computing the QR-decompositions of $\bar{\mbf{B}}_{t'}^{\{2,3\}}$ with $t'$ at level $p-2$ and so on until we finally updated
the root $\mbf{B}_{\{1,2,\hdots,d\}}$ (which may remain the only non right-orthogonal transfer tensor). 
We illustrate this right-orthogonalization process for an HT-decomposition of the form \eqref{4:rightorthogonal} related to the HT-tree of Figure~\ref{HTcovering}:
\begin{align*} 
  \mbf{X} & =\mbf{B}_{\{1,2,3,4\}} \bigtriangledown \left(\mbf{B}_{\{1,2\}} \bigtriangledown \mbf{U}_1 \bigtriangledown \mbf{U}_2 \right)\bigtriangledown \left(\mbf{B}_{\{3,4\}} \bigtriangledown \mbf{U}_3 \bigtriangledown \mbf{U}_4\right) \\
  & = \mbf{B}_{\{1,2,3,4\}} \bigtriangledown \left( \left[\mbf{Q}_{\{1,2\}} \times_1 \mbf{R}_{\{1,2\}}\right] \bigtriangledown \mbf{U}_1 \bigtriangledown \mbf{U}_2 \right)\bigtriangledown \left(\left[\mbf{Q}_{\{3,4\}} \times_1 \mbf{R}_{\{3,4\}}\right] \bigtriangledown \mbf{U}_3 \bigtriangledown \mbf{U}_4\right) \\
  & = \left[ \mbf{B}_{\{1,2,3,4\}} \times_2 \mbf{R}_{\{1,2\}} \times_3 \mbf{R}_{\{3,4\}} \right] \bigtriangledown \left(\mbf{Q}_{\{1,2\}} \bigtriangledown \mbf{U}_1 \bigtriangledown \mbf{U}_2 \right)\bigtriangledown \left(\mbf{Q}_{\{3,4\}} \bigtriangledown \mbf{U}_3 \bigtriangledown \mbf{U}_4\right). \\
  \end{align*}
The second identity is easily verified by writing out the expressions with index notation. The last expression is a right-orthogonal
HT decomposition with root tensor 
$\overline{\mbf{B}}_{\{1,2,3,4\}} = \mbf{B}_{\{1,2,3,4\}} \times_2 \mbf{R}_{\{1,2\}} \times_3 \mbf{R}_{\{3,4\}}$.

\begin{proof}[Proof of Lemma~\ref{coveringHT}]
For the sake of better readability, we will show the result for the special cases of the order $4$ HT-decomposition as in 
Figure~\ref{HTcovering} as well as for the special case of the TT decomposition for arbitary $d$.
The general case is then done analogously.

For the HT-tree $T_I$ as in Figure~\ref{HTcovering} we have $T_I=\skup{\{1,2,3,4\},\{1,2\},\{1\},\{2\},\{3,4\},\{3\},\{4\}}$ and the number of nodes is $\aps{T_I}=2d-1=7$. We  have to show that for $T_I$ as in Figure~\ref{HTcovering}, the covering numbers of 
$$ 
\boldsymbol{\mathcal{S}}_{\mathbf{r}}^{\text{HT}}=\skup{\mbf{X} \in \R^{n_1 \times n_2 \times \cdots \times n_d}: \rank_{\text{HT}}\zag{\mbf{X}} \leq \mbf{r}_{\text{HT}}, \norma{\mbf{X}}_F=1}, 
$$
satisfy
$$
\mathcal{N}\zag{\boldsymbol{\mathcal{S}}_{\mathbf{r}}^{\text{HT}}, \norma{\cdot}_F, \varepsilon} \leq  \zag{21\sqrt{r}/\varepsilon}^{r_{\{1,2,3,4\}} r_{\{1,2\}} r_{\{3,4\}} + r_{\{1,2\}} r_1 r_2 +r_{\{3,4\}} r_3 r_4+ \sum_{i=1}^4 r_i n_i } \quad \text{for  } 0\leq \varepsilon \leq 1. 
$$
For simplicity, we treat the case that $r_t=r$  for all $t \in T_I$ and $n_i = n$ for $i \in \uglate{4}$. We will use the right-orthogonal HT-decomposition introduced above and we cover the admissible components $\mbf{U}_i$ and $\mbf{B}_t$ in  
\eqref{4:rightorthogonal} separately, for all $t \in T_I$ and $i \in \uglate{4}$. 

We introduce the set of right-orthogonal tensors $\bm{\mathcal{O}}_{r,r,r}^{\text{right}}=\skup{\mbf{U} \in \R^{r \times r \times r}: {\mbf{U}^{\{2,3\}}}^T \mbf{U}^{\{2,3\}}=\mbf{I}_r}$ which we will cover with respect to the norm
\begin{equation}\label{norma1F}
  \norma{\mbf{U}}_{F,1}:=\max_{i}\norma{\mbf{U}\zag{i,:,:}}_F.
\end{equation}
The set $\bm{\mathcal{Q}}_{r,r,r}^{\text{right}}:=\skup{\mbf{X} \in \R^{r \times r \times r}: \norma{\mbf{X}}_{F,1}\leq 1}$ contains $\bm{\mathcal{O}}_{r,r,r}^{\text{right}}$. Thus, by Lemma~\ref{unitball} there is an $\varepsilon/\zag{7\sqrt{r}}$-set $\mathcal{N}_{\varepsilon/\zag{7\sqrt{r}}}^{\bm{\mathcal{O}}_{r,r,r}^{\text{right}}}$ for 
$\bm{\mathcal{O}}_{r,r,r}^{\text{right}}$ obeying
$$ \aps{\mathcal{N}_{\varepsilon/\zag{7\sqrt{r}}}^{\bm{\mathcal{O}}_{r,r,r}^{\text{right}}}} \leq \zag{3\cdot 7\sqrt{r}/\varepsilon}^{r^3}= \zag{21\sqrt{r}/\varepsilon}^{r^3}.$$
For the frames $\mbf{U}_i \in \R^{n \times r}$ with $i \in \uglate{4}$, we define the set $\bm{\mathcal{O}}_{n,r}=\skup{\mbf{U} \in \R^{n \times r}: \mbf{U}^T \mbf{U}=\mbf{I}_r}$ which we cover with respect to 
$$ \norma{\mbf{U}}_{1,2}:=\max_i \norma{\mbf{U}\zag{:,i}}_2.$$
Clearly, $\bm{\mathcal{O}}_{n,r} \subseteq \bm{\mathcal{Q}}_{n,r}:=\skup{\mbf{X} \in \R^{n \times r}: \norma{\mbf{X}}_{1,2}\leq 1}$ since the elements of an orthonormal set are unit normed. Again by Lemma~\ref{unitball}, there is an $\varepsilon/\zag{7\sqrt{r}}$-set $\mathcal{N}_{\varepsilon/\zag{7\sqrt{r}}}^{\bm{\mathcal{O}}_{n,r}}$ for $\bm{\mathcal{O}}_{n,r}$ obeying
$$ \aps{\mathcal{N}_{\varepsilon/\zag{7\sqrt{r}}}^{\bm{\mathcal{O}}_{n,r}}} \leq \zag{21\sqrt{r}/\varepsilon}^{nr}.$$
Finally, to cover $\mbf{B}_{\{1,2,3,4\}}$, we define the set $\bm{\mathcal{F}}_{r,r}=\skup{\mbf{X} \in \R^{1\times r \times r}: \norma{\mbf{X}}_F=1}$ which has an $\varepsilon/\zag{7\sqrt{r}}$-net $\mathcal{N}_{\varepsilon/\zag{7\sqrt{r}}}^{\bm{\mathcal{F}}_{r,r}}$ of cardinality at most $\zag{21\sqrt{r}/\varepsilon}^{r^2}$.
We now define
\begin{align*}
  \mathcal{N}_{\varepsilon}^{\bm{\mathcal{S}}_{\mbf{r}}^{\text{HT}}}
  :=&\left\{\overline{\mbf{B}}_{\{1,2,3,4\}} \bigtriangledown \left(\overline{\mbf{B}}_{\{1,2\}} \bigtriangledown \overline{\mbf{U}}_1  \bigtriangledown  \overline{\mbf{U}}_2\right) \bigtriangledown \left(\overline{\mbf{B}}_{\{3,4\}} \bigtriangledown \overline{\mbf{U}}_3  \bigtriangledown  \overline{\mbf{U}}_4\right)  : \right. \\ & \left.   \overline{\mbf{B}}_{\{1,2\}}, \overline{\mbf{B}}_{\{3,4\}} \in \mathcal{N}_{\varepsilon/\zag{7\sqrt{r}}}^{\bm{\mathcal{O}}_{r,r,r}^{\text{right}}},\overline{\mbf{B}}_{\{1,2,3,4\}} \in \mathcal{N}_{\varepsilon/\zag{7\sqrt{r}}}^{\bm{\mathcal{F}}_{r,r}}, \overline{\mbf{U}}_i \in \mathcal{N}_{\varepsilon/\zag{7\sqrt{r}}}^{\bm{\mathcal{O}}_{n,r}} \text{ for all } i \in \uglate{4} \right\}
  \end{align*}
  and remark that
 \begin{equation*}
 \mathcal{N}\zag{\bm{\mathcal{S}}_{\mbf{r}}^{\text{HT}}, \norma{\cdot}_F, \varepsilon} \leq \aps{\mathcal{N}_{\varepsilon/\zag{7\sqrt{r}}}^{\bm{\mathcal{O}}_{r,r,r}^{\text{right}}}}^{2} \aps{\mathcal{N}_{\varepsilon/\zag{7\sqrt{r}}}^{\bm{\mathcal{O}}_{n,r}}}^{4} \aps{\mathcal{N}_{\varepsilon/\zag{7\sqrt{r}}}^{\bm{\mathcal{F}}_{r,r}}} \leq \zag{21\sqrt{r}/\varepsilon}^{3r^3+4nr}.
 \end{equation*}
It remains to show that for any  
$\mbf{X} \in \bm{\mathcal{S}}_{\mbf{r}}^{\text{HT}}$ 
there exists  
$\overline{\mbf{X}} \in \mathcal{N}_{\varepsilon}^{\bm{\mathcal{S}}_{\mbf{r}}^{\text{HT}}}$  
such that $\norma{\mbf{X}-\overline{\mbf{X}}}_F \leq 1$.
For $\mbf{X}=\mbf{B}_{\{1,2,3,4\}} \bigtriangledown \left(\mbf{B}_{\{1,2\}} \bigtriangledown \mbf{U}_1  \bigtriangledown  \mbf{U}_2 \right)\bigtriangledown \left( \mbf{B}_{\{3,4\}} \bigtriangledown \mbf{U}_3  \bigtriangledown \mbf{U}_4\right)$, we choose $\overline{\mbf{X}}=\overline{\mbf{B}}_{\{1,2,3,4\}} \bigtriangledown \left(\overline{\mbf{B}}_{\{1,2\}} \bigtriangledown \overline{\mbf{U}}_1  \bigtriangledown  \overline{\mbf{U}}_2\right) \bigtriangledown \left( \overline{\mbf{B}}_{\{3,4\}} \bigtriangledown \overline{\mbf{U}}_3  \bigtriangledown  \overline{\mbf{U}}_4 \right)  \in \mathcal{N}_{\varepsilon}^{\bm{\mathcal{S}}_{\mbf{r}}^{\text{HT}}}$ 
such that
$\overline{\mbf{B}}_{\{1,2,3,4\}} \in \bm{\mathcal{F}}_{r,r}$, $\overline{\mbf{B}}_{\{1,2\}},\overline{\mbf{B}}_{\{3,4\}} \in \bm{\mathcal{O}}_{r,r,r}^{\text{right}}$, $\overline{\mbf{U}}_i \in \bm{\mathcal{O}}_{n,r}$ for all $i \in \uglate{4}$ and
\begin{align*}
& \norma{\mbf{U}_i-\overline{\mbf{U}}_i}_{1,2} \leq \frac{\varepsilon}{7\sqrt{r}} \quad \text{for all } i \in \uglate{4}, \\
& \norma{\mbf{B}_{\{1,2,3,4\}}-\overline{\mbf{B}}_{\{1,2,3,4\}}}_F\leq \frac{\varepsilon}{7\sqrt{r}}, \\
& \norma{\mbf{B}_{\{1,2\}}-\overline{\mbf{B}}_{\{1,2\}}}_{F,1} \leq \frac{\varepsilon}{7\sqrt{r}}, \quad \text{ and } \norma{\mbf{B}_{\{3,4\}}-\overline{\mbf{B}}_{\{3,4\}}}_{F,1} \leq \frac{\varepsilon}{7\sqrt{r}}.
\end{align*}
Applying the triangle inequality results in
\begin{align}
\norma{\mbf{X}-\overline{\mbf{X}}}_F &\leq \norma{\mbf{B}_{\{1,2,3,4\}} \bigtriangledown \left(\mbf{B}_{\{1,2\}} \bigtriangledown \mbf{U}_1  \bigtriangledown  \mbf{U}_2 \right) \bigtriangledown \left( \mbf{B}_{\{3,4\}} \bigtriangledown \mbf{U}_3  \bigtriangledown \zag{\mbf{U}_4-\overline{\mbf{U}}_4}\right)}_F \label{coveringHTeq1} \\ 
& + \norma{\mbf{B}_{\{1,2,3,4\}} \bigtriangledown \left(\mbf{B}_{\{1,2\}} \bigtriangledown \mbf{U}_1  \bigtriangledown  \mbf{U}_2\right) \bigtriangledown \left(\mbf{B}_{\{3,4\}} \bigtriangledown \zag{\mbf{U}_3 -\overline{\mbf{U}}_3} \bigtriangledown \overline{\mbf{U}}_4 \right)}_F \nonumber \\
&+\norma{\mbf{B}_{\{1,2,3,4\}} \bigtriangledown \left( \mbf{B}_{\{1,2\}} \bigtriangledown \mbf{U}_1  \bigtriangledown  \mbf{U}_2 \right)\bigtriangledown \left(\zag{\mbf{B}_{\{3,4\}}-\overline{\mbf{B}}_{\{3,4\}}} \bigtriangledown \overline{\mbf{U}}_3  \bigtriangledown \overline{\mbf{U}}_4 \right)}_F  \label{coveringHTeq2} \\
&+ \cdots +\norma{\zag{\mbf{B}_{\{1,2,3,4\}}-\overline{\mbf{B}}_{\{1,2,3,4\}}} \bigtriangledown \left(\overline{\mbf{B}}_{\{1,2\}} \bigtriangledown \overline{\mbf{U}}_1  \bigtriangledown  \overline{\mbf{U}}_2\right) \bigtriangledown \left(\overline{\mbf{B}}_{\{3,4\}} \bigtriangledown \overline{\mbf{U}}_3  \bigtriangledown \overline{\mbf{U}}_4 \right)}_F. \label{triangle4}
\end{align}
To estimate \eqref{coveringHTeq1}, we use orthogonality of $\mbf{U}_i$, $i \in \uglate{4}$, and the right-orthogonality of 
$\mbf{B}_{\{1,2\}}$, $\mbf{B}_{\{3,4\}}$ to obtain
\begin{align*}
& \norma{\mbf{B}_{\{1,2,3,4\}} \bigtriangledown \left(\mbf{B}_{\{1,2\}} \bigtriangledown \mbf{U}_1  \bigtriangledown  \mbf{U}_2 \right)\bigtriangledown \left(\mbf{B}_{\{3,4\}} \bigtriangledown \mbf{U}_3  \bigtriangledown \zag{\mbf{U}_4-\overline{\mbf{U}}_4\right)}}_F^2 \\
&= \sum_{i_1,\ldots,i_4} \sum_{\substack{j_1,\ldots,j_4 \\ k_1,\ldots, k_4}} \sum_{\substack{ j_{12}, \\ k_{12}}} \sum_{\substack{ j_{34}, \\ k_{34}}}  \mbf{B}_{\{1,2,3,4\}}\zag{1,j_{12},j_{34}} \mbf{B}_{\{1,2,3,4\}}\zag{1,k_{12},k_{34}} \mbf{B}_{\{1,2\}}\zag{j_{12},j_{1},j_2}  \mbf{B}_{\{1,2\}}\zag{k_{12},k_{1},k_2} \\ & \cdot \mbf{U}_1 \zag{i_1,j_1} \mbf{U}_1 \zag{i_1,k_1} \mbf{U}_2\zag{i_2,j_2} \mbf{U}_2\zag{i_2,k_2}\mbf{B}_{\{3,4\}}\zag{j_{34},j_{3},j_4}  \mbf{B}_{\{3,4\}}\zag{k_{34},k_{3},k_4} \\ & \cdot \mbf{U}_3 \zag{i_3,j_3} \mbf{U}_3 \zag{i_3,k_3} \zag{\mbf{U}_4-\overline{\mbf{U}}_4}\zag{i_4,j_4} \zag{\mbf{U}_4-\overline{\mbf{U}}_4}\zag{i_4,k_4} \\
&=\sum_{i_4} \sum_{\substack{j_3,j_4 \\  k_4}} \sum_{j_{12}} \sum_{\substack{ j_{34}, \\ k_{34}}}  \mbf{B}_{\{1,2,3,4\}}\zag{1,j_{12},j_{34}} \mbf{B}_{\{1,2,3,4\}}\zag{1,j_{12},k_{34}}    \mbf{B}_{\{3,4\}}\zag{j_{34},j_{3},j_4}  \mbf{B}_{\{3,4\}}\zag{k_{34},j_{3},k_4} \\ & \cdot \zag{\mbf{U}_4-\overline{\mbf{U}}_4}\zag{i_4,j_4} \zag{\mbf{U}_4-\overline{\mbf{U}}_4}\zag{i_4,k_4} = \interval{\Delta\mbf{U}_4, \square\mbf{B}_{\{3,4\}}} \leq \norma{\Delta\mbf{U}_4}_{2 \rightarrow 2} \norma{\square\mbf{B}_{\{3,4\}}}_*
\end{align*}
where 
\begin{align*}
 \Delta\mbf{U}_4\zag{j_4,k_4}=&\sum_{i_4} \zag{\mbf{U}_4-\overline{\mbf{U}}_4}\zag{i_4,j_4} \zag{\mbf{U}_4-\overline{\mbf{U}}_4}\zag{i_4,k_4} = (\mbf{U}_4 - \overline{\mbf{U}}_4)^T(\mbf{U}_4-\overline{\mbf{U}}_4)(j_4,k_4) ,\\
\square\mbf{B}_{\{3,4\}}\zag{j_4,k_4}=&\sum_{j_3} \sum_{j_{12}} \sum_{ j_{34}, k_{34}}  \mbf{B}_{\{1,2,3,4\}}\zag{1,j_{12},j_{34}} \mbf{B}_{\{1,2,3,4\}}\zag{1,j_{12},k_{34}}  \\ & \cdot \mbf{B}_{\{3,4\}}\zag{j_{34},j_{3},j_4}  \mbf{B}_{\{3,4\}}\zag{k_{34},j_{3},k_4}. 
\end{align*}
Since the Frobenius norm dominates the spectral norm, we have
\begin{equation}\label{DeltaU4}
  \norma{\Delta\mbf{U}_4}_{2\rightarrow 2} = \norma{\mbf{U}_4-\overline{\mbf{U}}_4}_{2 \rightarrow 2}^2 
  \leq \norma{\mbf{U}_4-\overline{\mbf{U}}_4}_F^2 \leq r \norma{\mbf{U}_4-\overline{\mbf{U}}_4}_{1,2}^2.
  \end{equation}
Since $\square\overline{\mbf{B}}_{\{3,4\}}$ is symmetric and positive semidefinite, it holds
$$1=\norma{\overline{\mbf{X}}}_F^2=\interval{\mbf{I},\square\overline{\mbf{B}}_{\{3,4\}}}=\tr\zag{\square\overline{\mbf{B}}_{\{3,4\}}}
=\norma{\square\overline{\mbf{B}}_{\{3,4\}}}_*.$$
Hence, 
\begin{align*}
&\norma{\mbf{B}_{\{1,2,3,4\}} \bigtriangledown \left( \mbf{B}_{\{1,2\}} \bigtriangledown \mbf{U}_1  \bigtriangledown  \mbf{U}_2 \right)\bigtriangledown \left( \mbf{B}_{\{3,4\}} \bigtriangledown \mbf{U}_3  \bigtriangledown \zag{\mbf{U}_4-\overline{\mbf{U}}_4 \right)}}_F \leq \sqrt{r}  \norma{\mbf{U}_4-\overline{\mbf{U}}_4}_{1,2} \leq \frac{\varepsilon}{7}.
\end{align*}
A similar procedure leads to the estimates
\begin{align*}
&\norma{\mbf{B}_{\{1,2,3,4\}} \bigtriangledown \left( \mbf{B}_{\{1,2\}} \bigtriangledown \zag{\mbf{U}_1-\overline{\mbf{U}}_1}  \bigtriangledown  \overline{\mbf{U}}_2 \right) \bigtriangledown \left( \overline{\mbf{B}}_{\{3,4\}} \bigtriangledown \overline{\mbf{U}}_3  \bigtriangledown \overline{\mbf{U}}_4\right)}_F \leq \sqrt{r}  \norma{\mbf{U}_1-\overline{\mbf{U}}_1}_{1,2} \leq \frac{\varepsilon}{7}, \\
&\norma{\mbf{B}_{\{1,2,3,4\}} \bigtriangledown \left( \mbf{B}_{\{1,2\}} \bigtriangledown \mbf{U}_1  \bigtriangledown  \zag{\mbf{U}_2-\overline{\mbf{U}}_2}\right) \bigtriangledown \left( \overline{\mbf{B}}_{\{3,4\}} \bigtriangledown \overline{\mbf{U}}_3  \bigtriangledown \overline{\mbf{U}}_4\right)}_F \leq \sqrt{r}  \norma{\mbf{U}_2-\overline{\mbf{U}}_2}_{1,2} \leq \frac{\varepsilon}{7}, \\
&\norma{\mbf{B}_{\{1,2,3,4\}} \bigtriangledown \left( \mbf{B}_{\{1,2\}} \bigtriangledown \mbf{U}_1  \bigtriangledown  \mbf{U}_2 \right)\bigtriangledown \left( \mbf{B}_{\{3,4\}} \bigtriangledown \zag{\mbf{U}_3-\overline{\mbf{U}}_3}  \bigtriangledown \overline{\mbf{U}}_4\right)}_F \leq \sqrt{r}  \norma{\mbf{U}_3-\overline{\mbf{U}}_3}_{1,2} \leq \frac{\varepsilon}{7} .
\end{align*}
Since  $\overline{\mbf{U}}_i$ is orthogonal for all $i \in \uglate{4}$ and $\overline{\mbf{B}}_{\{1,2\}}, \overline{\mbf{B}}_{\{3,4\}}$ are right-orthogonal, we similarly estimate \eqref{coveringHTeq2},
\begin{align*}
&\norma{\mbf{B}_{\{1,2,3,4\}} \bigtriangledown \zag{\mbf{B}_{\{1,2\}} \bigtriangledown \mbf{U}_1  \bigtriangledown  \mbf{U}_2} \bigtriangledown \zag{ \zag{\mbf{B}_{\{3,4\}}-\overline{\mbf{B}}_{\{3,4\}}} \bigtriangledown \overline{\mbf{U}}_3  \bigtriangledown \overline{\mbf{U}}_4}}_F \\
&= \sum_{i_1,\ldots,i_4} \sum_{\substack{j_1,\ldots,j_4 \\ k_1,\ldots, k_4}} \sum_{\substack{j_{12}, \\ k_{12}}} \sum_{\substack{ j_{34}, \\ k_{34}}}  \mbf{B}_{\{1,2,3,4\}}\zag{1,j_{12},j_{34}} \mbf{B}_{\{1,2,3,4\}}\zag{1,k_{12},k_{34}} \mbf{B}_{\{1,2\}}\zag{j_{12},j_{1},j_2}  \mbf{B}_{\{1,2\}}\zag{k_{12},k_{1},k_2} \\ & \cdot \mbf{U}_1 \zag{i_1,j_1} \mbf{U}_1 \zag{i_1,k_1} \mbf{U}_2\zag{i_2,j_2} \mbf{U}_2\zag{i_2,k_2}\zag{\mbf{B}_{\{3,4\}}-\overline{\mbf{B}}_{\{3,4\}}}\zag{j_{34},j_{3},j_4}  \zag{\mbf{B}_{\{3,4\}}-\overline{\mbf{B}}_{\{3,4\}}}\zag{k_{34},k_{3},k_4} \\ & \cdot \overline{\mbf{U}}_3 \zag{i_3,j_3} \overline{\mbf{U}}_3 \zag{i_3,k_3} \overline{\mbf{U}}_4\zag{i_4,j_4} \overline{\mbf{U}}_4\zag{i_4,k_4} \\
&= \sum_{j_3,j_4} \sum_{ j_{12} } \sum_{\substack{ j_{34}, \\ k_{34}}}  \mbf{B}_{\{1,2,3,4\}}\zag{1,j_{12},j_{34}} \mbf{B}_{\{1,2,3,4\}}\zag{1,j_{12},k_{34}}  \zag{\mbf{B}_{\{3,4\}}-\overline{\mbf{B}}_{\{3,4\}}}\zag{j_{34},j_{3},j_4}  \\ & \cdot \zag{\mbf{B}_{\{3,4\}}-\overline{\mbf{B}}_{\{3,4\}}}\zag{k_{34},j_{3},j_4} = \interval{\Delta\mbf{B}_{\{3,4\}}, \square\mbf{B}_{\{1,2,3,4\}}} \leq \norma{\Delta\mbf{B}_{\{3,4\}}}_{2 \rightarrow 2} \norma{\square\mbf{B}_{\{1,2,3,4\}}}_*
\end{align*}
where 
\begin{align*}
 \Delta\mbf{B}_{\{3,4\}}\zag{j_{34},k_{34}}=&\sum_{j_3,j_4}  \zag{\mbf{B}_{\{3,4\}}-\overline{\mbf{B}}_{\{3,4\}}}\zag{j_{34},j_{3},j_4}  \zag{\mbf{B}_{\{3,4\}}-\overline{\mbf{B}}_{\{3,4\}}}\zag{k_{34},j_{3},j_4}\\
 = & \left(\mbf{B}_{\{3,4\}}^{\{2,3\}}-\overline{\mbf{B}}_{\{3,4\}}^{\{2,3\}}\right)^T \left(\mbf{B}_{\{3,4\}}^{\{2,3\}}-\overline{\mbf{B}}_{\{3,4\}}^{\{2,3\}}\right)(j_{34},k_{34}) \\
\square\mbf{B}_{\{1,2,3,4\}}\zag{j_{34},k_{34}}=&\sum_{j_{12}} \mbf{B}_{\{1,2,3,4\}}\zag{1,j_{12},j_{34}} \mbf{B}_{\{1,2,3,4\}}\zag{1,j_{12},k_{34}} . 
\end{align*}
The spectral norm of $\Delta\mbf{B}_{\{3,4\}}$ can be estimated as
\begin{align}\label{DeltaB34}
  \norma{\Delta\mbf{B}_{\{3,4\}}}_{2\rightarrow 2} 
  & = \norma{\mbf{B}_{\{3,4\}}^{\{2,3\}}-\overline{\mbf{B}}_{\{3,4\}}^{\{2,3\}}}_{2 \rightarrow 2}^2
  \leq 
\norma{\mbf{B}_{\{3,4\}}-\overline{\mbf{B}}_{\{3,4\}}}_F^2  
\leq r \norma{\mbf{B}_{\{3,4\}}-\overline{\mbf{B}}_{\{3,4\}}}_{F,1}^2.
  \end{align}
Since $\square\overline{\mbf{B}}_{\{1,2,3,4\}}$ is symmetric and positive semidefinite
$$1=\norma{\overline{\mbf{X}}}_F^2=\interval{\mbf{I},\square\overline{\mbf{B}}_{\{1,2,3,4\}}}=\tr\zag{\square\overline{\mbf{B}}_{\{1,2,3,4\}}}
=\norma{\square\overline{\mbf{B}}_{\{1,2,3,4\}}}_*.$$
Hence, 
\begin{align*}
&\norma{\mbf{B}_{\{1,2,3,4\}} \bigtriangledown \left( \mbf{B}_{\{1,2\}} \bigtriangledown \mbf{U}_1  \bigtriangledown  \mbf{U}_2 \right)\bigtriangledown \left( \zag{\mbf{B}_{\{3,4\}}-\overline{\mbf{B}}_{\{3,4\}}} \bigtriangledown \overline{\mbf{U}}_3  \bigtriangledown \overline{\mbf{U}}_4\right)}_F \leq \sqrt{r}  \norma{\mbf{B}_{\{3,4\}}-\overline{\mbf{B}}_{\{3,4\}}}_{F,1} \leq \frac{\varepsilon}{7}.
\end{align*}
A similar procedure leads to the following estimates
\begin{align*}
&\norma{\zag{\mbf{B}_{\{1,2,3,4\}}-\overline{\mbf{B}}_{\{1,2,3,4\}}} \bigtriangledown \left( \overline{\mbf{B}}_{\{1,2\}} \bigtriangledown \overline{\mbf{U}}_1  \bigtriangledown  \overline{\mbf{U}}_2\right) \bigtriangledown \left(\overline{\mbf{B}}_{\{3,4\}} \bigtriangledown \overline{\mbf{U}}_3  \bigtriangledown \overline{\mbf{U}}_4\right)}_F \leq  \norma{\mbf{B}_{\{1,2,3,4\}}-\overline{\mbf{B}}_{\{1,2,3,4\}}}_{F} \leq \frac{\varepsilon}{7}, \\
&\norma{\mbf{B}_{\{1,2,3,4\}} \bigtriangledown \left(\zag{\mbf{B}_{\{1,2\}}-\overline{\mbf{B}}_{\{1,2\}}} \bigtriangledown \overline{\mbf{U}}_1  \bigtriangledown  \overline{\mbf{U}}_2\right) \bigtriangledown \left(\overline{\mbf{B}}_{\{3,4\}} \bigtriangledown \overline{\mbf{U}}_3  \bigtriangledown \overline{\mbf{U}}_4\right)}_F \leq \sqrt{r}  \norma{\mbf{B}_{\{1,2\}}-\overline{\mbf{B}}_{\{1,2\}}}_{F,1} \leq \frac{\varepsilon}{7}.
\end{align*}
Plugging the bounds into \eqref{triangle4} completes the proof for the HT-tree of Figure~\ref{HTcovering}.

\medskip

Let us now consider the TT-decomposition for tensors of order $d \geq 3$ as illustrated in Figure~\ref{TTdecomp}. We start with a right-orthogonal decomposition  (see also the discussion after Lemma~\ref{coveringHT}) of the form
\begin{align*}
\mbf{X}\zag{i_1,i_2,\ldots,i_d} & =\sum_{j_1, j_{23\ldots d}} \sum_{j_2, j_{3\ldots d}} \cdots \sum_{j_{d-1,d}, j_d} \mbf{B}_{\{1,2,\ldots, d\}} \zag{1,j_1,j_{23\ldots d}}\mbf{U}_1\zag{i_1,j_1}  \mbf{B}_{\{2,3,\ldots, d\}} \zag{j_{23\ldots d}, j_2, j_{3\ldots d}} 
\\
&\qquad \cdot \mbf{U}_2\zag{i_2,j_2}  \cdots  \mbf{B}_{\{d-1, d\}}\zag{j_{d-1,d}, j_{d-1},j_d} \mbf{U}_{d-1}\zag{i_{d-1},j_{d-1}} \mbf{U}_d\zag{i_d,j_d}. 
\end{align*}
As for the general HT-decomposition, we write this as
\begin{equation}\label{HT-decomp-simple}
\mbf{X}= \mbf{B}_{\{1,2,3,\ldots, d\}} \triangledown \mbf{U}_1 \triangledown \left( \mbf{B}_{\{2,3,\ldots, d\}} \triangledown\mbf{U}_2 \triangledown \left( \cdots \triangledown \left( \mbf{B}_{\{d-1,d\}}\triangledown \mbf{U}_{d-1} \triangledown \mbf{U}_d \right) \cdots \right) \right).
\end{equation}
\begin{figure*}[h]
\centering
\includegraphics[scale=0.8, trim=90 345 385 270, clip]{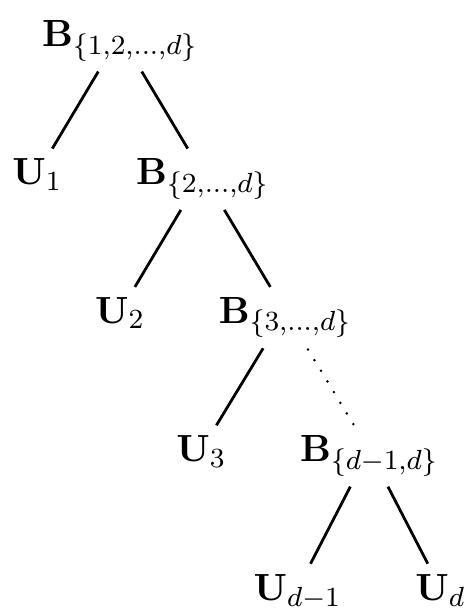}
\caption{TT decomposition}
\label{TTdecomp}
\end{figure*}
As above, we cover each set of admissible components $\mathbf{U}_i$, $\mathbf{B}_t$ separately, and then combine
these components in order to obtain a covering of 
\[
 \boldsymbol{\mathcal{S}}_{\mathbf{r}}^{\operatorname{TT}}=\skup{\mbf{X} \in \R^{n_1 \times n_2 \times \cdots \times n_d}: \rank_{\operatorname{TT}}\zag{\mbf{X}} \leq \mbf{r}_{\operatorname{TT}}, \norma{\mbf{X}}_F=1}
\]
with respect to the Frobenius norm, that is, we form
\begin{align*}
\mathcal{N}_{\varepsilon}^{\bm{\mathcal{S}}_{\mbf{r}}^{\text{TT}}} := & \left\{
\overline{\mbf{B}}_{\{1,2,3,\ldots, d\}} \triangledown \overline{\mbf{U}}_1 \triangledown \left( \overline{\mbf{B}}_{\{2,3,\ldots, d\}} \triangledown\overline{\mbf{U}}_2 \triangledown \left( \cdots \triangledown \left( \overline{\mbf{B}}_{\{d-1,d\}}\triangledown \overline{\mbf{U}}_{d-1} \triangledown \overline{\mbf{U}}_d \right) \cdots \right) \right)
:  \overline{\mbf{U}}_i \in \mathcal{N}_{\varepsilon/\zag{(2d-1)\sqrt{r}}}^{\bm{\mathcal{O}}_{n,r}}, \right. \\ 
&\left. \;\;
\overline{\mbf{B}}_{\{1,\hdots,d\}} \in \mathcal{N}_{\varepsilon/\zag{(2d-1)\sqrt{r}}}^{\bm{\mathcal{F}}_{r,r}},
\overline{\mbf{B}}_{\{j, j+1,\ldots, d\}} \in \mathcal{N}_{\varepsilon/\zag{(2d-1)\sqrt{r}}}^{\bm{\mathcal{O}}_{r,r,r}^{\text{right}}}, i \in [d-1], j = 2, \hdots, d-1\right\}.
\end{align*}
In order to show that $\mathcal{N}_{\varepsilon}^{\bm{\mathcal{S}}_{\mbf{r}}^{\text{TT}}}$ forms an $\varepsilon$-net of
$\boldsymbol{\mathcal{S}}_{\mathbf{r}}^{\operatorname{TT}}$ we choose an arbitrary $\mathbf{X} \in \boldsymbol{\mathcal{S}}_{\mathbf{r}}^{\operatorname{TT}}$ with right-orthogonal decomposition of the form \eqref{HT-decomp-simple} and
for each $\mbf{U}_i$ and $\mbf{B}_{\{j,\hdots,d\}}$ the closest corresponding points $\overline{\mbf{U}}_i \in  \mathcal{N}_{\varepsilon/\zag{(2d-1)\sqrt{r}}}^{\bm{\mathcal{O}}_{n,r}}$, $\overline{\mbf{B}}_{\{1,\hdots,d\}} \in \mathcal{N}_{\varepsilon/\zag{(2d-1)\sqrt{r}}}^{\bm{\mathcal{F}}_{r,r}}$, $\overline{\mbf{B}}_{\{j, j+1,\ldots, d\}} \in \mathcal{N}_{\varepsilon/\zag{(2d-1)\sqrt{r}}}^{\bm{\mathcal{O}}_{r,r,r}^{\text{right}}}$, $j=2,\hdots,d-1$ resulting in $\mbf{X} \in \mathcal{N}_{\varepsilon}^{\bm{\mathcal{S}}_{\mbf{r}}^{\text{TT}}}$.
The triangle inequality yields
\begin{align}
\norma{\mbf{X}-\overline{\mbf{X}}}_F  \leq &\norma{\mbf{B}_{\{1,2,\ldots, d\}}\triangledown \mbf{U}_1\triangledown \left(
\mbf{B}_{\{2,\hdots,d\}} \triangledown \cdots
\left(\mbf{B}_{\{d-1,d\}} \triangledown \mbf{U}_{d-1} \triangledown \zag{\mbf{U}_d-\overline{\mbf{U}}_d} \right) \cdots \right)}_F \nonumber 
\\ &+ \norma{\mbf{B}_{\{1,2,\ldots, d\}} \triangledown\mbf{U}_1 \triangledown \left( \mbf{B}_{\{2,\ldots,d\}} \triangledown \cdots \triangledown\left( \mbf{B}_{\{d-1,d\}} \triangledown \left( \mbf{U}_{d-1} - \overline{\mbf{U}}_{d-1}\right)  \triangledown
\overline{\mbf{U}}_d \right) \cdots \right)}_F \nonumber\\
& + \cdots +\norma{\zag{ \mbf{B}_{\{1,2,\ldots, d\}}-\overline{\mbf{B}}_{\{1,2,\ldots, d\}}} \triangledown \overline{\mbf{U}}_1  \triangledown \left( \overline{\mbf{B}}_{\{2,\hdots,d\}} \triangledown  \cdots \triangledown \zag{ \overline{\mbf{B}}_{\{d-1,d\}} \triangledown \overline{\mbf{U}}_{d-1} \triangledown \overline{\mbf{U}}_d} \cdots\right)}_F. \label{proofTTdiff}
\end{align}
We need to bound terms of the form
\begin{align}
&\norma{\mbf{B}_{\{1,2,\ldots, d\}}\triangledown \mbf{U}_1 \triangledown  \cdots \triangledown \left(\mbf{B}_{\{q-1,q,\ldots, d\}}\triangledown\zag{\mbf{U}_q-\overline{\mbf{U}}_q} \triangledown \left( \overline{\mbf{B}}_{\{q,q+1,\ldots, d\}} \triangledown\cdots \triangledown\overline{\mbf{U}}_d \right) \cdots \right)}_F,  
\; q \in \uglate {d} \label{coveringTTeq1}\\
\text{and } & \norma{\mbf{B}_{\{1,2,\ldots, d\}}\triangledown \mbf{U}_1 \triangledown  \cdots \triangledown\mbf{U}_{p-1} \triangledown\left( \zag{\mbf{B}_{\{p, p+1,\ldots, d\}}-\overline{\mbf{B}}_{\{p,p+1,\ldots, d\}}}\triangledown\overline{\mbf{U}}_{p} \triangledown \left( \cdots \triangledown\overline{\mbf{U}}_d \right) \cdots \right)}_F, \;
p\in \uglate{d-1}. \label{coveringTTeq2}
\end{align} 
To estimate \eqref{coveringTTeq1}, we use orthogonality of $\mbf{U}_q$, $\overline{\mbf{U}}_q$, $q \in \uglate{d}$, and right-orthogonality of $\mbf{B}_{\{p,p+1\ldots,d\}}$, $\overline{\mbf{B}}_{\{p,p+1,\ldots,d\}}$, $p=2,3,\ldots,d-1 $, to obtain
\begin{align*}
&\norma{\mbf{B}_{\{1,2,\ldots, d\}} \triangledown \mbf{U}_1 \triangledown  \cdots \triangledown \left(\mbf{B}_{\{q-1,q,\ldots, d\}}\triangledown\zag{\mbf{U}_q-\overline{\mbf{U}}_q} \triangledown \left( \overline{\mbf{B}}_{\{q,q+1,\ldots, d\}} \triangledown\cdots \triangledown\overline{\mbf{U}}_d \right) \cdots \right)}_F^2 \\
&= \sum_{i_1,\ldots,i_d} \sum_{\substack{j_1,\ldots,j_d \\ k_1,\ldots, k_d}} \sum_{\substack{j_{23\ldots d}, \\ j_{3\ldots d}, \\ \ldots,j_{d-1,d}}} \sum_{\substack{k_{23\ldots d}, \\ k_{3\ldots d}, \\ \ldots,k_{d-1,d}}}  \mbf{B}_{\{1,2,\ldots, d\}}\zag{1,j_1,j_{23\ldots d}} \mbf{B}_{\{1,2,\ldots, d\}}\zag{1,k_1,k_{23\ldots d}} \mbf{U}_1 \zag{i_1,j_1} \mbf{U}_1 \zag{i_1,k_1} \\
& \cdots \mbf{B}_{\{q-1,q,\ldots, d\}}\zag{j_{q-1,q\ldots d},j_{q-1},j_{q\ldots d}} \mbf{B}_{\{q-1,q,\ldots, d\}}\zag{k_{q-1,q\ldots d},k_{q-1},k_{q\ldots d}} \\
& \cdot \zag{\mbf{U}_q-\overline{\mbf{U}}_q}\zag{i_q,j_q} \zag{\mbf{U}_q-\overline{\mbf{U}}_q}\zag{i_q,k_q} \overline{\mbf{B}}_{\{q,q+1,\ldots, d\}}\zag{j_{q,q+1\ldots d},j_{q},j_{q+1\ldots d}} \\
& \cdot \overline{\mbf{B}}_{\{q, q+1,\ldots, d\}}\zag{k_{q,q+1\ldots d},k_{q},k_{q+1\ldots d}} \cdots \overline{\mbf{U}}_d \zag{i_d,j_d} \overline{\mbf{U}}_d \zag{i_d,k_d} \\
&= \sum_{i_q} \sum_{\substack{j_1,\ldots,j_q \\ k_q}} \sum_{\substack{j_{23\ldots d}, \\ j_{3\ldots d},\ldots, \\ j_{q+1\ldots d}}} \sum_{\substack{k_{23\ldots d}, \\ k_{3\ldots d},\ldots, \\ k_{q \ldots d}}}  \mbf{B}_{\{1,2,\ldots, d\}}\zag{1,j_1,j_{23\ldots d}} \mbf{B}_{\{1,2,\ldots, d\}}\zag{1,j_1,k_{23\ldots d}}  \\
 & \cdots  \mbf{B}_{\{q-1,q,\ldots, d\}}\zag{j_{q-1,q\ldots d},j_{q-1},j_{q\ldots d}} \mbf{B}_{\{q-1,q,\ldots, d\}}\zag{k_{q-1,q\ldots d},j_{q-1},k_{q\ldots d}} \\
& \cdot \zag{\mbf{U}_q-\overline{\mbf{U}}_q}\zag{i_q,j_q} \zag{\mbf{U}_q-\overline{\mbf{U}}_q}\zag{i_q,k_q} \\ & \cdot \overline{\mbf{B}}_{\{q,q+1,\ldots, d\}}\zag{j_{q,q+1\ldots d},j_{q},j_{q+1\ldots d}}  \overline{\mbf{B}}_{\{q, q+1,\ldots, d\}}\zag{k_{q,q+1\ldots d},k_{q},j_{q+1\ldots d}}  \\
& = \interval{\Delta\mbf{U}_q, \Box\overline{\mbf{B}}_{\{q,q+1, \ldots, d\}}} \leq \norma{\Delta\mbf{U}_q}_{2\rightarrow 2} \norma{\Box\overline{\mbf{B}}_{\{q,q+1 ,\ldots, d\}}}_*,
\end{align*}
where 
\begin{align*}
 \Delta\mbf{U}_q\zag{j_q,k_q}=&\sum_{i_q} \zag{\mbf{U}_q-\overline{\mbf{U}}_q}\zag{i_q,j_q} \zag{\mbf{U}_q-\overline{\mbf{U}}_q}\zag{i_q,k_q}, \\
 \Box\overline{\mbf{B}}_{\{q,q+1, \ldots, d\}} \zag{j_q,k_q}=& \sum_{j_1,\ldots,j_{q-1}} \sum_{\substack{j_{23\ldots d}, \\ j_{3\ldots d}, \\ \ldots,j_{q+1,\ldots,d}}} \sum_{\substack{k_{23\ldots d}, \\ k_{3\ldots d}, \\ \ldots,k_{q,\ldots,d}}}  \mbf{B}_{\{1,2,\ldots, d\}}\zag{1,j_1,j_{23\ldots d}} \mbf{B}_{\{1,2,\ldots, d\}}\zag{1,j_1,k_{23\ldots d}}  \\
 & \cdots  \mbf{B}_{\{q-1,q,\ldots, d\}}\zag{j_{q-1,q\ldots d},j_{q-1},j_{q\ldots d}} \mbf{B}_{\{q-1,q,\ldots, d\}}\zag{k_{q-1,q\ldots d},j_{q-1},k_{q\ldots d}} \\
& \cdot \overline{\mbf{B}}_{\{q,q+1,\ldots, d\}}\zag{j_{q,q+1\ldots d},j_{q},j_{q+1\ldots d}}  \overline{\mbf{B}}_{\{q, q+1,\ldots, d\}}\zag{k_{q,q+1\ldots d},k_{q},j_{q+1\ldots d}} . 
\end{align*}
We have
\begin{equation}\label{DeltaUq}
  \norma{\Delta\mbf{U}_q}_{2\rightarrow 2} = \| \mbf{U}_q-\overline{\mbf{U}}_q \|_{2 \rightarrow 2}^2 
\leq \norma{\mbf{U}_q-\overline{\mbf{U}}_q}_F^2 \leq r \norma{\mbf{U}_q-\overline{\mbf{U}}_q}_{1,2}^2.
  \end{equation}
Since $\Box\overline{\mbf{B}}_{\{q,q+1,\ldots, d\}}$ is symmetric and positive semidefinite
$$1=\norma{\overline{\mbf{X}}}_F^2=\interval{\mbf{I},\Box\overline{\mbf{B}}_{\{q,q+1 ,\ldots, d\}}}=\tr\zag{\Box\overline{\mbf{B}}_{\{q,q+1,\ldots, d\}}}=\norma{\Box\overline{\mbf{B}}_{\{q,q+1, \ldots, d\}}}_*.$$
Hence, 
\begin{align*}
\norma{\mbf{B}_{\{1,2,\ldots, d\}}\triangledown \mbf{U}_1 \triangledown  \cdots \triangledown \left(\mbf{B}_{\{q-1,q,\ldots, d\}}\triangledown\zag{\mbf{U}_q-\overline{\mbf{U}}_q} \triangledown \left( \overline{\mbf{B}}_{\{q,q+1,\ldots, d\}} \triangledown\cdots \triangledown\overline{\mbf{U}}_d \right) \cdots \right)}_F
& \leq \sqrt{r}  \norma{\mbf{U}_q-\overline{\mbf{U}}_q}_{1,2}\\
& \leq \frac{\varepsilon}{2d-1}.
\end{align*}
In a similar way, distinguishing the cases $p=1$ and $p=2,\hdots,d-1$, we estimate terms of the form \eqref{coveringTTeq2}
as
\begin{align*}
\norma{\mbf{B}_{\{1,2,\ldots, d\}}\triangledown \mbf{U}_1 \triangledown  \cdots \triangledown\mbf{U}_p \triangledown\left( \zag{\mbf{B}_{\{p, p+1,\ldots, d\}}-\overline{\mbf{B}}_{\{p,p+1,\ldots, d\}}}\triangledown\overline{\mbf{U}}_{p+1} \triangledown \left( \cdots \triangledown\overline{\mbf{U}}_d \right) \cdots \right)}_F
\leq \frac{\varepsilon}{2d-1}, \quad q \in \uglate{d-1}.
\end{align*}
Plugging the bounds into \eqref{proofTTdiff} completes the proof for the TT decomposition.
\end{proof}

The proof of Theorem~\ref{glavni} also requires a recent deviation bound \cite{krmera14,Dirksen13} 
for random variables of the form $X=\sup_{\mbf{B} \in \boldsymbol{\mathcal{B}}} \aps{\norma{\mbf{B}\boldsymbol{\xi}}_2^2-\mathbb{E}\norma{\mbf{B}\boldsymbol{\xi}}_2^2}$ in terms of a complexity parameter of the set of 
matrices $\boldsymbol{\mathcal{B}}$ involving covering numbers. 
In order to state it, we introduce  the radii of a set of matrices $\boldsymbol{\mathcal{B}}$ in the Frobenius norm, the operator norm, and the Schatten-4 norm as  
\begin{equation*}
d_F\zag{\boldsymbol{\mathcal{B}}}:=\sup_{\mbf{B}\in \boldsymbol{\mathcal{B}}}\norma{\mathbf{B}}_F, \,  d_{2\rightarrow 2}\zag{\boldsymbol{\mathcal{B}}}:=\sup_{\mbf{B}\in \boldsymbol{\mathcal{B}}}\norma{\mathbf{B}}_{2\rightarrow 2}, \,
d_4\zag{\boldsymbol{\mathcal{B}}}:=\sup_{\mbf{B}\in \boldsymbol{\mathcal{B}}} \norma{\mbf{B}}_{S_4}=\sup_{\mbf{B}\in \boldsymbol{\mathcal{B}}} \zag{\tr\zag{\mbf{B}^T\mbf{B}}^2}^{1/4}.
\end{equation*}
The complexity parameter is Talagrand's $\gamma_2$-functional $\gamma_2\zag{\boldsymbol{\mathcal{B}}, \norma{\cdot}_{2 \rightarrow 2}}$. 
We do not give the precise definition here, but refer to \cite{ta14-1} for details.
For us, it is only important that it can be bounded in terms of covering numbers via a Dudley type integral \cite{Dudley1967290,  ta14-1} as
\begin{equation}\label{gamma2}
\gamma_2\zag{\boldsymbol{\mathcal{B}},\norma{\cdot}_{2 \rightarrow 2}} \leq C \int_{0}^{d_{2\to 2}\zag{\boldsymbol{\mathcal{B}}}} \sqrt{\log\mathcal{N}\zag{\boldsymbol{\mathcal{B}}, \norma{\cdot}_{2 \rightarrow 2}, u}}du.
\end{equation}
We will use the following result from \cite[Theorem 6.5]{Dirksen13} which is a slightly refined version of the main result of \cite{krmera14}.
\begin{theorem} \label{chaos}
Let $\boldsymbol{\mathcal{B}}$ be a set of matrices, and let $\boldsymbol{\xi}$ be a random vector whose entries $\xi_j$ are independent, mean-zero, variance $1$ and $L$-subgaussian random variables. Set
\begin{align*}
E &= \gamma_2\zag{\boldsymbol{\mathcal{B}}, \norma{\cdot}_{2 \rightarrow 2}} \zag{\gamma_2\zag{\boldsymbol{\mathcal{B}}, \norma{\cdot}_{2 \rightarrow 2}}+d_F\zag{\boldsymbol{\mathcal{B}}}}+d_F\zag{\boldsymbol{\mathcal{B}}}d_{{2 \rightarrow 2}}\zag{\boldsymbol{\mathcal{B}}} \\
V &= d_{4 }^2\zag{\boldsymbol{\mathcal{B}}}, 
\text{ and } U=d_{{2 \rightarrow 2}}^2\zag{\boldsymbol{\mathcal{B}}}.
\end{align*}
Then, for $t>0$,
\begin{equation*}
\vjer{\sup_{\mathbf{B} \in \boldsymbol{\mathcal{B}}} \aps{\norma{\mathbf{B}\boldsymbol{\xi}}_2^2- \E \norma{\mathbf{B}\boldsymbol{\xi}}_2^2} \geq c_1E+t} \leq 2 \exp \zag{-c_2\min\skup{\frac{t^2}{V^2}, \frac{t}{U}}}.
\end{equation*}
The constants $c_1, c_2$ only depend on $L$.
\end{theorem}

\begin{proof}[Proof of Theorem~\ref{glavni}]
We write 
$$
\mathcal{A}\zag{\mathbf{X}}=\mathbf{V}_{\mathbf{X}}\boldsymbol{\xi},
$$
where $\boldsymbol{\xi}$ is an $L$-subgaussian random vector of length $n_1 n_2\cdots n_d m$ and 
$\mathbf{V}_{\mathbf{X}}$ is the $m \times  n_1n_2\cdots n_d m$ block-diagonal matrix 
$$ \mathbf{V}_{\mathbf{X}}=\frac{1}{\sqrt{m}}
\begin{bmatrix}
\mathbf{x}^T & \mbf{0} & \cdots & \mbf{0} \\
\mbf{0} & \mathbf{x}^T & \cdots & \mbf{0} \\
\vdots & \vdots & \ddots & \vdots \\
\mbf{0} &\mbf{0} & \cdots  & \mathbf{x}^T
\end{bmatrix},
$$
with $\mathbf{x}$ being the vectorized version of the tensor $\mathbf{X}$. 
With this notation the restricted isometry constant is given by
\[
\delta_{\mathbf{r}} = \sup_{\mathbf{X} \in \boldsymbol{\mathcal{T}}} \left| \|  \mathbf{V}_{\mathbf{X}} \boldsymbol{\xi}\|_2^2- \E \|\mathbf{V}_{\mathbf{X}} \boldsymbol{\xi}\|_2^2 \right|,
\]
where in the HOSVD case $\boldsymbol{\mathcal{T}}=\boldsymbol{\mathcal{S}}_{\mbf{r}}=\skup{\mbf{X} \in \Rd : \rank_{\text{HOSVD}}\zag{\mbf{X}}\leq\mbf{r}, \norma{\mbf{X}}_F=1}$, 
and in the HT-case (including the TT case) $\bm{\mathcal{T}}=\bm{\mathcal{S}}_{\mbf{r}}^{\text{HT}}=\skup{\mbf{X} \in \Rd: \rank_{\text{HT}}\zag{\mbf{X}} \leq \mbf{r}, \norma{\mbf{X}}_{F}=1} $.
Theorem~\ref{chaos} provides a general probabilistic bound for expressions in the form of the right hand side above in terms of the 
diameters $d_F(\boldsymbol{\mathcal{B}})$, $d_{2\rightarrow 2}(\boldsymbol{\mathcal{B}})$, and $d_{4}(\boldsymbol{\mathcal{B}})$ of the set $\boldsymbol{\mathcal{B}}:=\skup{\mathbf{V}_{\mathbf{X}}:\mathbf{X} \in \boldsymbol{\mathcal{T}}}$, as well as in terms of Talagrand's functional $\gamma_2(\boldsymbol{\mathcal{B}}, \norma{\cdot}_{2 \rightarrow 2})$.
It is straightforward to see that $d_F(\boldsymbol{\mathcal{B}}) = 1$, since $\norma{\mbf{X}}_F=1$, for all $\mbf{X} \in \boldsymbol{\mathcal{T}}$. Furthermore, for all $\mbf{X} \in \boldsymbol{\mathcal{T}}$,
\begin{equation}\label{VxVxT}
m\mbf{V}_{\mbf{X}} \mbf{V}_{\mbf{X}}^T= 
\begin{bmatrix}
\mbf{x}^T\mbf{x} & \mbf{0} & \cdots & \mbf{0} \\
\mbf{0} & \mbf{x}^T\mbf{x} & \cdots & \mbf{0} \\
\vdots & \vdots & \ddots & \vdots \\
\mbf{0} & \mbf{0} & \cdots & \mbf{x}^T \mbf{x} \\
 \end{bmatrix}=
\begin{bmatrix}
\norma{\mbf{x}}_2^2 & \mbf{0} & \cdots & \mbf{0} \\
\mbf{0} & \norma{\mbf{x}}_2^2 & \cdots & \mbf{0} \\
\vdots & \vdots & \ddots & \vdots \\
\mbf{0} & \mbf{0} & \cdots & \norma{\mbf{x}}_2^2 \\
 \end{bmatrix}=
\mbf{I}_m,
\end{equation}
so that $\norma{\mbf{V}_{\mbf{X}}}_{2\rightarrow 2}=\frac{1}{\sqrt{m}}$ 
and $d_{2\rightarrow 2}(\boldsymbol{\mathcal{B}})=\frac{1}{\sqrt{m}}$.  
(Since 
the operator norm of a block-diagonal matrix is the maximum of the operator norm of its diagonal blocks
we obtain
\begin{equation}\label{gammaoperator2}
\norma{\mbf{V}_{\mbf{X}}}_{2 \rightarrow 2} =\frac{1}{\sqrt{m}} \norma{\mbf{x}}_2= \frac{1}{\sqrt{m}} \norma{\mbf{X}}_F.)
\end{equation}
From the cyclicity of the trace and \eqref{VxVxT} it follows that
\begin{equation}
\norma{\mbf{V}_{\mbf{X}}}_{S_4}^4=\tr\left[\zag{\mbf{V}_{\mbf{X}}^T\mbf{V}_{\mbf{X}}}^2\right]=
\tr\left[\zag{\mbf{V}_{\mbf{X}}\mbf{V}_{\mbf{X}}^T}^2\right]=
\tr\left[\zag{\frac{1}{m}\mbf{I}_m}^2\right]=\tr\zag{\frac{1}{m^2}\mbf{I}_m}=\frac{1}{m},
\end{equation}
for all $\mbf{V}_{\mbf{X}} \in \bm{\mathcal{B}}$. Thus, $d_4^2\zag{\bm{\mathcal{B}}}=\sup_{\mbf{V}_{\mbf{X}} \in \bm{\mathcal{B}}} \norma{\mbf{V}_{\mbf{X}}}_{S_4}^2=\frac{1}{\sqrt{m}}$.
Using observation \eqref{gammaoperator2}, the bound of the $\gamma_2$-functional via the Dudley type integral \eqref{gamma2} 
yields
\begin{equation} \label{gammaboundweuse}
\gamma_2\zag{\boldsymbol{\mathcal{B}}, \norma{\cdot}_{2 \rightarrow 2}} \leq  C \frac{1}{\sqrt{m}} \int_{0}^{1} \sqrt{\log\zag{\mathcal{N}\zag{\boldsymbol{\mathcal{S}}_{\mathbf{r}}, \norma{\cdot}_F,u}}}\, du,
\end{equation}
where $\boldsymbol{\mathcal{S}}_{\mathbf{r}}$ is replaced by $\boldsymbol{\mathcal{S}}_{\mathbf{r}}^{\text{HT}}$ in the HT case.

Let us first continue with the HOSVD case.
Using the bound \eqref{eqcovering} for $\mathcal{N}\zag{\boldsymbol{\mathcal{S}}_{\mathbf{r}},\norma{\cdot}_F,u}$ and the triangle inequality we reach
\begin{align}
\gamma_2\zag{\boldsymbol{\mathcal{B}}, \norma{\cdot}_{2 \rightarrow 2 }} & \leq C \frac{1}{\sqrt{m}} \int_{0}^1 \sqrt{\zag{r_1r_2\cdots r_d + \sum_{i=1}^d n_ir_i} \log\zag{3\zag{d+1}/u}} \,du \nonumber \\
& = C \sqrt{\frac{r_1r_2 \cdots r_d + \sum_{i=1}^dn_ir_i}{m}} \int_0^1 \sqrt{\log\zag{d+1} + \log\zag{3/u}} \, du \nonumber\\
& \leq C \sqrt{\frac{r_1r_2 \cdots r_d + \sum_{i=1}^dn_ir_i}{m}} \zag{ \sqrt{\log\zag{d+1}} + \int_0^1\sqrt{\log\zag{3/u}} \, du }\nonumber\\
&  \leq \tilde{C} \sqrt{\frac{\zag{r_1r_2\cdots r_d+ \sum_{i=1}^{d}n_ir_i}\log \zag{d}}{m}} \leq \tilde{C} \sqrt{\frac{\zag{r^d+dnr}\log(d)}{m}},\label{gammaHOSVD}
\end{align}
where $r:=\max\skup{r_i: i \in \uglate{d}}$ and $n:=\max\skup{n_i:i \in \uglate{d}}$.

Let us now consider the HT case (including the TT case).
Using  the bound \eqref{gammaboundweuse} of the $\gamma_2$-functional via  Dudley type integral 
and the covering number bound \eqref{coveringHTmin} for $\mathcal{N}\zag{\boldsymbol{\mathcal{S}}^{\text{HT}}_{\mathbf{r}},\norma{\cdot}_F,u}$, we obtain
\begin{align}
\gamma_2\zag{\bm{\mathcal{B}}, \norma{\cdot}_{2 \rightarrow 2}}
& \leq  C \frac{1}{\sqrt{m}} \int_{0}^{1} \sqrt{\log\zag{\mathcal{N}\zag{\bm{\mathcal{S}}^{\text{HT}}_{\mathbf{r}}, \norma{\cdot}_F,u}}} \,du \notag\\
& \leq C \frac{1}{\sqrt{m}}\sqrt{\sum_{t \in \mathcal{I}\zag{T_I}}r_t r_{t_1} r_{t_2} + \sum_{i=1}^d r_i n_i} \cdot \int_{0}^{1}\sqrt{\log \zag{3(2d-1)\sqrt{r}/u}} \,du.\notag\\
&\leq \tilde{C}_1 \sqrt{\frac{\zag{\sum_{t \in \mathcal{I}(T_I)} r_{t}r_{t_1} r_{t_2} + \sum_{i=1}^d r_i n_i }\log \zag{(2d-1)\sqrt{r}}}{m}} \nonumber\\
 &\leq \tilde{C}_1 \sqrt{\frac{\zag{(d-1)r^3+dnr}\log \zag{(2d-1)\sqrt{r}}}{m}}.\label{gammaTT}
\end{align}
In order to apply Theorem~\ref{chaos} we note that 
\begin{align*}
E& = \gamma_2\zag{\boldsymbol{\mathcal{B}}, \norma{\cdot}_{2 \rightarrow 2}} \zag{\gamma_2\zag{\boldsymbol{\mathcal{B}}, \norma{\cdot}_{2 \rightarrow 2}} + d_F\zag{\boldsymbol{\mathcal{B}}}}+d_F\zag{\boldsymbol{\mathcal{B}}}d_{2 \rightarrow 2}\zag{\boldsymbol{\mathcal{B}}}\\&= \gamma_2^2\zag{\boldsymbol{\mathcal{B}}, \norma{\cdot}_{2 \rightarrow 2}} + \gamma_2\zag{\boldsymbol{\mathcal{B}}, \norma{\cdot}_{2 \rightarrow 2}} + \frac{1}{\sqrt{m}}, \\
V& = d_4^2\zag{\boldsymbol{\mathcal{B}}}=\frac{1}{\sqrt{m}},
\qquad U = d_{2\rightarrow 2}^2\zag{\boldsymbol{\mathcal{B}}}=\frac{1}{m}.
\end{align*}
The bound on $m$ of Theorem~\ref{glavni}  
ensures that $c_1 E \leq \delta/2$ and that 
$2\exp\zag{-c_2 \min\skup{\frac{t^2}{V^2}, \frac{t}{U}}} \leq \varepsilon$ with $t = \delta/2$ (provided constants are chosen appropriately).
Therefore, the claim follow from Theorem~\ref{chaos}.
\end{proof}

\section{Random Fourier measurements}
\label{Sec:Fourier}

While subgaussian measurements often provide benchmark guarantees in compressive sensing and low rank recovery
in terms of the minimal number of required measurements, they lack of any structure and therefore are of limited use in practice.
In particular, no fast multiplication routines are available for them.
In order to overcome such limitations, structured random measurement matrices have been studied in 
compressive sensing \cite{ra10,fora13,krra14,carota06}
and low rank matrix recovery \cite{care09,cata10,forawa11,krra14} and almost optimal recovery guarantees have been shown.

In this section, we extend one particular construction of a randomized Fourier transform from the matrix case \cite[Section 1]{forawa11} 
to the tensor case. The measurement map 
\[
\mathcal{A}:\Cd \rightarrow \C^m, \quad \mathcal{A} = \frac{1}{\sqrt{m}} \mathcal{R}_{\bm{\varOmega}} \mathcal{F}_d \mathcal{D}
\]
is the composition of a random sign flip map $\mathcal{D} : \Cd \to \Cd$ defined componentwise as $\mathcal{D}(\mbf{X})\zag{j_1,\hdots,j_d} = \epsilon_{j_1,\hdots,j_d} \mbf{X}\zag{j_1,\hdots,j_d}$
with the $\epsilon_{j_1,\hdots,j_d}$ being independent $\pm 1$ Rademacher variables, a $d$-dimensional Fourier transform
\[
\mathcal{F}_d : \Cd \to \Cd, \quad \mathcal{F}_d(\mbf{X})\zag{j_1,\hdots,j_d} =  
\sum_{k_1=1}^{n_1} \cdots \sum_{k_d=1}^{n_d}\mbf{X}\zag{k_1,\hdots,k_d} e^{-2 \pi i \sum_{\ell=1}^d \frac{k_\ell j_\ell}{n_\ell}},
\]
and a random subsampling operator $\mathcal{R}_{\bm{\varOmega}} : \Cd \to \C^{\bm{\varOmega}}=\C^m$, $\mathcal{R}_{\bm{\varOmega}}(\mbf{X})_{\mathbf{j}} = \mbf{X}\zag{\mathbf{j}}$ for
$\mathbf{j} \in \bm{\varOmega} \subset [n_1] \times \cdots \times [n_d]$, where $\bm{\varOmega}$ is selected uniformly at random
among all subsets of $[n_1] \times \cdots \times [n_d]$ of cardinality $m$. Instead of the $d$-dimensional Fourier transform, we can also use the $1$-dimensional Fourier transform applied to the vectorized version of a tensor $\mbf{X}$ without changes in the results below. Since the Fourier transform can be applied 
quickly in $\mathcal{O}(n^d \log^d n)$, $n = \max{\{n_\ell : \ell \in [d]\}}$, operations using the FFT, the map $\mathcal{A}$ runs with this computational complexity -- as opposed to the trivial running time of $\mathcal{O}(n^{2d})$ for unstructured measurement maps.
By vectorizing tensors in $\Cd$, the map $\mathcal{A}$ can be written as a partial random Fourier matrices with randomized column signs.

The randomized Fourier map $\mathcal{A}$ satisfies the TRIP for an almost optimal number of measurements as shown
by the next result.
\begin{theorem}\label{FourierTRIP}
Let $\mathcal{A}: \Cd \rightarrow \C^m$ be the randomized Fourier map described above.
Then $\mathcal{A}$ satisfies the TRIP with tensor restricted isometry constant $\delta_{\mbf{r}}$ with probability 
exceeding $1-2e^{-\eta}$ as long as
\begin{equation}\label{m:Fourier}
m \geq C\delta_{\mbf{r}}^{-1}\zag{1+\eta} \log^2(n^d) \max\skup{\delta_{\mbf{r}}^{-1}\zag{1+\eta} \log^2(n^d), f(n,d,r)},
\end{equation}
where 
\begin{align*}
&f(n,d,r)=\zag{r^d+dnr}\log\zag{d} \quad \text{ for the HOSVD case },\\
& f(n,d,r)=\zag{dr^3+dnr}\log\zag{dr} \quad \text{ for  the TT and HT case},
\end{align*}
 $n=\max \skup{n_i: i \in \uglate{d}}$ and $r=\max\skup{ r_t: t \in T_I}$. 
\end{theorem}

To prove Theorem~\ref{FourierTRIP} we use a special case of Theorem~3.3 in \cite{DBLP:journals/corr/OymakRS15} for 
the partial Fourier matrix with randomized column signs, which generalizes the main result of \cite{krwa11}.
Using that the Gaussian width of a set $T$ is equivalent to $\gamma_2(T,\|\cdot\|_2)$ by 
Talagrand's majorizing theorem \cite{talagrand2001, talagrand2005}, this result reads in our notation as follows.

\begin{theorem}\label{THFourier}
Let $\bm{\mathcal{T}} \subset \Cd$ and let $\mathcal{A} : \Cd \to \C^m$ be the randomized Fourier map as described above.
Then for $0<\delta<1$
$$ \sup_{\mbf{X} \in \mathcal{T}} \aps{\norma{\mathcal{A}(\mbf{X})}_2^2-\norma{\mbf{X}}_2^2} \leq \delta \cdot \zag{d_F\zag{\bm{\mathcal{T}}}}^2,$$
holds with probability at least $1-2e^{-\eta}$ as long as
\begin{equation}\label{m:Fourier:abstract}
m \geq C \delta^{-2} \zag{1+\eta}^2\zag{\log(n_1\cdots n_d)}^4 \max\skup{1,\frac{\gamma_2\zag{\bm{\mathcal{T}}, \norma{\cdot}_F}}{\zag{d_F\zag{\bm{\mathcal{T}}}}^2}}.
\end{equation}
\end{theorem}

\begin{proof}[Proof of Theorem~\ref{FourierTRIP}]
We use $\bm{\mathcal{T}} = \bm{\mathcal{S}}_{\mbf{r}}$ and $\bm{\mathcal{T}} = \bm{\mathcal{S}}_{\mbf{r}}^{\text{HT}}$ and
recall that $d_F(\bm{\mathcal{T}}) = 1$. Moreover, $\gamma_2(\bm{\mathcal{T}}, \|\cdot\|_F)$ has been estimated in \eqref{gammaHOSVD} and
\eqref{gammaTT}. By distinguishing cases, one then verifies that \eqref{m:Fourier} implies \eqref{m:Fourier:abstract} so that
Theorem~\ref{THFourier} implies Theorem~\ref{FourierTRIP}.
\end{proof}

Using recent improved estimates for the standard RIP for random partial Fourier matrices \cite{bo14,DBLP:conf/soda/HavivR16} 
in connection with techniques from
\cite{DBLP:journals/corr/OymakRS15} it may be possible to improve Theorem~\ref{THFourier} and thereby \eqref{m:Fourier} 
in terms of logarithmic factors.

\section{Numerical results}\label{sec:Numerics}
\label{Sec:Numerics}

We present numerical results for recovery of third order tensors $\mathbf{X} \in \R^{n_1 \times n_2 \times n_3}$ 
and the HOSVD format which illustrate
that tensor iterative hard thresholding works very well despite the fact that we only have a partial recovery result.
We ran experiments for both versions (CTIHT and NTIHT) of the algorithm and for Gaussian random measurement maps,
randomized Fourier measurement maps (where $\mbf{X} \in \C^{n_1 \times n_2 \times n_3}$), and tensor completion, i.e., recovery from randomly chosen entries of the tensor.
(No theoretical investigations are yet available for the latter scenario).

For other related numerical results, we refer to papers \cite{demoraisgoulart:hal-01132367, GengAcc}, where they have considered a slightly different versions of the tensor iterative hard thresholding algorithm and compared it with NTIHT.

We consider recovery of a cubic tensor, i.e., $n_1=n_2=n_3=10$, with equal and unequal ranks of its unfoldings, respectively, (first and second experiment) and of a non-cubic tensor $\mathbf{X} \in \R^{6 \times 10 \times 15}$ with equal ranks of the unfoldings, i.e., 
$r_1=r_2=r_3=r$ (third experiment). For fixed tensor dimensions $n_1 \times n_2 \times n_3$, 
fixed HOSVD-rank $\mathbf{r}=\zag{r_1,r_2,r_3}$ 
and a fixed number of measurements $m$ we performed $200$ simulations. 
We say that an algorithm successfully recovers the original tensor $\mathbf{X}_0$ if the reconstruction 
$\mathbf{X}^{\#}$ satisfies $\norma{\mathbf{X}_0-\mathbf{X}^{\#}}_F<10^{-3}$ for Gaussian measurement maps and Fourier measurement ensembles, and  
$\mathbf{X}^{\#}$ such that $\norma{\mathbf{X}_0-\mathbf{X}^{\#}}_F<2.5\cdot10^{-3}$ for tensor completion.
The algorithm stops in iteration $j$ if $\norma{\mathbf{X}^{j+1}-\mathbf{X}^{j}}_F<10^{-4}$ in which case we say that the algorithm converged, or it stops if it reached $5000$ iterations.

A Gaussian linear mapping $\mathcal{A}:\R^{n_1 \times n_2 \times n_3} \rightarrow \R^m$ is defined by tensors $\mathbf{A}_k \in \R^{n_1 \times n_2 \times n_3}$ via $\uglate{\mathcal{A}\zag{\mathbf{X}}}\zag{k}=\interval{\mathbf{X},\mathbf{A}_k}$, for all $k \in \uglate{m}$, where the entries of the tensors $\mathbf{A}_k$ are i.i.d.\ Gaussian $\mathcal{N}\zag{0,\frac{1}{m}}$.
The tensor $\mathbf{X}^0 \in \R^{n_1 \times n_2 \times n_3}$ of rank $\mathbf{r}=\zag{r_1,r_2,r_3}$ is generated via its Tucker decomposition $\mathbf{X}^0=\mathbf{S} \times_1 \mathbf{U}_1 \times_2 \mathbf{U}_2 \times_3 \mathbf{U}_3$: Each of the elements of the tensor $\mathbf{S}$ is taken independently from the normal distribution $\mathcal{N}\zag{0,1}$, and the components $\mathbf{U}_k \in \R^{n_k \times r_k}$ are the first $r_k$ left singular vectors of a matrix $\mathbf{M}_k \in \R^{n_k \times n_k}$ whose elements are also drawn independently from the normal distribution $\mathcal{N}\zag{0,1}$.

We have used the toolbox TensorLab \cite{TensorLab} 
for computing the HOSVD decomposition of a given tensor and the truncation operator $\mathcal{H}_r$. 
By exploiting the Fast Fourier Transform (FFT), the measurement operator $\mathcal{A}$ from Section~\ref{Sec:Fourier} 
related to the Fourier transform and its adjoint $\mathcal{A}^*$ can be applied efficiently which leads to reasonable
run-times for comparably large tensor dimensions, see Table~\ref{TableRuntimes}. 

The numerical results for low rank tensor recovery obtained via the NTIHT algorithm for Gaussian measurement maps are presented in Figures~\ref{figsim},~\ref{figsim2}, and~\ref{figsim3}.
In Figure~\ref{figsim} and~\ref{figsim2} we present the recovery results for low rank tensors of size $10 \times 10 \times 10$. The horizontal axis represents the number of measurements taken with respect to the number of degrees of freedom of an arbitrary tensor of this size. To be more precise, for a tensor of size $n_1 \times n_2 \times n_3$, the number $\bar{n}$ on the horizontal axis represents $m=\left\lceil n_1n_2n_2 \frac{\bar{n}}{100} \right\rceil$ measurements. The vertical axis represents the percentage of successful recovery.

Finally, in Table~\ref{TableRes} we present numerical results for third order tensor recovery via the CTIHT and the NTIHT algorithm. We consider Gaussian measurement maps, Fourier measurement ensembles, and tensor completion. With $m_0$ we denote the minimal number of measurements that are necessary to get full recovery and with $m_1$ we denote the maximal number of measurements for which we do not manage to recover any out of $200$ tensors.

\begin{figure}[!t]
\centering
\includegraphics[scale=0.6]{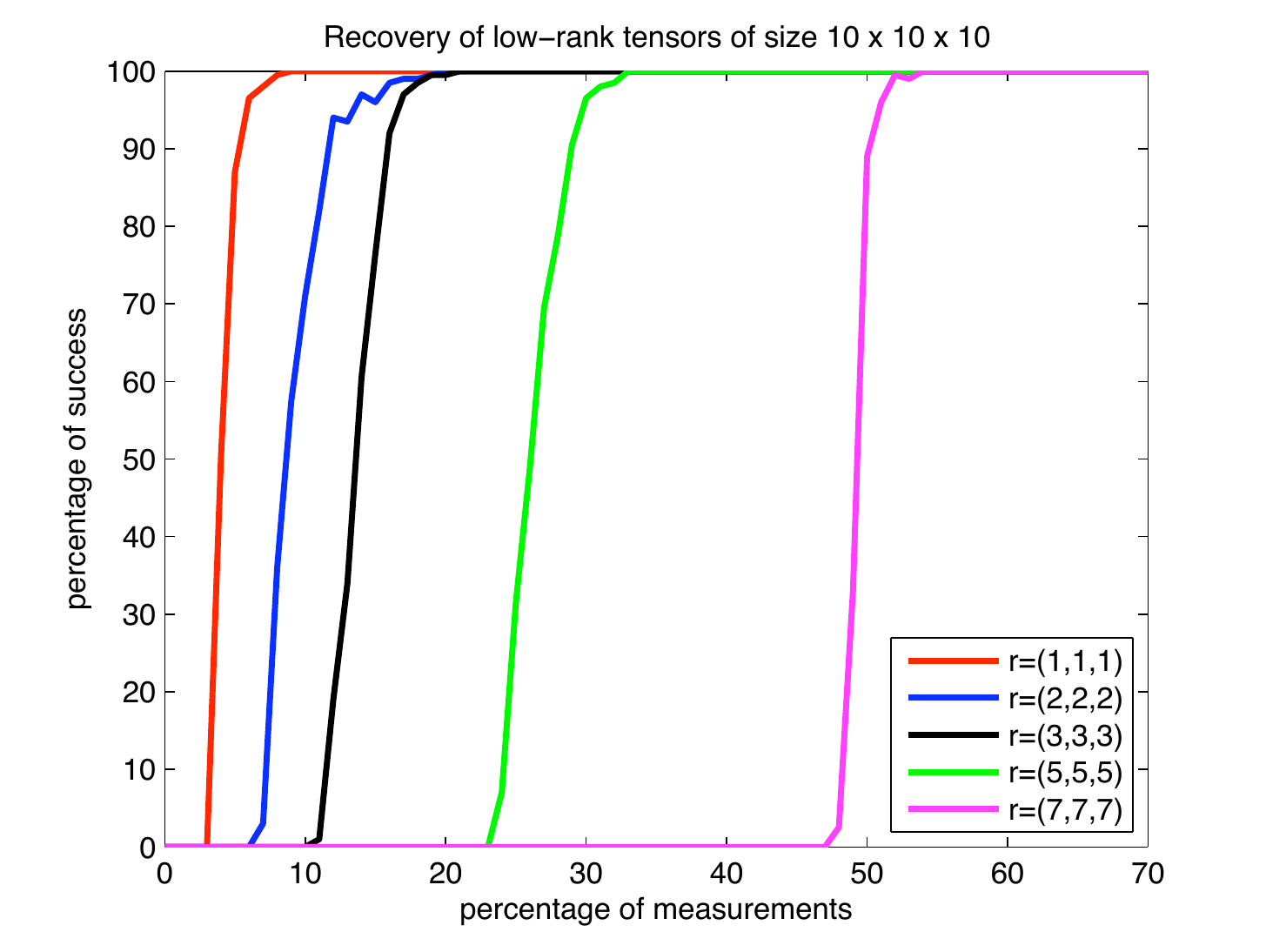}
\caption{Recovery of low rank 10 x 10 x 10 tensors of the same rank via NTIHT}
\label{figsim}
\end{figure}

\begin{figure}[!t]
\centering
\includegraphics[scale=0.6]{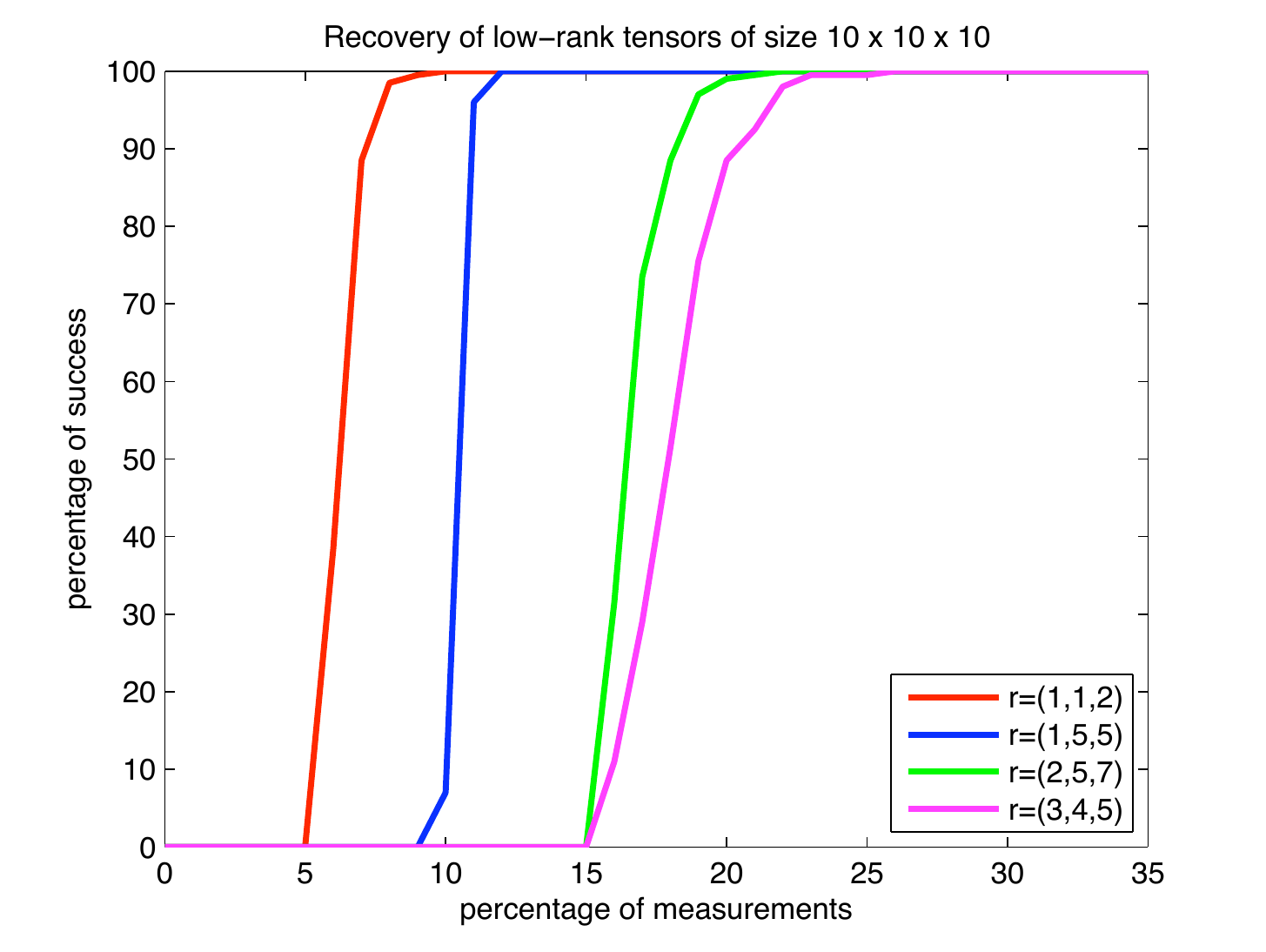}
\caption{Recovery of low rank $10 \times 10 \times 10$ tensors of a different rank via NTIHT}
\label{figsim2}
\end{figure}

\begin{figure}[!t]
\centering
\includegraphics[scale=0.6]{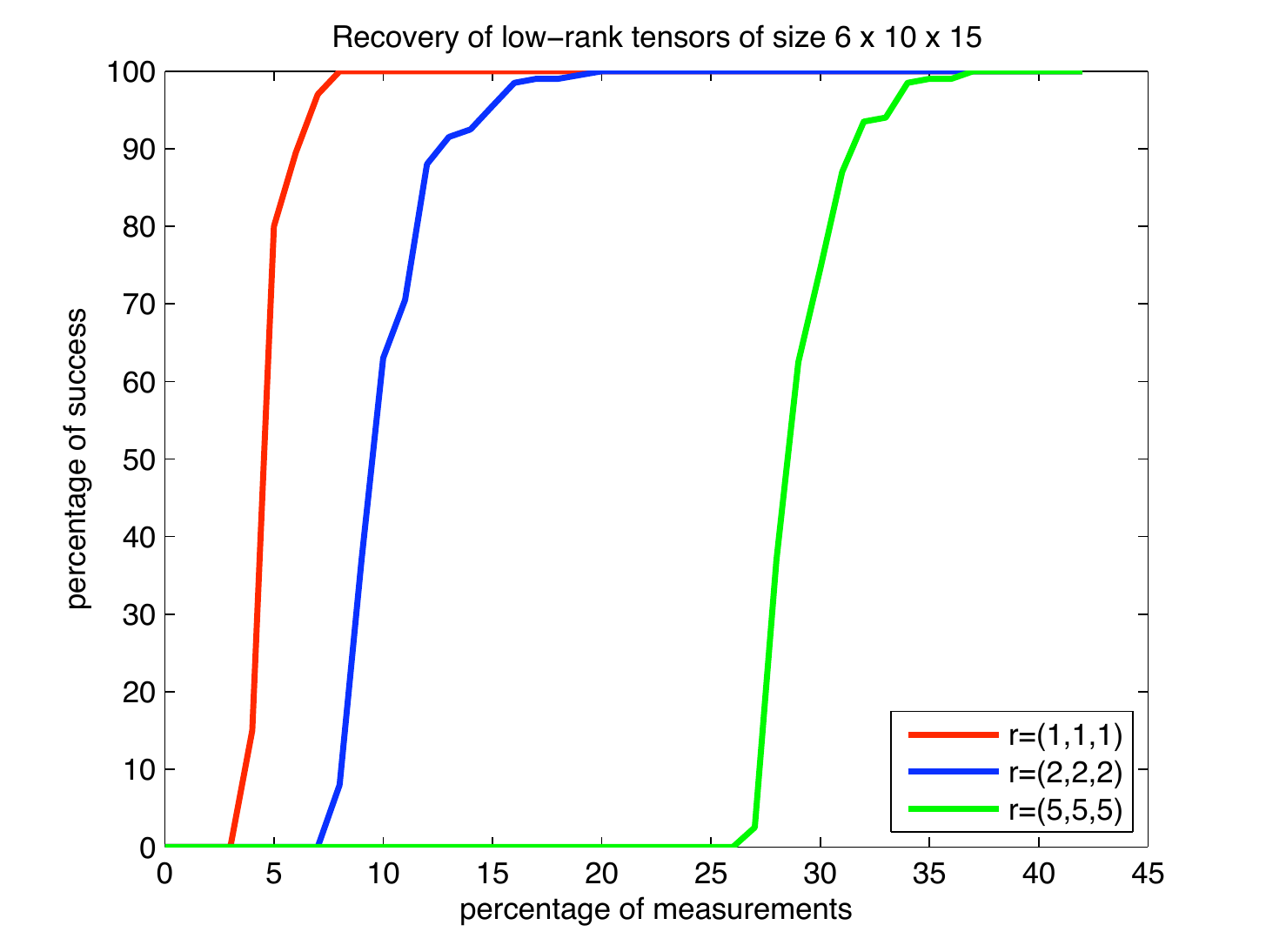}
\caption{Recovery of low rank $6 \times 10 \times 15$ tensors of a different rank via NTIHT}
\label{figsim3}
\end{figure}

\begin{table}\label{TableRes}
\begin{tabular}{ c c c c c c c}
\hline 
type & tensor dimensions & rank & NTIHT-$\overline{n}_{0}$ & NTIHT-$\overline{n}_1$ & CTIHT-$\overline{n}_0$ & CTIHT-$\overline{n}_1$ \\  \hline 
Gaussian & $10 \times 10 \times 10$ & $\zag{1,1,1}$ & $8$  &  $3$ &  $24$ &  $6$ \\
&$10 \times 10 \times 10$ & $\zag{2,2,2}$ & $20$ &  $6$ & $39$ &  $21$ \\
&$10 \times 10 \times 10$ & $\zag{3,3,3}$ & $21$ & $11$ & $60$ & $40$ \\
&$10 \times 10 \times 10$ & $\zag{5,5,5}$ & $33$ & $23$ & $-$ & $-$ \\
&$10 \times 10 \times 10$ & $\zag{7,7,7}$ & $53$ & $47$ & $-$ & $-$ \\
\hline
Gaussian & $10 \times 10 \times 10$ & $\zag{1,2,2}$ & $10$  &  $5$ &  $34$ &  $16$ \\
&$10 \times 10 \times 10$ & $\zag{1,5,5}$ & $12$ &  $9$ & $57$ &  $37$ \\
&$10 \times 10 \times 10$ & $\zag{2,5,7}$ & $20$ & $15$ & $83$ & $64$ \\
&$10 \times 10 \times 10$ & $\zag{3,4,5}$ & $23$ & $15$ & $83$ & $62$ \\
\hline
Gaussian & $6 \times 10 \times 15$ & $\zag{1,1,1}$ & $9$  &  $3$ &  $25$ &  $8$ \\
&$6 \times 10 \times 15$ & $\zag{2,2,2}$ & $20$ &  $7$ & $44$ &  $27$ \\
&$6 \times 10 \times 15$ & $\zag{5,5,5}$ & $34$ & $26$ & $-$ & $-$ \\
\hline

Fourier & $10 \times 10 \times 10$ & $\zag{1,1,1}$ & $16$  &  $3$ & $15$ & $8$ \\
&$10 \times 10 \times 10$ & $\zag{2,2,2}$ & $11$ &  $6$ & $25$ & $16$ \\
&$10 \times 10 \times 10$ & $\zag{3,3,3}$ & $16$ &  $14$ & $31$ & $26$ \\
&$10 \times 10 \times 10$ & $\zag{5,5,5}$ & $29$ &  $26$ & $43$ & $40$ \\
&$10 \times 10 \times 10$ & $\zag{7,7,7}$ & $51$ &  $48$ & $50$ & $49$ \\
\hline
Fourier & $10 \times 10 \times 10$ & $\zag{1,2,2}$ & $10$  &  $5$ &  $21$ &  $14$ \\
&$10 \times 10 \times 10$ & $\zag{1,5,5}$ & $16$ &  $12$ & $31$ & $25$ \\
&$10 \times 10 \times 10$ & $\zag{2,5,7}$ & $21$ & $18$ & $37$ & $33$ \\
&$10 \times 10 \times 10$ & $\zag{3,4,5}$ & $21$ & $18$ & $37$ & $33$ \\
\hline
Fourier & $6 \times 10 \times 15$ & $\zag{1,1,1}$ & $12$  &  $3$ &  $16$ &  $9$ \\
&$6 \times 10 \times 15$ & $\zag{2,2,2}$ & $13$ &  $8$ & $25$ &  $20$ \\
&$6 \times 10 \times 15$ & $\zag{5,5,5}$ & $32$ & $29$ & $45$ & $42$ \\
\hline

completion & $10 \times 10 \times 10$ & $\zag{1,1,1}$ & $17$  &  $2$ &  $27$ &  $2$ \\
&$10 \times 10 \times 10$ & $\zag{2,2,2}$ & $43$ &  $8$ & $45$ &  13 \\
&$10 \times 10 \times 10$ & $\zag{3,3,3}$ & $37$ & $12$ & $32$ &  $16$ \\
&$10 \times 10 \times 10$ & $\zag{5,5,5}$ & $44$ & $24$ & $50$ &  $30$ \\
&$10 \times 10 \times 10$ & $\zag{7,7,7}$ & $71$ & $46$ & $84$ &  $54$ \\
\hline
completion & $10 \times 10 \times 10$ & $\zag{1,2,2}$ & $33$  &  $6$ &  $38$ &  $10$ \\
&$10 \times 10 \times 10$ & $\zag{1,5,5}$ & $57$ & $15$ & $58$ &  $21$ \\
&$10 \times 10 \times 10$ & $\zag{2,5,7}$ & $35$ & $17$ & $47$ &  $24$ \\
&$10 \times 10 \times 10$ & $\zag{3,4,5}$ & $36$ & $17$ & $41$ &  $22$ \\
\hline
completion & $6 \times 10 \times 15$ & $\zag{1,1,1}$ & $20$  &  $3$ &  $33$ &  $8$ \\
&$6 \times 10 \times 15$ & $\zag{2,2,2}$ & $47$ &  $10$ & $51$ &  $14$ \\
&$6 \times 10 \times 15$ & $\zag{5,5,5}$ & $46$ & $27$ & $51$ &  $33$ \\
\hline

\end{tabular}
\caption{Recovery results for low rank matrix recovery via Gaussian measurement maps, Fourier measurement ensembles and tensor completion for NTIHT and CTIHT algorithm. An algorithm successfully recovers the sensed tensor $\mbf{X}_0$ if it returns a tensor $\mbf{X}^{\#}$ such that  $\norma{\mathbf{X}_0-\mathbf{X}^{\#}}_F<10^{-3}$ for Gaussian measurement maps and Fourier measurement ensembles, and  $\mathbf{X}^{\#}$ such that $\norma{\mathbf{X}_0-\mathbf{X}^{\#}}_F<2.5\cdot10^{-3}$ for tensor completion. $\overline{n}_0$: minimal percentage of measurements needed to get hundred percent recovery; $\overline{n}_1$: maximal percentage of measurements for which recover is not successful for all out of $200$ tensors; That is, the number of measurements $m_i=\lceil n_1n_2n_3\frac{\overline{n}_i}{100}\rceil$, for $i=0,1$; $-$ means that we did not manage to recover all $200$ tensors with percentage of measurements less than $\overline{n}=100$;}
\end{table}

\begin{table}\label{TableRuntimes}
\begin{center}
\begin{tabular}{ c c c c c}
\hline 
type & tensor dimensions & rank & CTIHT-$\overline{n}$ &CPU time in sec \\  \hline 

Fourier & $100 \times 100 \times 100$ & $\zag{1,1,1}$ & $10$  &  $16.2709$   \\
 & $100 \times 100 \times 100$ & $\zag{1,1,1}$ & $20$  &  $14.9761$   \\
&$100 \times 100 \times 100$ & $\zag{5,5,5}$ & $10$ &  $31.8866$   \\
&$100 \times 100 \times 100$ & $\zag{5,5,5}$ & $20$ &  $26.3486$   \\
&$100\times 100 \times 100$ & $\zag{7,7,7}$ & $20$ &  $27.2222$  \\
&$100 \times 100 \times 100$ & $\zag{10,10,10}$ & $20$ &  $36.3950$ \\
\hline
Fourier & $200 \times 200 \times 200$ & $\zag{1,1,1}$ & $10$  &  $142.2105$  \\
\hline
\end{tabular}
\end{center}
\caption{Computation times for reconstruction from Fourier type measurements. The numerical experiments are run on a PC with Intel(R) Core(TM) i7-2600 CPU @ 3.40 GHz on Windows $7$ Professional Platform (with 64-bit operating system) and $8$ GB RAM; $\overline{n}$ denotes the percentage of measurements, so that the number of measurements $m=\lceil n_1 n_2 n_3 \frac{\overline{n}}{100} \rceil$.}
\end{table}

\section*{References}

\bibliographystyle{abbrv}
\bibliography{TensorBib}

\end{document}